\documentclass[a4paper,11pt]{amsart}
\usepackage[utf8]{inputenc}
\usepackage{xcolor}
\usepackage{graphicx}
\usepackage{dcolumn}
\usepackage{hyperref}
\hypersetup{colorlinks=true, linkcolor=blue, citecolor=blue, urlcolor=blue}
\usepackage[english]{babel}
\usepackage[autostyle, english = american]{csquotes}
\MakeOuterQuote{"}
\usepackage{amsthm}
\usepackage{mathrsfs}
\usepackage{amssymb}
\usepackage{bm}
\usepackage{physics}

\usepackage[p,osf]{cochineal}
\usepackage[cochineal]{newtxmath}
\usepackage{eucal}
\usepackage{eufrak}

\usepackage[backend=bibtex,style=numeric,sorting=nyt,maxcitenames=5,mincitenames=3,maxbibnames=10,giveninits=true, doi = false, isbn = false, url=false]{biblatex}

\theoremstyle{plain}

\newtheorem{theorem}{Theorem}[section]
\newtheorem{proposition}[theorem]{Proposition}
\newtheorem{lemma}[theorem]{Lemma}
\newtheorem{corollary}[theorem]{Corollary}

\theoremstyle{definition}
\newtheorem{definition}[theorem]{Definition}
\newtheorem*{remark}{Remark}

\DeclareMathOperator\supp{supp}

\DeclareMathOperator\Hess{Hess}
\DeclareMathOperator\Ran{Ran}
\DeclareMathOperator\Ent{Ent}

\newcommand{\vN}{\mathfrak W}
\newcommand{\vNb}{\mathfrak W_\beta}
\newcommand{\AvN}{\mathfrak U}
\newcommand{\gvN}{\mathfrak A}
\newcommand{\E}{\mathbb E}
\newcommand{\Cauchy}{\mathcal C}

\newcommand{\Mb}{\mathbb M_\beta}

\newcommand{\h}{\mathbf{h}}
\newcommand{\K}{\mathbf{K}}

\newcommand{\Q}{\mathcal Q}

\newcommand{\mnorm}[1]{\left \Vert #1 \right \Vert}

\newcommand{\Hi}{\mathcal H}
\newcommand{\p}{p}

\newcommand{\Ba}{\mathcal B}

\newcommand{\CCR}{\mathcal W(\mathfrak h)}
\newcommand{\Emery}{Émery }
\newcommand{\Cet}{C_{\epsilon,t}}
\newcommand{\Dcet}{\dot C_{\epsilon,t}}
\newcommand{\Vn}{V^{(n)}}
\newcommand{\FVn}{\hat{V}^{(n)}}
\newcommand{\x}{\mathbf{x}}

\begin{document}
\title{Renormalization group bounds on Araki relative entropy}
\author{Edoardo D'Angelo}
\address{Leibniz University Hannover, Institute of Analysis, 30167 Hannover, Germany}
\email{dangelo@math.uni-hannover.de}


\date{}

\begin{abstract}
This paper derives a multiscale Bakry-\'Emery criterion for an upper bound on the Araki relative entropy between an equilibrium state in relativistic QFT and a class of affiliated perturbations. Using stochastic positivity, affiliated perturbations of the relativistic bosonic field can be put in correspondence with perturbations of the massive Euclidean free field. This gives a Feynman-Kac formula, which in turn is used to derive a probabilistic expression for the Araki relative entropy. The description of the massive Euclidean free field as an abstract Wiener space provides the infinite-dimensional framework to derive a Polchinski equation. Extending the methods of Bauerschmidt, Bodineau, and Dagallier to this context yields a probabilistic bound on the Araki relative entropy, which can be formulated entirely in operator-algebraic terms. Using the criterion, it is shown that the relative entropy bound holds for a Sine-Gordon perturbation of the relative bosonic field in Minkowski spacetime, in every local region and uniformly in the ultraviolet regularisation.
\end{abstract}

\maketitle

\section{Introduction}
Probabilistic methods have been at the heart of constructive Quantum Field Theory (QFT) since its very inception. After the seminal works of Schwinger~\cite{Schwinger1958} and Symanzik~\cite{Symanzik1966}, in 1973 Nelson realised the free relativistic bosonic field in terms of a stochastic Markov process evolving in a fictitious Euclidean time~\cite{Nelson1973,Nelson1973a}. 

Nelson's approach brought techniques from probability and statistical mechanics to the construction of relativistic field theories~\cite{GlimmJaffe1987}. Conversely, ideas from QFT crossed into probability theory, leading to new and important results. One prominent example is the use of the \emph{Logarithmic Sobolev Inequality} (LSI), introduced by Gross~\cite{Gross1975} as a refinement of Nelson's \emph{hypercontractivity}~\cite{Nelson1966}.

On a probability space $(\Q, \mu)$, in terms of the relative entropy of a smooth function $F:\Q \to \mathbb R_+$ with respect to $\mu$,
\[
\mathrm{Ent}_\mu(F) \equiv \E_\mu(\Phi(F)) - \Phi(\E_\mu(F)) \ , \quad \Phi(x) \equiv x \log x \ , \quad \E_\mu(F) \equiv \int_\Q F \dd \mu \ ,
\]
the LSI takes the expression
\[
\mathrm{Ent}_\mu(F^2)  \leq \frac{1}{2\gamma} \E_\mu(\abs{\nabla F}^2) \ ,
\]
for some constant $\gamma > 0$. 

Since their introduction, LSIs have become a central tool in probability theory with many important consequences, such as concentration of measure and transport inequalities~\cite{Bakry2013,Ledoux2004}.

More recently, Bauerschmidt, Bodineau and Dagallier~\cite{BauerschmidtBodineau2020,Bauerschmidt2024,Bauerschmidt2023b,Bauerschmidt2023} initiated a programme to derive LSIs from a continuum decomposition of the interacting measure in scales along the Renormalisation Group (RG) flow. Using as central tool the \emph{Polchinski equation}, they derived a multiscale version of the celebrated \emph{Bakry-\Emery criterion}, allowing them to derive LSIs for non-convex interactions.

In this paper, we close the circle by using the multiscale Bakry-\Emery criterion to derive a bound for Araki relative entropy in relativistic QFT.

In the non-commutative setting of von Neumann algebras, the entropy functional $\mathrm{Ent}_\mu$ is generalised by the Araki formula~\cite{Araki1975}. Consider a von Neumann algebra $\vN$ in standard form $(\vN, J, \mathcal H, \mathcal P)$, where $\mathcal H$ is the separable Hilbert space on which $\vN$ acts, $J$ is the antilinear, isometric involution of $\mathcal H$ mapping $\vN$ to its commutant, and $\mathcal P$ is a self-dual cone, invariant under $J$. 

The relative entropy between two normal states $\omega, \psi$ on $\vN$, with unique vector representatives $\Omega, \ \Psi \in \mathcal P$, is
\[
S(\omega| \psi) = - \left (\Omega, \log\Delta_{\Omega,\Psi} \Omega \right ) \ ,
\]
where $\Delta_{\Omega,\Psi}$ is the relative modular operator  (see Def.~\ref{def:relative-entropy} for more details). 

Thanks to its generality, the Araki formula can be applied also to the type III algebras appearing in relativistic QFT, where the absence of a trace makes the definition of an entropy measures particularly challenging. However, precisely because it is a very general formula, there are few examples of explicit computations of the Araki relative entropy, and upper and lower bounds can provide important indirect information (for recent results, see Refs.~\cite{Chialastri2026,Frob2025,Hollands2025,Hollands2026,Longo2024,LongoMorinelli2024,LongoXu2018}).

The method presented here uses a Feynman-Kac formula to represent the Araki relative entropy between equilibrium states in terms of probabilistic quantities. Indeed, if a von Neumann algebra is \emph{stochastically positive}, its correlation functions can be computed as expectation values of a random process realised on a probability space $(\Q,\mu)$ \cite{Klein1981b,Klein1981}. 
In particular, a suitable class of (possibly unbounded) perturbations of the relativistic bosonic field is in correspondence with a class of perturbations of Nelson's Euclidean free field.

\subsection{The relativistic bosonic field and stochastic positivity}
The classical relativistic bosonic field at inverse temperature $\beta$ is described as a solution with spatially compact support of the massive Klein-Gordon equation on Minkowski spacetime,
\[
(\partial_t^2 - \Delta_x+m^2)\phi =0 \ .
\]
By the Leray theorem, the space of solutions is in one-to-one correspondence with the linear space of Cauchy data $(q,p) \ni \mathscr L \equiv C^\infty_c(\Cauchy) \oplus C^\infty_c(\Cauchy)$, such that $\phi|_\Cauchy=q$, $\partial_t\phi|_\Cauchy = p$, and $\Cauchy$ is a Cauchy surface isometric to $\mathbb R^{d-1}$. 

Its quantization is described by the Weyl von Neumann algebra $\vNb \subset \mathcal B(\mathcal H)$ generated by the Weyl unitaries $W_\omega(f)$, with $f \equiv \h^{1/2} q + i \h^{-1/2}p \in \mathfrak h$, where $\h \equiv (-\Delta_x + m^2)^{1/2}$ and $\mathfrak h$ is the one-particle Hilbert space. The cyclic and separating vector $\Omega_\beta \in \mathcal H$ is assumed to be the vector representative of a Gaussian KMS state on $\vNb$ at inverse temperature $\beta$. In particular, this means that $e^{isL}\Omega_\beta = \Omega_\beta$ for every $s\in \mathbb R$, where $L$ is the \emph{Liouvillean} associated with the free dynamics, described by a strongly continuous one-parameter group of automorphisms $\alpha^s : \vN \to \vN$. For analytic elements of the free dynamics $A_i \in \vNb$ and $-\beta/2 < s_1 <\ldots < s_n <\beta/2$, the \emph{Euclidean Green's functions} are defined by
\[
G^E_n(A_1,\ldots,A_n;s_1,\ldots,s_n) \equiv (\Omega,\prod_{j=1}^n\alpha^{is_j}(A_j) \Omega ) \ .
\]

It is possible to show that on the Abelian subalgebra $\vNb^+$ generated by $W_\omega(\mathfrak h_+)$, with $\mathfrak h_+ = \overline{\{ f \in \mathfrak h \ | \ f = \h^{1/2}q, \ q \in C^\infty_c(\Cauchy) \}}^{L^2(\Mb)}$, for every positive element $A_1,\ldots ,A_n \in \AvN$, the Euclidean Green's functions are positive. The system $(\vNb,\vNb^+,\alpha,\omega)$ is called a \emph{Gaussian stochastically positive KMS system}.

Using the Klein-Landau correspondence, the relativistic bosonic fields on Minkowski spacetime at finite temperature can be put in correspondence with Nelson's massive Euclidean free field on the thermal cylinder $\Mb \equiv \mathbb S_\beta \times \mathbb R^{d-1}$. This is defined as the probability space $(\Q,\Sigma,\mu_C)$, where $\Q=\mathscr D'(\Mb)$ and $\mu_C$ is the Gaussian measure with covariance $C(f,g) = (f, (-\partial_\tau^2-\Delta_x +m^2)^{-1} g)$ for every $f, \ g \in \mathscr D(\Mb) \equiv C^\infty_c(\Mb)$. The abelian subalgebra $\vNb^+$ is then in correspondence with a distinguished sub-$\sigma-$algebra $\Sigma_0$. Together with a strongly continuous group of measure-preserving automorphisms $U(\tau)$ (implementing the Euclidean time evolution) and a reflection operator $R$, the tuple $(\Q,\Sigma,\Sigma_0,U(\tau),R,\mu_C)$ is a \emph{generalised path space}.

Our first result gives a probabilistic formula for the relative entropy between a state $\Omega$ which is KMS with respect to the free dynamics, and a perturbed $\Omega_K$ state, KMS with respect to a perturbed dynamics. 

Consider a self-adjoint operator $K$ affiliated to $\vNb^+$, representing an interaction Hamiltonian. By the Klein-Landau correspondence, $K$ can be identified with a measurable function $K:\Q\to \mathbb R$. We assume that there are two real numbers $p,\ q \geq 2$, with $\frac1p+\frac1q=\frac12$, such that $K \in L^p(\Q,\Sigma_0,\mu_C)$ and $e^{-\beta K} \in L^q(\Q,\Sigma_0,\mu_C)$, and define $V \equiv \int_{-\beta/2}^{\beta/2} U(\tau) K \dd \tau$. Then, by Prop.~\ref{prop:FKN-kernels}, the  KMS state for the perturbed dynamics is $\omega_K(A) \equiv \norm{\Omega_K}^{-2} (\Omega_K, A\Omega_K)$, where $\Omega_K = e^{-\beta/2(L+K)}\Omega$. Using a Feynman-Kac formula, we get (see Prop.~\ref{prop:relative-entropy-probabilistic})
\[
S(\omega | \omega_K) = 
\E_\mu\left ( V \right ) + \log \E_\mu(e^{-V}) \ .
\]
Up to the normalisation of the perturbation, the Araki relative entropy corresponds to the classical functional $\Ent_\mu$ for the observable $F = \E_\mu\left(e^{-V}\right)e^V$. 

The probabilistic methods of Bauerschmidt, Bodineau, and Dagallier can then be transported back to QFT, yielding a bound for the Araki relative entropy if the interaction $V$ satisfies the multiscale Bakry-\Emery criterion.

Bauerschmidt, Bodineau, and Dagallier formulate the Polchinski equation initially in a finite-dimensional probability space $\Q \subseteq \mathbb R^N$, introducing a lattice cutoff. Proving that the LSI holds uniformly in $N$ gives a well-defined continuum limit $N \to \infty$.

In the case of a stochastically positive relativistic bosonic field, the corresponding probability space is $\mathscr D'(\Mb)$. This requires a genuinely infinite-dimensional formulation of the Polchinski equation. The formulation of potential theory in abstract Wiener spaces introduced by Gross \cite{Gross1967,Gross1967potential} provides a reasonably simple framework to give a non-perturbative derivation of the RG flow in infinite dimensions.

\subsection{Polchinski equation in infinite dimensions and the entropy bound}

The first step is a refinement of the construction of the massive Euclidean free field, so that its probability space is not $\mathscr D'(\Mb)$, but rather a smaller Hilbert space $\Ba$ such that $\mu_C(\Ba) = 1$.

More precisely, the massive Euclidean free field can be realised as an abstract Wiener space $(\Hi_C, \Ba,\mu_C)$, in terms of the Cameron-Martin space 
\[
\Hi_C \equiv \left ( C^{1/2}L^2(\Mb), \ ( \cdot \ , \cdot )_{\Hi_C} \right ) \ ,
\]
and the weighted Sobolev space
\[
\Ba \equiv \left \{ u \in \mathscr D'(\Mb) \ | \ \norm{w_\rho C^a u }_{L^2(\Mb)} < \infty \right  \} \ , \quad w_\rho: \mathbb R^{d-1} \ni x \mapsto \left (1+ \abs{x}^2 \right )^{-\rho/2} \ ,
\]
for some fixed $\rho > \frac{d-1}{2}$ and $a > \frac{d-2}{4}$. We denote by $\p$ the realisation of $\mu_C$ on $\Ba$. The distinguished sub-$\sigma-$algebra $\Sigma_0$ is generated by the sharp-time random
fields $\varphi(0,f)$, with $f\in\mathscr D(\mathbb R^{d-1})$. Together with the operator $R$ implementing reflections in Euclidean time, $(\Ba,\Sigma,\Sigma_0,U(\tau), R, \mu_C)$ is the generalised path  space associated to the Euclidean massive free field.

The rich structure of an abstract Wiener space allows to introduce an RG flow for perturbations of the massive Euclidea free field, in the form of an infinite-dimensional Hamilton-Jacobi equation.
The idea behind a continuum RG flow is a smooth decomposition of the covariance in infinitesimal increments.

Consider a family of positive, symmetric bi-distributions $\dot C_t$
with bounded realisations on $L^2(\Mb)$, and set
\[
C_t\equiv\int_0^t\dot C_s\dd s \ , \quad C_\infty=C \ ,
\]
where the integral is understood weakly on test functions. We assume that $t\mapsto\dot C_t$ is weakly continuous at $t=0$ on $L^2(\Mb)$ and differentiable for $t>0$, with $\ddot C_t$ bounded. For $t>0$, the kernel of $\dot C_t$ belongs to $\mathscr E(\Mb\times\Mb)$ and depends continuously on $t$ in this topology. For the Bochner formula, we also assume differentiability in the same topology. Finally, the covariance decomposition is invariant under Euclidean time translations.

Let $j:L^2(\Mb)\to\Ba$ be the canonical inclusion. Since $0\leq C_t-C_s\leq C$ for $s\leq t$, the operator $Q_{s,t}\equiv j(C_t-C_s)j^*$ is dominated by the covariance operator $jCj^*$ of $\p$. Thus $Q_{s,t}$ is trace class and defines a centred Gaussian measure $\p_{s,t}$ on $\Ba$.

Let $\mathcal F$ be the class of bounded, measurable functions $F:\Ba\to\mathbb C$ with compact spacetime support, such that is three times continuously differentiable in $\Hi_C$ directions, with
\[
\sup_{\varphi\in\Ba}\norm{D^kF(\varphi)}<\infty \ ,
\qquad k=1,2,3 \ .
\]
Here $D^kF(\varphi)$ is the $k$-th Fréchet derivative at $0$ of $h\mapsto F(\varphi+h)$ on $\Hi_C$. Its restriction to test functions defines a compactly supported distribution on $\Mb^k$. We denote the first two distribution kernels by $\nabla F(\varphi)$ and $\Hess F(\varphi)$. We also require \(F(\eta_R\varphi)\to F(\varphi)\) under the spatial cutoffs \(\eta_R\) specified in Sec.~\ref{sec:assumptions}.

For $t>0$, the smoothness of $\dot C_t$ gives the pairings
\[
\Delta_{\dot C_t}F(\varphi)
\equiv \langle\Hess F(\varphi),\dot C_t\rangle \ , \qquad
\dot C_t(\nabla F,\nabla F)(\varphi)
\equiv \langle\nabla F(\varphi)\otimes\nabla F(\varphi),\dot C_t\rangle
\ .
\]
With this definition of the functional Laplacian, the Gaussian
operators 
\[
P_{s,t}F(\phi) \equiv \int_\Ba F(\varphi + \phi) \dd \p_{s,t}(\varphi)
\]
form a time-dependent Markov semigroup~\ref{prop:generators}. Its time-dependent generator is the defined by
\[
\Delta_{\dot C_t} F(\varphi) \equiv \Tr_{L^2(\Mb)} \dot C_t \Hess F(\varphi) \ .
\]
This result allows one to import standard RG techniques into the abstract Wiener space of the Euclidean free field: given a bare interaction $V_0 \in \mathcal F$, the \emph{effective interaction} $
V_t(\phi) \equiv - \log\left[ P_{0,t}e^{-V_0}(\phi) \right]
$
satisfies the Polchinski equation~\cite{Polchinski1983}:
\[
\frac{\partial}{\partial t} V_t = \frac{1}{2}\Delta_{\dot C_t} V_t - \frac{1}{2}\dot C_t(\nabla V_t, \nabla V_t) \ .
\]

Now that the tools of Bauerschmidt, Bodineau, and Dagallier have been reformulated in the infinite-dimensional context of stochastically positive QFT, their derivation of the multiscale Bakry-\Emery criterion carries over to the Araki relative entropy. 
In particular, the RG decomposition of the Araki relative entropy takes the form (Prop. \ref{prop:entropy-decomposition})
\[
S\left(\omega | \omega_K\right) = \frac12 \int_0^\infty \dd t P_{t,\infty}\left[\dot C_t(\nabla V_t, \nabla V_t) \right](0) \ .
\]
A Bochner formula then leads to a probabilistic bound of the Araki relative entropy in the same form as the classical LSI, 
\[
S\left(\omega | \omega_K\right) \leq \frac{1}{2\gamma}\int_\Ba \dot C_0(\nabla V_0,\nabla V_0)(\phi) \dd p(\phi) \ .
\]
The last step consists of translating back the RHS of the bound to an expression in terms of the von Neumann algebra $\vNb$.

Let $D_f:D(D_f) \to R(D_f)$ be the derivation acting on elements of $\vNb$ by
\[
D_f(A) = i [\phi_\omega(f),A ] \ \forall \ A \in D(D_f) \ ,
\]
where $\phi_\omega(f)$ is the generator of the Weyl unitary $W_\omega(f) \in \vNb$. Let $\{p_j\}_{j\geq 1} \subset h^{-1}C^\infty_c(\mathbb R^{d-1})$ be an orthonormal basis in $L^2(\mathbb R^{d-1})$, and set $f_j \equiv i \h^{1/2}p_j$.

Since $K$ is a sharp-time observable and $V_0=\int_{-\beta/2}^{\beta/2}U(\tau)K\dd\tau$, a variation of the field in the direction $h\in\mathscr D(\Mb)$ acts on $U(\tau)K$ through $h(\tau,\cdot)$. Thus, in the sense of weak derivatives,
\[
\nabla V_0(\tau,\cdot)=U(\tau)\nabla K \ .
\]
By invariance of $\p$ under $U(\tau)$, the assumption of Sec.~\ref{sec:assumptions} gives
\[
\beta\E\norm{\nabla K}_{L^2(\mathbb R^{d-1})}^2
=\E\norm{\nabla V_0}_{L^2(\Mb)}^2<\infty \ .
\]
Under the Klein-Landau correspondence, the spatial components of this gradient are the derivations $D_{f_j}K=(\nabla K,p_j)_{L^2(\mathbb R^{d-1})}$, by Eq.~\eqref{eq:quantum-classical-gradient} with $p =\h p_j$. For an affiliated \(K\), \(D_{f_j}K\) denotes the weak \(L^2(\Q)\) derivative of the sharp-time observable under Weyl translations, that is, \(
\nabla K\in L^2\!\left(\mathcal Q,\Sigma_0,\mu_C;\,L^2(\mathcal C)\right)
\).

\begin{theorem}[Multiscale Bakry-\Emery criterion]\label{thm:entropy-inequality}
Consider the Gaussian stochastically positive KMS system $(\vNb,\vNb^+, \alpha,\omega)$ associated to the relativistic bosonic field at inverse temperature $\beta$, together with the $(\beta, \alpha^K)-$KMS state $\omega^K$, and let $(\Hi_C,\Ba,\p)$ be the corresponding abstract Wiener space describing the massive Euclidean free field. 

Let $K$ and the covariance decomposition satisfy the assumptions of Sec.~\ref{sec:assumptions}, and let $V_0=\int_{-\beta/2}^{\beta/2}U(\tau)K\dd\tau\in\mathcal F$. 
Assume that $\nabla V_0(\varphi)$ admits an $L^2(\Mb)$ realisation for $\p-$almost every $\varphi$ and that
\[
\E\norm{\nabla V_0}_{L^2(\Mb)}^2<\infty \ .
\]
Suppose that there is a real locally integrable function $t\mapsto\dot\lambda_t$ on $[0,\infty)$ such that, for every $t>0$ and every $\varphi\in\Ba$,
\begin{equation}\label{eq:BE-criterion}
\dot C_t\Hess V_t(\varphi)\dot C_t-\frac{1}{2}\ddot C_t
\geq \dot\lambda_t\dot C_t
\end{equation}
as quadratic forms on $L^2(\Mb)$. Set
\[
\lambda_t\equiv\int_0^t\dot\lambda_s\dd s \ , \qquad
\frac{1}{\gamma}\equiv\int_0^\infty e^{-2\lambda_t}\dd t<\infty \ .
\]
Then,
\begin{equation}\label{eq:entropy-inequality}
S \left (\omega \mid \omega_K \right ) \leq \frac{1}{2\gamma} \beta \sum_{i,j}\int_{-\beta/2}^{\beta/2}\dd \tau (p_i ,\dot C_0(\tau) p_j)_{L^2(\mathbb R^{d-1})} G^E_2 \left (D_{f_i}K,D_{f_j}K;\tau,0 \right )  \ .
\end{equation}
\end{theorem}

In the last part of the paper, we show an application of the present framework to the Sine-Gordon model. 
Let $\alpha \in \mathbb R$ and $z \in \mathbb R$ be two real parameters defining the theory and let $\epsilon >0$ be a UV regularisation parameter. Let $g \in C^\infty_c(\mathbb R)$, $0\leq g(x) \leq 1$ for every $x\in \mathbb R$. The \emph{Sine-Gordon Hamiltonian} is (Def.~\ref{def:sine-gordon-hamiltonian})
\[
K_\epsilon(\varphi)  \equiv z e^{\frac{\alpha^2}{2}C_{\epsilon,T}(0,0)}\int_{\mathbb R} \dd x g(x) \cos(\alpha \varphi_\epsilon(0,x))   \ .
\]

The interacting equilibrium state $\omega_{K_\epsilon}$ at finite temperature in Minkowski spacetime has been obtained by Bahns, Pinamonti, and Rejzner \cite{Bahns2021}, in the ultraviolet finite regime $\alpha^2 < 4\pi$, and therefore we also restrict to this regime. Then, an adaptation of the methods of Brydges and Kennedy \cite{Brydges1987} and Bauerschmidt and Bodineau \cite{BauerschmidtBodineau2020} to the continuum case shows that the interacting state of the Lorentzian, massive Sine-Gordon model satisfies the multiscale Bakry-\Emery criterion, so we can give a quantitative RG bound for the Araki relative entropy between $\omega_{K_\epsilon}$ and the free state for the bosonic field $\omega$.

In order to control the ultra-violet limit $\epsilon \to 0$ of the Sine-Gordon Hamiltonian, we apply the generalisation of Theorem~\ref{thm:entropy-inequality} given in Remark~\ref{rmk:generalisation}. This allows to prove that the relative entropy bound is uniform in $\epsilon$.

\begin{theorem}\label{thm:entropy-inequality-sine-gordon}
Consider the KMS state for the free relativistic bosonic field $\omega$, and the interacting equilibrium state for the regularised Sine-Gordon dynamics $\omega_{K_\epsilon}$. Fix $\beta,m,T>0$ and $\alpha^2<4\pi$. Consider the heat-kernel decomposition
\[
\dot C_t=\frac{1}{2}e^{-tP/2} \ , \quad
P=-\partial_\tau^2-\Delta_x+m^2 \ .
\]
Then, there is a $z_*=z_*(\alpha,\beta,m,T)>0$ such that, for $\abs{z}<z_*$, every $\epsilon>0$ and every $g\in C^\infty_c(\mathbb R)$ with $0\leq g\leq1$, the effective potentials \(V_{t,\epsilon,g}\) generated by \(V_{0,\epsilon,g}=\int_0^\beta U(\tau)K_\epsilon\,d\tau\) satisfy Eq.~\eqref{eq:BE-criterion} with a single choice of \(\dot\lambda_t\).  Therefore,
\[
S(\omega\mid\omega_{K_\epsilon})
\leq \frac{1}{4\gamma}\E\norm{\nabla V_0}_{L^2(\Mb)}^2 \ ,
\]
or, equivalently, Eq.~\eqref{eq:entropy-inequality}. The constant
$\gamma>0$ can be chosen independently of $\epsilon$ and $g$. Moreover,
\[
S(\omega\mid\omega_{K_\epsilon})
\leq C_{\alpha,\beta,m,T}\beta \abs{z}^2
\norm{g}_{L^2(\mathbb R)}^2 \ .
\]
In particular, for each fixed spatial cutoff \(g\), the relative entropy is bounded uniformly in the ultraviolet regulator \(\epsilon\).
\end{theorem}

\subsection{Bibliographical remarks and outlook}

The Polchinski equation is by now a standard tool in RG methods \cite{Salmhofer1999}. It is a formulation of the Wilsonian RG~\cite{Wilson73} in terms of a continuum flow between RG scales, and it has been widely used in particular for proofs of perturbative renormalisability of various interacting QFTs \cite{Keller1992,Keller1993}. One of the first applications of the Polchinski equation in mathematical physics was the derivation of rigorous bounds on the Mayer expansion by Brydges and Kennedy \cite{Brydges1987}, which we use in Prop.~\ref{prop:fourier-estimates}.

However, it is usually formulated by introducing a lattice cutoff, so that the probability space becomes finite-dimensional, and the proofs of renormalisability are extended to the continuum relying on estimates which are uniform in the cutoff parameter.

Barashkov and Gubinelli have introduced an infinite-dimensional formulation of the RG as a stochastic control problem~\cite{Barashkov2020}; however, it relies on a variational formulation of the effective interaction, rather than a flow equation. Bauerschmidt, Bodineau, and Dagallier's review~\cite{Bauerschmidt2024} discusses in detail the connection between Barashkov and Gubinelli's variational formula and the Polchinski equation. Moreover, Duch introduced a flow equation approach to the renormalisation of stochastic PDEs in the continuum \cite{Duch2026,Duch2025,Duch2022}, but it relies on a different flow equation, and it would be interesting to investigate the connections with his approach. See also Ref.~\cite{Ziebell2023} for a continuum formulation of the Wetterich equation in an infinite-dimensional setting.

There are two possible directions to extend the results in this paper. First, it would be interesting to apply the multiscale Bakry-\Emery criterion and Theorem~\ref{thm:entropy-inequality} to different models. Thanks to stochastic positivity and the probabilistic formula for the Araki relative entropy, there are two natural strategies. One may prove stochastic positivity for a Lorentzian model whose Euclidean LSI is known, or establish an LSI for an already stochastically positive model.

A first, immediate generalisation would be to the Sine-Gordon model up to the threshold $\alpha^2 < 6\pi$, for which Bauerschmidt and Bodineau \cite{BauerschmidtBodineau2020} already proved the result in the Euclidean setting. This extension requires the construction of the interacting state for the Lorentzian Sine-Gordon model, extending the work of Bahns, Pinamonti, and Rejzner~\cite{Bahns2021}. 

Another important class of models would be $P(\varphi)_2$ and $\varphi^4$. Indeed, Bauerschmidt and Dagallier proved the corresponding LSIs for $\varphi^4_{2,3}$~\cite{Bauerschmidt2023}. In contrast to the Sine-Gordon model, in this case the Lorentzian $P(\varphi)_2$ model has been constructed using stochastic positivity by Gérard and Jäkel~\cite{Gerard2005,Gerard2005b}, so that proving the classical LSIs can be translated to an Araki relative entropy bound. It would also be interesting to derive RG bounds for the Araki relative entropy for non-relativistic Bose gases, with applications to Bose-Einstein condensation, in the context introduced by Galanda and Pinamonti \cite{Galanda2025,Galanda2026}.

A second direction of research consists of deriving consequences for the Araki relative entropy bound. Just as the LSI gives quantitative control on the relaxation of stochastic dynamics, it would be interesting to investigate whether Theorem~\ref{thm:entropy-inequality} can be used to obtain quantitative control on the return to equilibrium for open quantum systems~\cite{Derezinski2026}, and put in relation to entropic fluctuation theorems~\cite{Benoist2025}. Moreover, it would be interesting to investigate whether Theorem~\ref{thm:entropy-inequality} can be related to other examples of Araki relative entropy bounds, such as the Bekenstein bound \cite{Bekenstein1980,Hollands2025,Hollands2026,Longo2024,LongoXu2018}. 

Finally, it would be interesting to investigate connections of the RG entropy bound with non-commutative LSIs~\cite{Olkiewicz1999}, which are known to imply spectral gap inequalities~\cite{Carbone2015} and exponential decay of relative entropy~\cite{Carbone2014}, through non-commutative Dirichlet forms~\cite{Albeverio1977,Cipriani1997,Cipriani2003}.

\section{Stochastically positive KMS systems and perturbations}\label{sec:Weyl}

We give here a summary on the construction of the GNS representation of the Weyl $C^*-$algebra as a stochastically positive KMS system, its states and GNS representations, and on imaginary Green's functions. Main references for this section are the standard textbook by Bratteli and Robinson~\cite{BratteliRobinson-1,BratteliRobinson-2}, the classical work of Klein and Landau~\cite{Klein1981} and the works on stochastically positive KMS Weyl systems by Gérard and Jäkel~\cite{Gerard2005,Gerard2005b} and Gielerak, Jakóbczyk, and Olkiewicz~\cite{Gielerak1998}.

\subsection{Weyl algebra}\label{sec:quasi-free-systems}

Consider a Hilbert space $(\mathfrak h, ( \ \cdot \ , \ \cdot \ )_{\mathfrak h})$, together with a symplectic form $\sigma(f,g) = \Im( f,g )_{\mathfrak h}$ (a bilinear, antisymmetric, and non-degenerate map). 

The \emph{Weyl unitaries} are abstract unitaries $W(f)$, indexed by $f \in\mathfrak h$, with product
\[
W(f)W(g) = e^{-\frac{i}{2}\sigma(f,g)} W(f+g) \ ,
\]
and involution $W(f)^* = W(-f)$. The complex algebra generated by $W(f)$ is a $*-$algebra $\mathcal W_0(\mathfrak h, \sigma)$. A $*-$norm on the complex algebra is 
\[
\norm{\sum_{n=1}^N c_n W(f_n)}_1 = \sum_{n=1}^N \abs{c_n} \ ,
\]
and the completion $\overline{\mathcal W_0(\mathfrak h, \sigma)}^{\norm{ \ \cdot \ }_1}$ is a unital Banach $*-$algebra.

A state on an algebra is a normalised, positive functional $\omega:\mathcal A \to \mathbb C$. The completion of $\overline{\mathcal W_0(\mathfrak h, \sigma)}^{\norm{ \ \cdot \ }_1}$ with respect to the $C^*-$norm
\[
\norm{a} = \sup_{\omega} \sqrt{\omega(a^* a)} 
\]
defines the Weyl $C^*-$algebra $\CCR$.

A state is \emph{regular} if the map
\[
\mathbb R \ni t \mapsto \omega\left(W(tf)\right)
\]
is continuous for all $f \in \mathfrak h$. We will always work with regular states in the following.

\subsection{Gaussian and KMS states}
Let $\mathbb R \ni s \mapsto T^s$ be a one-parameter unitary group of symplectic automorphisms on $\mathfrak h$, with self-adjoint generator $\h$. The operator $\h$ is called the \emph{one-particle Hamiltonian}. Then,
\[
\alpha^s : W(f) \mapsto W(T^s f) = W(e^{is\h}f)
\]
defines a strongly continuous one-parameter unitary group of automorphisms of the Weyl algebra $\CCR$. This allows to define a free dynamics in the algebra, even though the free Hamiltonian is usually not an element of $\CCR$. 

A state is $\alpha-$invariant if $\omega \circ \alpha^s = \omega$ for every $s$. Equilibrium states are characterised by the stronger \emph{KMS condition}.

\begin{definition}[KMS state, separability and continuity]\label{def:equilibrium-state}
A state $\omega$ is a $(\alpha,\beta)-$KMS state or \emph{equilibrium state} for the dynamics $\alpha$ at inverse temperature $\beta$ if, for all $a, \ b \in \CCR$, there is a function $\mathbb C \ni z \mapsto f_{a,b}(z)$ such that
\begin{enumerate}
\item \label{property-continuity}  $f_{a,b}$ is continuous and bounded for $0\leq \Im z \leq \beta$;
\item $f_{a,b}$ is holomorphic in $0 < \Im z<\beta$; and
\item for real $s$,
\[
f_{a,b}(s) = \omega \left (a \alpha^s(b) \right ) \ , \quad f_{a,b}(s+i\beta) = \omega\left(\alpha^s(b)a\right ) \ .
\]
\end{enumerate}
\end{definition}
The above conditions characterise KMS states. Furthermore, we assume that
\begin{enumerate}\setcounter{enumi}{3}
\item \label{property-separability} the algebra $\CCR$ can be equipped with a topology $\tau$, so that $(\CCR,\tau)$ is a separable topological space, and $\omega_\beta(ab)$ is jointly continuous in $a$ and $b$ in the product topology on $\CCR \times \CCR$.
\end{enumerate}
The last condition is necessary for infinite systems in order to have a GNS representation on a separable Hilbert space \cite{Birke2002}.

It directly follows from the definition that a $(\alpha,\beta)-$KMS state with $\beta \in \mathbb R \setminus \{0\}$ is an $\alpha-$invariant state.
 Furthermore, we will focus our attention on Gaussian states and their perturbations.

\begin{definition}[Gaussian state]\label{def:gaussian-state}
Consider a real, symmetric, positive, bilinear map $H : \mathfrak h \times \mathfrak h \to \mathbb R$ satisfying
\[
\abs{\sigma(f,g)} \leq 2 \sqrt{H(f,f)}\sqrt{H(g,g)} \ \forall f, \ g \in \mathfrak h \ .
\]

A \emph{Gaussian} (or \emph{quasifree}) state on $\CCR$ is uniquely determined by
\[
\omega(W(f)) = e^{- \frac{1}{2} H(f,f)} \ .
\]
\end{definition}

The covariance $H$ for a Gaussian $(\alpha, \beta)-$KMS state can be written in terms of a bounded operator on $\mathfrak h$ as \cite[Theorem 3.1]{Gielerak1998}
\[
H(f,g) = \frac12 \Re\left \langle f, \frac{1+e^{-\beta \h}}{1-e^{-\beta \h}} g \right \rangle \ .
\]

\subsection{GNS representation and imaginary Green's functions}

The GNS construction of $(\CCR,\omega)$ is a representation on the algebra of bounded operators on the Hilbert space $\mathcal H_\omega$, $\pi_\omega : \CCR \to \pi_\omega(\CCR) \subset \mathcal B(\mathcal H_\omega)$, together with a distinguished vector $\Omega \in \mathcal H_\omega$ such that 
\[
\omega(a) = (\Omega, \pi_\omega(a) \Omega) \ \forall \ a \in \CCR \ ,
\]
where the brackets denote the inner product in $\mathcal H_\omega$. The von Neumann algebra representing the CCR relations is $\vN \equiv \pi_\omega(\CCR)''$.

In the following, by a slight abuse of notation we denote $\omega: \vN \to \mathbb C$ the algebraic state on the von Neumann algebra obtained by $\omega(A) \equiv (\Omega, A \Omega)$ for all $A \in \vN$. Moreover, we denote $W_\omega(f) \equiv \pi_\omega(W(f))$.

The GNS representation $(\mathcal H_\omega, \pi_\omega, \Omega)$ of $(\CCR, \omega)$ is called \emph{regular} if $\omega$ is regular.
If the representation is regular, by Stone's theorem we can define the \emph{field operator} as the generator of the Weyl unitary operator,
\[
\phi_\omega(f) \equiv \frac{1}{i}\eval{\frac{\dd}{\dd t}W_\omega(tf)}_{t=0} \ .
\]

If property~\ref{property-separability} in Def.~\ref{def:equilibrium-state} holds for a state $\omega$, then the Hilbert space $\mathcal H_\omega$ obtained by the GNS construction of $(\CCR, \omega)$ is separable, and $\Omega_\omega$ is a cyclic and separating vector.

From now on, we assume that $(\CCR, \omega, \alpha)$ satisfies Properties~\ref{property-continuity} and~\ref{property-separability} in Def. \ref{def:equilibrium-state}, we fix a regular, Gaussian, and equilibrium state $\omega$ at inverse temperature $\beta$, and we consider the associated von Neumann algebra $\vN \equiv \pi_{\omega}(\CCR)''$. The vector representative of $\omega$ in $\Hi_{\omega}$ is the cyclic and separating vector $\Omega$.

In this setting, there exists a strongly continuous one-parameter group of unitaries $\{e^{isL}\}_{s \in \mathbb R}$, with a self-adjoint generator $L$, which represents $\alpha$ in $\vN$. In particular, we have
\[
\pi_\omega(\alpha^s(a)) = e^{i sL} \pi_\omega(a) e^{-i s L} \ \forall a \in \CCR \ , \quad e^{i sL} \Omega =\Omega \ .
\]
The self-adjoint generator $L$ is called the \emph{Liouvillean}. In the following, with a slight abuse of notation we denote $\pi_\omega(\alpha^s) =\alpha^s$. If $\omega$ is an algebraic $(\alpha,\beta)-$KMS state, the corresponding state on the von Neumann algebra $\omega(A) = (\Omega, A \Omega)$ is also a $(\alpha,\beta)-$KMS state.

For arbitrary $A_i \in \vN$ and $s_i \in \mathbb R^n$, with $i=1,\ldots,n$, we define the \emph{real-time Green's functions}
\[
\omega_n(A_1,\ldots,A_n;s_1,\ldots,s_n) \equiv \omega \left (\prod_{i=1}^n\alpha^{s_i}(A_i)\right ) 
\]
From the KMS condition, a theorem of Araki~\cite{Araki1968} shows that the real-time Green's functions for a KMS state $\omega_n(A_1,\ldots,A_n;s_1,\ldots,s_n)$ can be analytically continued to the functions $\omega^\beta_n(A_1,\ldots,A_n;z_1,\ldots,z_n)$. These are holomorphic in the region
\[
\mathcal I_\beta \equiv \left \{ (z_1,\ldots,z_n) \in \mathbb C^n \ | \ - \frac{\beta}{2} < \Im z_1 < \ldots < \Im z_n < \frac{\beta}{2} \right \}
\]
and continuous on the boundary of $\mathcal I_\beta$. 

\begin{definition}\label{def:euclidean-Greens-functions}
Let $(\vN,\omega)$ be a von Neumann algebra with a regular equilibrium state $\omega$, and let
\[
\mathcal R_\beta \equiv \left \{ (s_1,\ldots,s_n) \in \mathbb R^n \ | \ - \frac{\beta}{2} < s_1 < \ldots < s_n < \frac{\beta}{2} \right \} \ .
\]
The \emph{Euclidean Green's functions} $G^E(A_1,\ldots,A_n;s_1,\ldots,s_n)$ are defined by
\[
G^E(A_1,\ldots,A_n;s_1,\ldots,s_n) \equiv \omega_n(A_1,\ldots,A_n;is_1,\ldots,i s_n) \ , \quad (s_1,\ldots,s_n) \in \mathcal R_\beta \ .
\]
\end{definition}

For a Gaussian state, the two point Green function can be computed explicitly in terms of $H$ and $\sigma$ as
\[
G_2(f,g;s)= \exp{-\frac{1}{2}H(f,f)-\frac{1}{2}H(g,g)-F(g,f;s)} \ ,
\]
where
\[
F(g,f;s)= H(f,T^s g)+ \frac{i}{2}\sigma(f,T^s g) \ .
\]

Since the state is KMS, $F(f,g;t)$ can be analytically continued to a function $F(f,g;z)$, analytic in the strip $0 < \mathrm{Im}z < \beta$ and such that $F(f,g;t+i0) = F(f,g;t)$. The imaginary-time field Green's function is
\begin{equation}\label{eq:gaussian-covariance} 
F^E(f,g;s) \equiv F(f,g;is)  \ ,
\end{equation}
and the two-point Euclidean Green's function is 
\[
G_2^E(f,g;s)= \exp{-\frac{1}{2}H(f,f)-\frac{1}{2}H(g,g)-F^E(g,f;s)} \ .
\]

\begin{proposition}\label{prop:OS-positivity}
\emph{\cite[Proposition 3.1]{Gielerak1998}}
Let $F^E : \mathfrak h \times \mathfrak h \times [0,\beta]  \to \mathbb R$ be defined by Eq.~\eqref{eq:gaussian-covariance}. Then, for every $f, \ g \in \mathfrak h$ and $s \in [0,\beta]$,
\[
F^E(f,g;s) = F^E(g,f; \beta-s) \ .
\]
Moreover, for every sequence $\{s_n\}_n$, $s_n \in [0,\beta/2]$, and $\{c_n\}_n \in \mathbb C$,
\[
\sum_{n,m} \bar c_n c_m F^E(f,f;s_n+s_m) \geq 0 \ .
\]
The last condition is called \emph{Osterwalder-Schrader (OS) positivity}. More generally, for every sequence $\{f_n\}_n$, $f_n \in \mathfrak h$,
\[
\sum_{n,m} \bar c_n c_m F^E(f_n,f_m;s_n+s_m) \geq 0 \ .
\]
\end{proposition}

As an immediate corollary, it follows that 
\begin{gather*}
G^E_2(f,g;s) = G^E_2(g,f;\beta-s) \ , \quad s \in \left [0,\frac{\beta}{2} \right ] \ , \\
\sum_{n,m}\bar c_n c_m G^E_2(-f_n,f_m;s_n+s_m) \geq 0 \ .
\end{gather*}

A \emph{stochastically positive KMS system} is identified if the von Neumann algebra $\vN$ contains an Abelian subalgebra $\AvN$, which generates $\vN$ under the free dynamics.

For the von Neumann algebra $\vN = \pi_\omega(\CCR)''$, there is a canonical choice for the abelian subalgebra given by what is known as \emph{Abelian splitting} \cite{Gielerak1998}.
\begin{definition}[Abelian splitting]
An \emph{Abelian splitting} of a symplectic space $(\mathscr L,\sigma)$ is a pair $(\mathscr L_+,\mathscr L_-)$ of linear subspaces of $\mathscr L$ such that $\mathscr L_+ + \mathscr L_- = \mathscr L$, and $\sigma(\mathscr L_\pm,\mathscr L_\pm) = 0$. For such a splitting, let $\vN^\pm \subset \vN$ be obtained as the double commutant of the Abelian subalgebras generated by $W_\omega(f_\pm) \ , \ f_\pm \in \mathscr L_\pm$. Then, the system $(\vN, \vN^\pm, \alpha,\omega)$ is \emph{$\vN^\pm-$stochastically positive} if, for every positive element $A_1,\ldots,A_n \in \vN_\pm^+$ and every $s_1,\ldots,s_n \in \mathcal R_\beta$,
\[
G^E(A_1,\ldots,A_n;s_1,\ldots,s_n) \geq 0 \ .
\]
\end{definition}

In the case in which $\mathscr L$ is a Hilbert space $\mathfrak h$, a conjugation on $\mathfrak h$ canonically identifies an Abelian splitting. In fact, consider an abstract conjugation $\Theta$ on $\mathfrak h$ such that $\Theta D(\h) = D(\h)$ and $\Theta \h f = \h \Theta f$ for every $f \in \mathfrak h$. This defines a natural Abelian splitting $\mathfrak h = \mathfrak h_+ + \mathfrak h_-$, with $\mathfrak h_\pm \equiv \{ f \in \mathfrak h \ | \ \Theta f = \pm f \}$. In fact, since $\Theta$ is a complex conjugation, $(f,g)_{\mathfrak h_\pm} \in \mathbb R$, and so it identifies isotropic subspaces. Furthermore, 
\[
F^E(f,g;s) = F^E(g,f;s)  \ \forall \ f, \ g \in \mathscr L_\pm \ , s \in [0,\beta] \ ,
\]
and by Prop.~\ref{prop:OS-positivity} we can extend $s\mapsto F^E(f,g;s)$ to a periodic function with period $\beta$, defined over the entire real line. We denote the extended function by the same symbol.

Moreover, for all finite sequences $f_n \in \mathfrak h_\pm$, $c_1,\ldots,c_n \in \mathbb C$ and  $s_1,\ldots,s_n \in \mathbb R$, $F^E$ is positive-definite:
\[
\sum_{n,m} \bar c_n c_m F^E(f_n,f_m,s_n-s_m) \geq 0 \ .
\]

Since $F^E$ is bounded, symmetric, bilinear, and positive-definite, there exists a bounded and positive operator $R_\beta(s)$ on $\mathfrak h_+$ such that $F^E(f,g;s) = \frac12 \langle f, R_\beta(s) g \rangle $ for $f, \ g \in \mathfrak h_+$, and the function $s \mapsto R_\beta(s)$ is positive-definite, OS-positive, and periodic. The operator $\tfrac12 R_\beta(s)$ will be identified with the covariance of a Gaussian process under the Klein-Landau correspondence.

 \subsection{Automorphisms of translations}\label{sec:automorphisms-translations}
To obtain the quantum inequality in Thm. \ref{thm:entropy-inequality} we need to introduce an appropriate notion for the infinitesimal derivation with respect to the field operators in the von Neumann algebra. To this end, we can consider the adjoint action of Weyl operators on the algebra $\CCR$. By direct computation, this acts by translations,
\[
W(f) W(g) W(f)^* = e^{-i\sigma(f,g)}W(g) \ .
\]

Then, in a regular GNS representation $\vN$, the adjoint action of Weyl operators generates a strongly continuous unitary group of automorphisms defined by
\[
\beta^s_f: A \mapsto W_\omega(s f) A W_\omega(s f)^* \ .
\]
By Stone's theorem, since $\beta^s_f$ is a strongly continuous group on $\vN$, it is possible to define a derivation $D_f: D(D_f) \to R(D_f)$ by
\[
D_f(A) = \lim_{s \to 0} \frac{1}{s}\left [ \beta^s_f(A) - A \right ] \ ,
\] 
where $D(D_f)$ is a strongly dense subalgebra of $\vN$. Since $\beta^s_f(A) = e^{i s\phi_\omega(f)}Ae^{-i s\phi_\omega(f)}$, then
\[
D_f(A) = i[\phi_\omega(f),A] \ \forall \ A \in D(D_f) \ .
\]


\subsection{Perturbations}\label{sec:perturbations}

The pair $(\vN, \alpha)$ is an example of a \emph{$W^*-$dynamical system}, and it is well-suited to describe the free dynamics of a system at equilibrium. Now, we study perturbations of the dynamics, and identifies the associated perturbed equilibrium state. 

We can now consider a self-adjoint element $K \in \vN$. If $K$ is bounded, by a well-known theorem of Araki~\cite{Araki1973}, it is possible to associate a time evolution $\alpha_K$, given by a convergent perturbative expansion with respect to the free evolution $\alpha$, and there is a unique $(\alpha_K,\beta)-$equilibrium state, constructed by Araki in terms of the $(\alpha,\beta)-$KMS state for the free dynamics $\alpha$~\cite[Proposition 5.4.1.]{BratteliRobinson-2}. 

More recently, Derezi\'{n}ski, Jak{\v s}i{\'c}, and Pillet proved that the same results hold for an unbounded perturbation affiliated to a von Neumann algebra, under certain conditions \cite{DerezinskiJaksicPillet2003}. We summarise the results relevant for this work in the following theorem (for the proof, see Ref. \cite{DerezinskiJaksicPillet2003}).

\begin{theorem}\label{thm:interacting-equilibrium-state}
\emph{\cite[Theorem 3.3 and Theorem 5.1]{DerezinskiJaksicPillet2003}}
Consider the $W^*-$dynamical system $(\vN, \alpha)$, and let $L$ denote the infinitesimal generator of $\alpha$. Let $\Omega$ be the vector representative of the $(\alpha,\beta)-$KMS state. Let $K$ be a self-adjoint operator affiliated to $\vN$, such that $L+K$ is essentially self-adjoint on $\mathcal D(L) \cap \mathcal D(K)$, where $\mathcal D$ denotes the domain of the operator; denote its closure by the same symbol. Let
\[
\alpha^s_K(A) = e^{is(L+K)}Ae^{-is(L+K)} \ \forall \ A \in \vN \ .
\]
Then, $\alpha_K$ is a strongly continuous group of $*-$automorphisms on $\vN$. 

Moreover, if
\[
\norm{ e^{-\frac{\beta}{2}K}\Omega} < \infty \ ,
\]
then, $\Omega \in D(e^{-\beta(L+K)/2})$. Set $\Omega_K = e^{-\beta(L+K)/2}\Omega$, and
\[
\omega_K(A) = \frac{(\Omega_K, A \Omega_K)}{(\Omega_K, \Omega_K)}  \ \forall \ A \in \vN  \ .
\]
Then, $\Omega_K$ is cyclic and separating and $\omega_K$ is a $(\alpha_K,\beta)-$KMS state.
\end{theorem}

\subsection{Araki relative entropy}

The generalisation of the classical entropy functional $\Ent_\mu(F)$ given in the introduction is based on the modular theory of Tomita and Takesaki~\cite{Takesaki1970}. We give here only the fundamental notions to define the Araki relative entropy, referring to the textbook~\cite{BratteliRobinson-1,BratteliRobinson-2} for more details.

Let $\omega$ be a faithful state on $\vNb$, with a cyclic and separating vector representative $\Omega \in \Hi$, and let $\Psi \in \Hi$. Let $S_{\Omega,\Psi}$ be the closure of the unique antilinear operator defined on $\mathfrak M \Omega$ by
\[
S_{\Omega,\Psi} A \Omega = A^* \Psi \ \forall A \in \vNb \ .
\]

Now, consider an anti-unitary isometric involution $J$ on $\Hi_\beta$ (meaning that $J$ is antilinear, $J^2 = 1$ and $J^*=J$), mapping $J \vNb = \vNb'$, and the self-dual cone $\mathcal P  = \overline{\{ AJAJ\Omega \ | \ A \in \vNb \}}$ associated with $\Omega$ and invariant under $J$.

\begin{definition}\label{def:relative-entropy}
Consider two normal states $\varphi,\psi$ on $\vNb$, with unique vector representatives $\Phi, \ \Psi \in \mathcal P$.
The \emph{relative entropy} of $\varphi$ and $\psi$ is defined by
\[
S(\varphi | \psi ) \equiv - (\Phi, \log\Delta_{\Phi,\Psi} \Phi) \ ,
\]
where the \emph{relative modular operator} is $\Delta_{\Phi,\Psi}\equiv S_{\Phi, \Psi}^* S_{\Phi,\Psi}$.
\end{definition}

In the special case in which $\varphi$ is the $(\alpha,\beta)-$equilibrium state $\omega$, and $\psi = \omega_K$, the relative entropy admits an expression reminiscent of the classical thermodynamical relation. 

\begin{proposition}\label{prop:relative-entropy} 
\emph{\cite[Theorem 5.5]{DerezinskiJaksicPillet2003}}
Consider the assumptions in Thm. \ref{thm:interacting-equilibrium-state}.
Furthermore, assume that the operator $L_K \equiv L + K - J K J$ is essentially self-adjoint on
\[
\mathcal D(L) \cap \mathcal D(K) \cap \mathcal D(J K J) \ .
\]
Then,
\begin{equation}\label{eq:relative-entropy-1}
S(\omega |\omega_K) = \beta(\Omega,K \Omega) +\log \norm{\Omega_K}^2 \ .
\end{equation}

\end{proposition}
The above formula translates to the quantum setting the thermodynamical relation between the entropy and the energy of a classical system; indeed, on one hand, the first term is the expectation value of the relative energy of the system times the inverse temperature; on the other hand, $-1/\beta \log \norm{\Omega_K}^2$ represents the quantum free energy of the system.

\section{Stochastically positive KMS systems and stochastic processes}\label{sec:stochastic}

We can now introduce the correspondence between quantum von Neumann algebras and stochastic processes. The correspondence relies on the identification of an abelian subalgebra of the von Neumann algebra, which must be large enough to generate the full algebra under the free dynamics. The abelian algebra can then be realised as a sub-$\sigma-$algebra of a probability space. A time evolution of the subalgebra in a fictitious Euclidean direction preserves commutativity, generating the full probability space. Osterwalder-Schrader positivity and stochastic positivity guarantee a correspondence between the commutative probability space and the full von Neumann algebra. Main references for this section are the classical work of Klein and Landau~\cite{Klein1981}, the work on quasifree Weyl systems by Gielerak, Jakóbczyk, and Olkiewicz~\cite{Gielerak1998}, and the papers on the application to $P(\phi)_2$ models by Gérard and Jäkel~\cite{Gerard2005,Gerard2005b}, which also give a modern introduction to the subject.

We start with a few definitions.
\begin{definition}[Stochastically positive KMS systems]\label{def:stochastically-positive-KMS-system}
Let $(\gvN, \AvN, \alpha,\omega)$ be a von Neumann algebra $\gvN \subset \mathcal B(\Hi)$ with a weakly closed abelian sub-algebra $\AvN$, a dynamics $\alpha:\gvN \to \gvN $ given by
\[
\alpha^s(A) = e^{is L}A e^{-i s L}
\]
for a self-adjoint operator $L$ on $\mathcal H$, and an equilibrium state $\omega$. Then, $(\gvN, \AvN, \alpha,\omega)$ is a \emph{stochastically positive KMS system} if
\begin{enumerate}
\item The von Neumann algebra generated by $\bigcup_{t \in \mathbb R} \alpha^t(\AvN)$ is $\gvN$; 
\item the Euclidean Green's functions $G^E(A_1,\ldots,A_n; s_1,\ldots,s_n) $ are positive for all $A_1,\ldots,A_n \in \AvN^+ = \{ A \in \AvN \ | \ A \geq 0 \}$.
\end{enumerate}
If $\omega$ is Gaussian, $(\gvN, \AvN, \alpha,\omega)$ is a \emph{Gaussian stochastically positive KMS system}.
\end{definition}

Stochastically positive KMS states are associated to generalised path spaces.

\begin{definition}
A generalised path space $(\Q,\Sigma,\Sigma_0,U(t),R,\mu)$ consists of
\begin{enumerate}
\item A probability space $(\Q,\Sigma,\mu)$;
\item A distinguished sub-$\sigma-$algebra $\Sigma_0 \subset \Sigma$;
\item A one-parameter group $\mathbb R \ni \tau \mapsto U(\tau)$ of measure preserving automorphisms of $L^\infty(\Q,\Sigma,\mu)$, strongly continuous in measure, such that $\Sigma = \bigvee_{\tau \in \mathbb R} U(\tau)\Sigma_0$;
\item a measure-preserving automorphism $R$ of $L^\infty(\Q,\Sigma,\mu)$ such that $R^2 = 1$, $R U(\tau) = U(-\tau)R$ and $R E_0 = E_0 R$, where $E_0$ is the conditional expectation with respect to $\Sigma_0$.
\end{enumerate}

If $U(\beta)=1$, then the path space is $\beta-$periodic. For $I \subset \mathbb R$, $E_I$ denotes the conditional expectation with respect to the $\sigma-$algebra $\Sigma_I = \bigvee_{t \in I} U(t)\Sigma_0$.

A $\beta-$periodic path space  is \emph{Osterwalder-Schrader (OS-) positive} if $E_{[0,\beta/2]}RE_{[0,\beta/2]} \geq 0$ as an operator on $L^2(\Q,\Sigma,\mu)$.
\end{definition}

\subsection{The isomorphism}
Stochastically positive KMS systems are in one-to-one correspondence with $\beta-$periodic, OS-positive path spaces \cite{Gerard2005,Klein1981}. More precisely, Klein and Landau formulated the correspondence in terms of stochastic processes supported on path spaces. More economically, generalised path spaces allow to give a coordinate-free probabilistic representation of a stochastically positive KMS system \cite{Gerard2005,Klein1978}.
We give here a brief account of part of the correspondence \cite{Klein1981}, to exhibit useful formulas for the next sections. 

Set $\mathcal H_{OS} \equiv L^2(\Q,\Sigma_{[0,\beta/2]},\mu)$, and define the inner product 
\[
(\phi,\psi) = \int_\Q \overline \phi R \psi \dd \mu \ .
\]
Let $\mathcal N \subset \mathcal H_{OS}$ be the kernel of the positive quadratic form $(\psi,\psi)$. The \emph{physical Hilbert space} is $\mathcal H = \overline {\mathcal H_{OS}/\mathcal N}^{( \ \cdot \ , \ \cdot \ )}$. Let $\mathcal V$ be the canonical map $\mathcal V : \mathcal H_{OS} \to \mathcal H_{OS}/\mathcal N$, $\mathcal V : \psi \mapsto [\psi] \equiv \psi + \mathcal N$. Then, the distinguished vector $\Omega \in \mathcal H$ is
\[
\Omega = \mathcal V 1 \ ,
\]
where $1 \in \mathcal H_{OS}$ is the constant function equal to $1$ on $\Q$. For $A \in L^\infty(\Q,\Sigma_0,\mu)$, define
\[
\tilde A \mathcal V \psi \equiv \mathcal V A \psi \ , \quad \tilde A \in \mathcal B(\mathcal H) \ ,
\]
and define $\AvN \subset \mathcal B(\mathcal H)$ the von Neumann subalgebra $\AvN \equiv \{ \tilde A \ | \ A \in L^\infty(\Q,\Sigma_0,\mu)\}$. Then $F \mapsto \tilde F$ is a weakly continuous $*-$isomorphism $\iota: L^\infty(\Q,\Sigma_0,\mu) \to \AvN$ \cite{Klein1981}. 

Now, having identified the Abelian subalgebra $\AvN$ with $L^\infty(\Q,\Sigma_0,\mu)$, the identification of the full algebra $\gvN$ with the path space proceeds identifying the commutative and quantum dynamics, and then using the facts that $\left(\bigcup_{s \in \mathbb R} \alpha^s(\AvN) \right)''= \gvN$ and $\Sigma = \bigvee_{\tau \in \mathbb R} U(\tau)\Sigma_0$.

To this end, set $\mathcal M_\tau = L^2(\Q,\Sigma_{[0,\beta/2-\tau]},\mu)$ with $\tau \in [0,\beta/2]$, and set $\mathcal D_\tau = \mathcal V \mathcal M_\tau$. The automorphism group $U(\tau)$ generates a map $P(s) : \mathcal D_\tau \to \mathcal H$ for $s \in [0,\tau]$ as
\[
P(s) \mathcal V \psi \equiv \mathcal V U(s) \psi \ , \psi \in \mathcal M_\tau \ .
\]
The triple $(P(\tau),\mathcal D_\tau, \beta/2)$ is a \emph{local symmetric semigroup} \cite{Frohlich1980,Klein1981b} (see e.g. Ref~\cite[Definition 6.2]{Gerard2005} for the definition), and it has an associated self-adjoint operator $L$ on $\mathcal H$ such that $P(s) u = e^{-sL}u$, for all $u\in \mathcal D_\tau$ and $s \in[0,\tau]$. 

Now, $\gvN$ is the von Neumann algebra generated by $\{ e^{isL} A e^{-isL} \ | \ s \in \mathbb R, \ A \in \AvN \}$, the weakly continuous group of $*-$automorphisms is given by $\alpha^s: A \mapsto e^{isL}A e^{-isL}$ for $A \in \gvN$, and the state is $\omega: A \mapsto (\Omega, A\Omega) \ \forall A \in \gvN$. Klein and Landau \cite{Klein1981} showed that $(\gvN, \AvN,\alpha,\omega)$ is a stochastically positive KMS system, and the identification between the two objects is fixed by the key identity \cite[Eq. 2.2]{Gerard2005b}
\begin{equation}\label{eq:n-point-functions-probabilistic}
G^E(\tilde A_1,\ldots \tilde A_n; s_1,\ldots,s_n)= \int_\Q \prod_{i=1}^n U(s_i)A_i \dd \mu 
\\ .
\end{equation}

The Klein-Landau correspondence identified above implies that Gaussian, stochastically positive, KMS Weyl systems can be represented by a Gaussian thermal process with covariance $\tfrac12 R_\beta$ \cite{Gielerak1998}.
\begin{theorem}\label{thm:probabilistic-Weyl-algebra}
\emph{\cite[Theorem 2.3 and Corollary 3.2]{Gielerak1998}}.
Let $(\vN, \vN^+,\alpha,\omega)$ be a stochastically positive KMS system for a Gaussian equilibrium state $\omega$. Then, there exists a Gaussian thermal process $\{\xi_s^\beta\}_s$, indexed by $\mathfrak h_+$, with underlying probability space $(\Q,\mu_C)$, and the covariance is $\tfrac12 R_\beta(s)$. More precisely, denoting
\[
\E(F) = \int_\Q F \dd \mu_C \ ,
\]
we have
\[
\E\left [ \xi_s^\beta(f) \right ] = 0 \ , \quad \E \left [ \xi_{s_1}^\beta(f_1)\xi_{s_2}^\beta(f_2)\right ]= \frac{1}{2}\langle f_1, R_\beta(s_1-s_2) f_2 \rangle \ ,
\]
and
\[
G^E_2(f_1,f_2; s_1,s_2) = \E \left [e^{i \xi_{s_1}^\beta(f_1)} e^{i\xi_{s_2}^\beta(f_2)} \right ] \ .
\]
Moreover, higher $n-$point functions are
\begin{equation}\label{eq:n-point-functions-probabilistic-Weyl}
G^E_n(f_1,\ldots,f_n;s_1,\ldots,s_n) = \E \left [ \prod_{i=1}^n e^{i\xi_{s_i}^\beta(f_i)} \right ] \ .
\end{equation}
\end{theorem}

\subsection{Feynman-Kac-Nelson perturbations and probabilistic formula for Araki relative entropy} 
The correspondence in Thm. \ref{thm:probabilistic-Weyl-algebra} is best suited to give a probabilistic description of the dynamics for free systems, or, in any case, systems for which the equilibrium state is Gaussian. We can now introduce a class of perturbations in terms of perturbations of a fixed measure. Under the Klein-Landau correspondence, the perturbed measure is equivalent to a stochastically positive KMS system with state $\Omega_K$, as given in Thm. \ref{thm:interacting-equilibrium-state}.

Let $(\Q,\Sigma, \Sigma_0,U(t),R,\mu)$ be an OS-positive, $\beta-$periodic generalised path space, with a corresponding stochastically positive KMS system $(\gvN,\AvN,\alpha,\omega)$. Let $K$ be a self-adjoint operator on $\Hi_\beta$, affiliated to $\AvN$; under the Klein-Landau isomorphism, $K$ can also be interpreted as a real $\Sigma_0-$measurable function on $\Q$, which we denote by the same symbol.

Assume that $K \in L^1(\Q,\Sigma_0,\mu)$ and $e^{-\beta K} \in L^1(\Q,\Sigma_0,\mu)$. Then, we can define
\[
V \equiv \int_{-\beta/2}^{\beta/2} U(\tau) K \dd \tau \in L^1(\Q,\Sigma,\mu) \ ,
\]
and it is possible to prove (see Ref.~\cite[Proposition 6.2]{Gerard2005} or Ref.~\cite{Klein1981}) that
\[
e^{-V} \in L^1(\Q,\Sigma,\mu) \ .
\]
The family $\{F_{[0,s]}\}_{s \in [0,\beta/2]}$ with
\[
F_{[0,s]} \equiv e^{-\int_0^s U(t) K \dd t} \in L^2(\Q,\Sigma_{[0,\beta/2]},\mu)
\]
is called a \emph{Feynman-Kac-Nelson kernel}.

Therefore, we can define the perturbed probability measure $\mu_V \equiv \left ( \int_\Q e^{-V} \dd \mu \right )^{-1} e^{-V} \dd \mu$. The generalised path space $(\Q,\Sigma,\Sigma_0,U(t),R,\mu_V)$ is OS-positive and $\beta-$periodic \cite{Klein1981}, and we can associate to it a stochastically positive KMS system \cite{Gerard2005,Klein1981}.

\begin{proposition}\label{prop:FKN-kernels}
\emph{\cite[Section 2.3]{Gerard2005b},\cite[Theorem 6.12]{Gerard2005},\cite{Klein1981} }.
Let $(\Q,\Sigma, \Sigma_0,U(t),R,\mu)$ be an OS-positive, $\beta-$periodic generalised path space, let $(\gvN,\AvN,\alpha,\omega)$ be the corresponding stochastically positive KMS system. Let $K \in L^{2+\epsilon}(\Q,\Sigma_0,\mu)$ for $\epsilon >0$, and $e^{-\beta K} \in L^1(\Q,\Sigma_0,\mu)$. Finally, define $V \equiv \int_{-\beta/2}^{\beta/2} U(t) K \dd t$ and
\[
\mu_V \equiv \left ( \int_\Q e^{-V} \dd \mu \right )^{-1} e^{-V} \dd \mu \ .
\]
Then, the operator sum $L+K$ is essentially self-adjoint on $D(L) \cap D(K)$; we denote its closure by the same symbol. Moreover, the OS-positive, $\beta-$periodic path space $(\Q,\Sigma,\Sigma_0,U(t),R,\mu_V)$ is associated, via the Klein-Landau reconstruction theorem, to the stochastically positive KMS system $(\gvN,\AvN,\alpha_K,\omega_K)$. In particular,
\[
\alpha_K^s(A) = e^{is(L+K)} A e^{-is (L+K)} \ , \forall \ A \in \gvN \ ,
\]
and the perturbed KMS state is 
\[
\omega_K(A) \equiv \frac{(\Omega_K, A\Omega_K)}{\norm{\Omega_K}^2} \ ,
\]
where
\[
\Omega_K \equiv e^{-\frac{\beta}{2}(L+K)}\Omega \ .
\]
Moreover, if
\[
K \in L^p(\Q,\Sigma_0,\mu) \ , \quad e^{-\frac{\beta}{2}K} \in L^q(\Q,\Sigma_0,\mu) \ , \frac{1}{p}+\frac{1}{q}= \frac{1}{2} \ , 2 \leq p, \ q \leq \infty \ ,
\]
then, the operator $L+K - JKJ$ is essentially self-adjoint, and it coincides with the Liouvillean for the perturbed KMS system $(\gvN,\alpha_K, \omega_K)$. 
\end{proposition}

The last Proposition gives a quantum-mechanical interpretation of the construction of perturbations of measures with Feynman-Kac-Nelson kernels; indeed, by direct comparison the stochastically positive KMS state $\omega_K$ corresponds exactly to the perturbation theory discussed in Section \ref{sec:perturbations}. A probabilistic interpretation of the Araki relative entropy $S(\omega|\omega_K)$ becomes a direct consequence of the Feynman-Kac formula.

\begin{proposition}[Probabilistic representation of Araki relative entropy]\label{prop:relative-entropy-probabilistic}
Let $(\gvN,\AvN,\alpha,\omega_\beta)$ be the Gaussian stochastically positive KMS system associated to a Gaussian equilibrium state $\omega_\beta$ on a von Neumann algebra $\vN$ acting on $\Hi$, and let $(\Q,\Sigma,\Sigma_0,R,U(\tau),\mu)$ be the corresponding generalised path space. Assume the same hypotheses of Prop. \ref{prop:FKN-kernels}; in particular,
\[
K \in L^p(\Q,\Sigma_0,\mu) \ , \quad e^{-\frac{\beta}{2}K} \in L^q(\Q,\Sigma_0,\mu) \ , \quad \frac{1}{p}+\frac{1}{q}= \frac{1}{2} \ , \  2 \leq p, \ q \leq \infty \ .
\]
Denote
\[
\E(F) \equiv \int_\Q F \dd \mu \ ,
\]
and define $V \equiv \int_{-\beta/2}^{\beta/2} U(\tau) K \dd \tau$.
Then, the Araki relative entropy can be written as
\begin{equation}\label{eq:relative-entropy-probabilistic}
S(\omega|\omega_K) = \E(V) + \log \E(e^{-V}) \ .
\end{equation}
\end{proposition}
\begin{proof}
First, from Eq. \eqref{eq:n-point-functions-probabilistic}, setting $A_1 = K$, $A_{n\geq 2} = 0 $ and $s_{n\geq 1} =0$ we get
\[
\beta \omega(K) = \beta \E(K) \ ,
\]
which, using the invariance of $\E$ under $U(\tau)$ we can rewrite as
\[
 \beta \E(K) = \E(V) \ .
\]
Secondly, for $\tau \in [0,\beta/2]$, we set
\[
\mathcal M^V_\tau \equiv \text{ linear span of } \bigcup_{s \in [0,\beta/2-\tau]} F_{[0,s]} L^\infty(\Q,\Sigma_{[0,\beta/2-\tau]},\mu) \ ,
\]
and
\[
\mathcal D^V_\tau \equiv \mathcal V(\mathcal M^V_\tau) \ .
\]
Then, it is possible to show that
\[
P_K(s) : \mathcal D^K_\tau \ni \mathcal V(\psi) \mapsto \mathcal V(F_{[0,s]}U(s) \psi) 
\]
is a well-defined linear operator, and that $(P_K(\tau), \mathcal D_\tau^K,\beta/2)$ is a local symmetric semigroup on the physical Hilbert space $\Hi_\beta$ \cite[Theorem 7.4]{Gerard2005}. Then, there is a unique self-adjoint operator $H_K$ associated to the local symmetric semigroup, such that $P_K(s) u = e^{-sH_K} u$ for every $u \in \mathcal D_\tau^K$; and if $K \in L^{2+\epsilon}(\Q,\Sigma_0,\mu)$, then $H_K = L+K$ (or, more precisely, its closure). Therefore,
\[
e^{-s(L+K)} (\mathcal V \mathbf 1) = \mathcal V(F_{[0,s]}) \ .
\]

From the physical inner product $(\mathcal V\psi,\mathcal V\psi) = \int_\Q \psi\, R\psi\,d\mu$ we get
\[
\norm{e^{-\frac{\beta}{2}(L+K)}\Omega}^2 = \int_\Q F_{[0,\beta/2]}\cdot R\left (F_{[0,\beta/2]}\right ) \dd \mu \ .
\]
Using $RU(\tau)=U(-\tau)R$ and $RK=K$, it follows that $RF_{[0,\beta/2]} = e^{-\int_{-\beta/2}^0 U(\tau)K \dd \tau}$, and so
\[
F_{[0,\beta/2]}\cdot R(F_{[0,\beta/2]}) = e^{-\int_{-\beta/2}^{\beta/2}U(\tau)K\dd \tau} \ ,
\]
from which we have the Feynman-Kac formula
\begin{equation}
(\Omega, e^{-\beta(L+K)} \Omega) = \int_\Q e^{-V} \dd \mu_C \ .
\end{equation}
Identifying $\norm{e^{-\frac{\beta}{2}(L+K)} \Omega}^2 = (\Omega, e^{-\beta(L+K)} \Omega)$ the statement follows.
\end{proof}

\section{The Euclidean massive free field at positive temperature}\label{sec:Euclidean-free-field}

\subsection{Relativistic bosonic field}
Before moving to the derivation of the Polchinski equation, we give a concrete probabilistic representation of a Gaussian KMS system, in terms of the \emph{free Euclidean massive field at finite temperature}. The motivating example is the study of perturbations of the massive Klein-Gordon field on Minkowski spacetime; the associated probabilistic model is the free Euclidean massive field.

Let $\mathcal M=(\mathbb R^d,\eta)$ be $d-$dimensional Minkowski spacetime with a Cauchy surface $\Cauchy= \{ \x \in \mathcal M \ | \ x^0 = 0 \}$, isometric to $\mathbb R^{d-1}$. Let $\mathscr L =C^\infty_c(\Cauchy) \oplus C^\infty_c(\Cauchy)$ be the space of real, smooth, and compactly supported Cauchy data on $\Cauchy$. By the Leray theorem, the Cauchy data $(q,p) \in \mathscr L$ define a unique solution $\phi$ to the massive Klein-Gordon equation
\[
(\partial^2_t -\Delta_x + m^2)\phi = 0 \ ,
\]
with spatially compact support, such that $\phi|_\Cauchy = q$ and $\partial_t \phi|_\Cauchy = p$. Here $\Delta_x$ denotes the Laplacian on $\mathbb R^{d-1}$. 
The symplectic form is
\[
\sigma((q_1,p_1),(q_2,p_2))=\int_{\Cauchy}(q_1 p_2 - q_2 p_1) \dd \Cauchy \ .
\]
In order to consider the Abelian splitting defined in Section \ref{sec:quasi-free-systems}, we need to identify a Hilbert space $(\mathfrak h, ( \cdot , \cdot )_{\mathfrak h})$ containing $\mathscr L$. Consider the operator $\h \equiv (-\Delta_x+m^2)^{\frac{1}{2}}$ on $L^2(\mathbb R^{d-1},\dd x)$. In this case, a canonical construction gives the real-linear map $\K:\mathscr L \to L^2(\Cauchy)$, defined by $\K(q,p) \equiv \h^{1/2}q+i \h^{-1/2}p$. The completion of $\K\mathscr L$ with respect to the $L^2$ inner product gives the one-particle Hilbert space $\mathfrak h \equiv L^2(\Cauchy, \dd x)$. It is immediate to see that $\Im( \K(q_1,p_1),\K(q_2,p_2))_{\mathfrak h} = \sigma((q_1,p_1),(q_2,p_2))$. 

The classical Klein-Gordon evolution is the one-parameter group
\[
T^s \equiv 
\begin{pmatrix}
\cos(s\h) &\h^{-1}\sin(s\h) \\
-\h\sin(s\h) &\cos(s\h)
\end{pmatrix} \ ,
\]
and again by direct computation $\K T^s = e^{-i s \h} \K$, so that $\h$ corresponds to the one-particle Hamiltonian introduced in Section~\ref{sec:Weyl}.

Now, consider the time reversal operator $\theta$ on compactly supported, smooth solutions, acting by $\theta \phi(t,x) = \phi(-t,x)$; on Cauchy data, this acts as $\theta (q,p) = (q,-p)$, so that
\[
\K\theta (q,p) = \h^{1/2}q - i \h^{-1/2}p \ .
\] 
Therefore, on the one-particle Hilbert space $\mathfrak h$, time reversal is realised as the complex conjugation $\Theta : \mathfrak h \ni f \mapsto \overline f$.

Now, the Weyl $C^*-$algebra is $\mathcal W(\mathfrak h, \sigma)$, and the unique Gaussian KMS state $\omega_\beta$ at inverse temperature $\beta$ has covariance \cite[Theorem 3.1]{Gielerak1998}
\[
H(f,g) =\frac12 \Re \left ( f, \left( 1+2\rho \right) g \right )_{\mathfrak h} \ ,
\]
with $\rho \equiv (e^{\beta \h}-1)^{-1}$.

There is a convenient GNS representation of $(\mathcal W(\mathfrak h, \sigma), \omega_\beta)$, known as the \emph{Araki-Woods representation} \cite{ArakiWoods1963}. 

This gives a von Neumann algebra $\vNb = \pi_\omega(\CCR)''$, generated by $\{ W_\omega(f) \ | \ f \in \mathfrak h\}$. The Abelian splitting determined by $\Theta$ is $\mathfrak h_\pm = \{ f \in \mathfrak h \ | \ \Theta f = \pm f \}$, and the Abelian subalgebra is $\vNb^+$, generated by $\{ W_\omega(f) \ | \ f \in \mathfrak h_+\}$. 

In terms of Cauchy data, we have
\[
\mathfrak h_+ = \{ f \in  \mathfrak h \ | \ f = \h^{\frac{1}{2}} q \ , \ q \in C_c^\infty(\Cauchy)\} \ .
\]
In other words, $\mathfrak h_+$ is the Hilbert space corresponding to one-particle elements with vanishing initial momentum.

\subsection{Massive Euclidean free field}
Applying the construction explained in Section \ref{sec:quasi-free-systems}, from $\sigma$ and $H$ we get the positive-definite bounded operator
\[
R_\beta(s) = \frac{e^{-s \h}+e^{-(\beta-s)\h}}{1-e^{-\beta \h}} \ , \qquad 0 \leq  s \leq \beta \ ,
\]
which can be extended periodically to the whole real line, and, from Thm. \ref{thm:probabilistic-Weyl-algebra}, we get a Gaussian thermal process $\{\xi_s^\beta\}_s$ indexed on $\mathfrak h_+$, with covariance $\frac12 R_\beta(s)$.

In order to construct the Euclidean free field, consider the space of functions $f : \mathbb S_\beta \ni \tau \mapsto f(\tau) \in \mathfrak h_+$, and define the covariantly smeared random field
\[
\Xi(f) \equiv \int_{\mathbb S_\beta} \langle \xi_\tau^\beta,f(\tau) \rangle \dd \tau \ .
\]
This satisfies
\[
\E\left [ \Xi(f) \Xi(g) \right ] = \frac{1}{2}\int_{\mathbb S_\beta^{\otimes 2}} \dd \tau_1 \dd \tau_2 \langle f(\tau_1),R_\beta(\tau_1-\tau_2) g(\tau_2) \rangle \ .
\]
Now, recalling that $f \in \mathfrak h_+$ is in the form $f = \h^{\frac{1}{2}}q$ for $q \in C^\infty_c(\Cauchy)$, we get that the random field $\varphi(f) \equiv \Xi(\h^{-1/2}f)$, where $\h^{-1/2}$ acts on the spatial variable, is indexed by $\mathscr D(\mathbb S_\beta \times \mathbb R^{d-1}) \equiv C^\infty_c(\mathbb S_\beta \times \mathbb R^{d-1})$ with covariance
\[
\E\left [ \varphi(f) \varphi(g) \right ] =\int_{\mathbb S_\beta^{\otimes 2}} \dd \tau_1 \dd \tau_2 \frac{1}{2\h} \langle f(\tau_1),R_\beta(\tau_1-\tau_2) g(\tau_2) \rangle \ .
\]
More explicitly, the covariance is a bilinear map $C : \mathscr D(\mathbb S_\beta \times \mathbb R^{d-1}) \times \mathscr D(\mathbb S_\beta \times \mathbb R^{d-1}) \to \mathbb R$ with integral kernel
\begin{equation}\label{eq:thermal-covariance}
C(\tau,x) \equiv \frac{1}{2\h}\frac{e^{-\tau \h}+e^{-(\beta-\tau)\h}}{1-e^{-\beta \h}} \ .
\end{equation}

This is exactly the integral kernel of the Green's function of the Laplacian on the thermal cylinder $\Mb \equiv \mathbb S_\beta \times \mathbb R^{d-1}$, where $\mathbb S_\beta=\mathbb R/\beta\mathbb Z$,
\[
\langle f, C g \rangle = \langle f, \left( -\partial_\tau^2-\Delta_x+m^2
\right)^{-1}  g \rangle \ , \quad f, \ g \in \mathscr D(\mathbb S_\beta \times \mathbb R^{d-1}) \ ,
\]
with periodic boundary conditions in $\tau$. Therefore, the Gaussian KMS Weyl system $(\vNb,\vNb^+,\alpha,\omega_\beta)$ is equivalent to the random Gaussian process with zero mean and covariance $C: (f, g) \in \mathscr D(\Mb)\times \mathscr D(\Mb) \mapsto \langle f ,(-\partial_\tau^2 - \Delta_x + m^2)^{-1} g \rangle$.

By the Minlos theorem, this can be realised as a Borel measure $\mu_C$ on the topological dual $\mathscr D'(\Mb)$, defined by the characteristic function
\[
\int_{\mathscr D'(\Mb)} e^{i\varphi(f)} \dd \mu_C = e^{-\frac{1}{2}C(f,f)} \ , \forall \ f \in \mathscr D(\Mb) \ .
\]

The probability space then is $\Q= \mathscr D'(\Mb)$ with Gaussian measure $\mu_C$. Moreover, the probability space $(\mathscr D'(\Mb), \Sigma, \mu_C)$ supports a generalised path space. First, it is possible~\cite{Gerard2005b} to construct sharp-time random fields
\[
\varphi(\tau,f) \equiv \lim_{k\to \infty}\varphi(\delta_k(\cdot-\tau) \otimes f) \ ,
\]
where $\delta_k$ is a family of approximants of the Dirac delta on $\mathbb S_\beta$ and $f \in L^2(\mathbb R^{d-1})$. Then, $\Sigma_0$ is generated by the functions $\{ \varphi(0,f) \ | \ f \in \mathscr D(\mathbb R^{d-1}) \}$.

Similarly, we see that
\begin{multline*}
C_0(\tau_1,\tau_2,f_1,f_2) \equiv \lim_{k\to \infty} C\left (\delta_k(\cdot - \tau_1) \otimes f_1,\delta_k(\cdot - \tau_2) \otimes f_2 \right ) \\
= \left ( f_1,\frac{1}{2h}\frac{e^{-\abs{\tau_2-\tau_1}h}+e^{-(\beta-\abs{\tau_2-\tau_1})h}}{1-e^{-\beta h}}  f_2 \right)_{L^2(\mathbb R^{d-1})} 
\end{multline*}
gives a well-defined sharp-time covariance. In particular, under the Klein-Landau correspondence, $e^{i\varphi(0,q)}$ is identified with $W_\omega(h^{-1/2}q)$ for real $q \in C^\infty_c(\mathbb R^{d-1})$.

Now, denoting with $\iota_s:(\tau,x) \mapsto (\tau+s,x)$ the map induced on $\Q$ by the time translations, the time evolution is the one-parameter group which, on functions $F: \Q \to \mathbb C$, acts by
\[
U(\tau) F(\varphi) = F(\iota_{-\tau} \varphi) \ ,
\]
so that in particular
\[
U(\tau) \varphi(0,f) = \varphi(\tau,f) \ .
\]

Then, it is possible to prove \cite[Section 4.2]{Gerard2005b} that $\Sigma = \bigvee_{\tau \in \mathbb S_\beta} \Sigma_\tau$. Finally, if $r$ is the Euclidean time reflection around $\tau = 0$ and $R$ is the measure-preserving transformation of $(\mathscr D'(\mathbb S_\beta \times \mathbb R^{d-1}), \Sigma, \mu_C)$ generated by $r$, the generalised path space associated to the Euclidean free field equivalent to the relativistic bosonic field is $(\mathscr D'(\mathbb S_\beta \times \mathbb R^{d-1}),\Sigma,\Sigma_0,U(\tau),R,\mu_C)$.

\begin{remark}
Thanks to the regularity properties of the covariance, it is actually possible to extend the index space to $\mathscr S(\mathbb S_\beta\times\mathbb R^{d-1})$, the space of smooth periodic functions in $\tau$ which are Schwartz in $\mathbb R^{d-1}$. Then, the Gaussian random field becomes a tempered distribution $\varphi \in \mathscr S'(\mathbb S_\beta\times\mathbb R^{d-1})$.
\end{remark}

\section{Abstract Wiener space for the massive Euclidean free field}\label{sec:AWS}

The distribution space $\mathscr D'(\Mb)$ given by the Minlos theorem is much larger than what is actually needed: the support of the covariance, and hence of the measure, can be restricted by careful considerations on the decay properties and integrability of $C$. This has been done for the Euclidean free field on $\mathbb R^d$ by Reed and Rosen \cite{Reed1974}. Below, we will give a weighted Sobolev space $\Ba$ such that $\mu_C(\Ba) = 1$. This gives the massive Euclidean free field the structure of an \emph{abstract Wiener space}, which in turn provides the right framework to derive the Polchinski equation in infinite dimensions.

We begin with a brief summary of the definition of abstract Wiener spaces, following the original formulation by Gross \cite{Gross1967}. 

Consider a separable Hilbert space $(\Hi, ( \ \cdot \ , \ \cdot \ )_\Hi)$, with norm $\norm{ \ \cdot \ }_\Hi$. 
A norm $\mnorm{ \ \cdot \ }$ on $\Hi$ is \emph{measurable} if, for every $\epsilon >0$, there is a finite-dimensional subspace $E_\epsilon$ such that, for every finite-dimensional subspace $E \subset E_\epsilon^\perp$, $\mu_E(\{x \in E \ | \ \mnorm{x} > \epsilon \})< \epsilon$. Here, $\mu_E$ is the standard Gaussian measure on $E$ with variance $1$:
\[
\dd \mu_E(x) = (2\pi)^{-k/2} e^{-\frac{\abs{x}^2}{2}} \dd x \ ,
\]
where $k$ is the dimension of $E$ and $\dd x$ is the Lebesgue measure in $E$.

It is a consequence of the definition that there is a constant $a$ such that \cite{Gross1967potential}
\[
\mnorm{x} \leq a \norm{x}_\Hi \ , \ \forall \ x \in \Hi \ .
\]

A useful class of measurable norms is given by Hilbert-Schmidt operators \cite[Proposition 2.2]{Sheffield2007}.
\begin{lemma} \label{lemma:measurable-norm}
\emph{\cite[Example 1]{Gross1967}}.
If $T$ is a Hilbert-Schmidt operator on $\Hi$, then the norm $\norm{T \ \cdot \ }_\Hi$ is measurable.
\end{lemma}

Denote by $\Ba$ the Banach space completion of $\Hi$ with respect to a measurable norm $\mnorm{ \ \cdot \ }$ and with $\Ba^*$ the space of continuous linear functionals on $\Ba$. Since each element in $\Ba^*$ can be identified with a continuous linear functional on $\Hi$, thanks to the canonical identification $\Hi^* \simeq \Hi$ we have the chain of inclusions
\[
\Ba^* \subset \Hi^* \simeq \Hi \subset \Ba \ .
\]

In the larger Banach space $\Ba$ it is possible to identify a Gaussian measure. Indeed,  we have the following theorem \cite{Gross1967,Sheffield2007}. 
\begin{theorem}[Gross 1967]\label{thm:gross}
\emph{\cite[Theorem 2.4]{Sheffield2007}}.
If $\mnorm{ \ \cdot \ }$ is a measurable norm, then there is a unique probability measure $\p$ on $\Ba$ such that, if $\varphi$ is a random variable with probability measure $\p$, then, for every $f \in \Ba^*$, the random variable $f(\varphi)$ is a one-dimensional Gaussian with zero mean and variance $\norm{f}^2_\Hi$; that is,
\[
\int_\Ba e^{if(\varphi)} \dd \p( \varphi) = e^{-\frac{1}{2}\norm{f}_\Hi^2} \ .
\]
\end{theorem}

\begin{definition}[Abstract Wiener space]
An abstract Wiener space is a triple $(\Hi, \Ba, \p)$ consisting of a Hilbert space $\Hi$, densely and continuously embedded in a Banach space $\mathcal B$ via an inclusion $\iota:\mathcal H \to \mathcal B$, and a measure $\p$ such that
\[
\int_\Ba e^{if(x)} \dd \p( x) = e^{-\frac{1}{2}\norm{f}_\Hi^2} \ .
\]
\end{definition}

\subsection{Cameron-Martin space}
In order to accommodate the massive Euclidean free field in the framework of abstract Wiener spaces, we follow Gross' procedure of identifying a suitable Hilbert space and then complete it to a Banach space, such that $\Ba \subset \mathscr D'(\Mb)$. We will then show that the massive Euclidean free field measure $\mu_C$ satisfies $\mu_C(\Ba)=1$, so that the random field is almost surely in $\Ba$.

First, the Hilbert space $\mathcal H$ is the Cameron-Martin space for $(\mathscr D'(\Mb),\Sigma,\mu_C)$. This is canonically associated with the covariance $C$. By the Schwartz kernel theorem,  the covariance as a bilinear map $C: \mathscr D(\Mb) \times \mathscr D(\Mb) \to \mathbb R$ can be realised as an operator $C : \mathscr D(\Mb) \to \mathscr D'(\Mb)$, so that $C(f,g) = \langle f, Cg\rangle$, where the angle brackets denote the standard pairing between the test function space and the distribution space. Then, we can equip $C \mathscr D$ with the product and associated norm
\[
(Cf,Cg)_{\Hi_C} \equiv \langle f, C g \rangle \ , \quad \norm{Cf}_{\Hi_C}\equiv \sqrt{(Cf,Cf)}=\sqrt{C(f,f)} \ .
\]
\begin{definition}[Cameron-Martin space]
The \emph{Cameron-Martin space} is the Hilbert space completion
\[
\Hi_C \equiv \overline{C\mathscr D}^{\norm{ \ \cdot \ }_{\Hi_C}} \ .
\]
\end{definition}

\begin{proposition}\label{prop:Cameron-Martin-bound}\emph{\cite[Proposition 3.32]{Hairer2026}}
For any $h \in \mathscr \Hi_C$, there is a positive constant $A$ such that
\[
\abs{\langle h,f\rangle} \leq A \sqrt{C(f,f)} \ , \quad \forall f \in \mathscr D(\Mb) \ .
\]
\end{proposition}
\begin{proof}
Since $C: \mathscr D \to \mathscr D'$, there is a canonical continuous injection $i_C:\Hi_C \to \mathscr D'$.
For every $f \in \mathscr D$ this is given by 
\[
\langle i_C h, f\rangle = ( h, Cf)_{\Hi_C} \ .
\]
By Cauchy-Schwarz, we then have
\[
\abs{\langle i_C h, f\rangle} \leq \norm{h}_{\Hi_C} \norm{C f}_{\Hi_C} = \norm{h}_{\Hi_C} \sqrt{C(f,f)} \ .
\]

\end{proof}
The subspace $i_C(\Hi_C)$ identifies the Cameron-Martin directions leaving the Gaussian measure $\mu_C$ quasi-invariant, by the Cameron-Martin theorem.

There is an equivalent realisation of the Cameron-Martin space based on the triple $\mathscr D(\Mb) \subset L^2(\Mb) \subset \mathscr D'(\Mb)$. Suppose $C$ is represented by a bounded positive operator $C: L^2(\Mb) \to L^2(\Mb)$. By the spectral theorem we can construct $C^{1/2}:L^2(\Mb) \to L^2(\Mb)$. Consider $f \in \mathscr D(\Mb) \subset L^2(\Mb)$; then,  
\[
\norm{Cf}_{\Hi_C}^2 = C(f,f) = (C^{1/2}f,C^{1/2}f) = \norm{C^{1/2}f}_{L^2(\Mb)}^2 \ .
\]
This implies that the map $C^{-1/2}: Cf \mapsto C^{1/2}f$ is an isometry from the dense subspace $C\mathscr D(\Mb) \subset \Hi_C$ into $L^2(\Mb)$.

Since $C$ is injective, $\Ran C^{1/2}$ is dense in $L^2(\Mb)$, and $C^{-1/2}$ extends uniquely to a unitary
\[
C^{-1/2}: \Hi_C \to L^2(\Mb) \ .
\]
Its inverse uniquely identify the Cameron-Martin space with $C^{1/2}L^2(\Mb)$, equipped with the product
\[
(h,k)_{\Hi_C} = (C^{-1/2}h,C^{-1/2}k)_{L^2(\Mb)}
\]
and the norm
\[
\norm{h}_{\Hi_C} = \norm{C^{-1/2}h}_{L^2(\Mb)} \ , \quad h = C^{1/2} u \ , \quad u \in L^2(\Mb) \ .
\]
Equivalently, as reference for a computation in the next section, we see that
\[
\norm{C^{1/2}u}_{\Hi_C} = \norm{ u }_{L^2(\Mb)} \ .
\]

In the case of the massive Euclidean free field, the covariance is $C = (-\partial_\tau^2 - \Delta_x + m^2)^{-1}$, and the Cameron-Martin space is the Sobolev space $\Hi_C = H^1_m(\Mb)$, with norm
\[
\norm{h}_{\Hi_C}^2 = \int_{\Mb} \left (\abs{\partial_\tau h}^2 + \abs{\nabla_x h}^2 + m^2 \abs{h}^2 \right ) \dd\tau \dd x \ .
\]
Thanks to the previous discussion, this is equivalent $H^1_m(\Mb) = C^{1/2}L^2(\Mb)$, with inner product
\[
(h,k)_{\Hi_C} =( C^{-1/2}h,C^{-1/2}k)_{L^2(\Mb)} \ .
\]

\subsection{Banach space completion}
The identification of the Banach space $\Ba$ is non-canonical, and it requires to find a suitable measurable norm on $\Hi_C$. 

The Banach space $\Ba$ will be a suitably weighted Sobolev space. We define the weight
\[
w_\rho:\mathbb R^{d-1} \ni x \mapsto w_\rho(x) \equiv (1+\abs{x}^2)^{-\rho/2} \ ,
\]
with $\rho >\frac{d-1}{2}$, so that $w_\rho(x) \in L^2(\mathbb R^{d-1})$. Now, fix an $a>\frac{d-2}{4}$ and introduce the space
\[
\mathcal L_{a,\rho} \equiv \left \{u \in \mathscr D'(\Mb) \ | \ \mnorm{u}_{a,\rho} < \infty \right \} \ ,
\]
where the norm is
\[
\mnorm{u}_{a,\rho}\equiv \norm {w_\rho(x) C^{a} u}_{L^2(\Mb)} \ .
\]
The Cameron-Martin space then is $\Hi_C = \mathcal L_{-1/2,0}$.

\begin{proposition}\label{prop:measurable-norm}
If  $a>\frac{d-2}{4}$ and  $\rho >\frac{d-1}{2}$, then $\norm{\ \cdot \ }_{a,\rho}$ is a measurable norm on $\Hi_C$.
\end{proposition}
\begin{proof}
Introduce the generalised eigenfunctions
\[
e_{n,k}(\tau,x) \equiv \beta^{-1/2}(2\pi)^{-(d-1)/2}e^{i(\omega_n \tau+ k \cdot x)} \ , \ \omega_n = \frac{2\pi}{\beta}n \ , \ k \in \mathbb R^{d-1}  \ ,
\]
with $k \cdot x$ the Euclidean product in $\mathbb R^{d-1}$. They are a generalised orthonormal system for $L^2(\Mb)$, thanks to the unitary Fourier decomposition $L^2(\Mb) \simeq \ell^2(\mathbb Z) \otimes L^2(\mathbb R^{d-1})$.
These satisfy $C^{-1} e_{n,k}(\tau,x)  = E_{n,k}e_{n,k}(\tau,x)$, with $E_{n,k} = \omega_n^2 + k^2 + m^2$. Thanks to the identification $\Hi_C = C^{1/2}L^2(\Mb)$, the functions $f_{n,k} \equiv C^{1/2}e_{n,k}$, explicitly given by $f_{n,k} \equiv E_{n,k}^{-1/2} e_{n,k}$ are a generalised orthonormal system in $\Hi_C$. Their norm is
\[
\norm{f_{n,k}}_{a,\rho}^2=E_{n,k}^{-1-2a} \norm{ w_\rho e_{n,k}}_{  L^2(\Mb) }^2 \ .
\]
Since $\abs{e_{n,k}(z)}^2=\beta^{-1}(2\pi)^{1-d}$ this simplifies into
\[
\norm{f_{n,k}}_{a,\rho}^2=E_{n,k}^{-1-2a} \kappa_{d-1}(\rho) \ ,
\]
with
\[
\kappa_{d-1}(\rho) \equiv (2\pi)^{1-d}\int_{\mathbb R^{d-1}}(1+\abs{x}^2)^{-\rho} \dd x \ .
\]

Now, define the operator
\[
T_{a,\rho} \equiv C^{1/2}w_\rho C^a : \Hi_C \to \Hi_C \ ,
\]
(first defined on the dense domain $C^\infty_c(\Mb)$, and then closed in $\Hi_C$).
Then, $\norm{Tf}_{\Hi_C} = \norm{f}_{a,\rho}$. Indeed,
\[
\norm{Tf}_{\Hi_C} = \norm{C^{-1/2}Tf}_{L^2(\Mb)} = \norm{w_\rho(x) C^a f}_{L^2(\Mb)} = \norm{f}_{a,\rho} \ .
\]
Its Hilbert-Schmidt norm is
\[
\norm{T}_{HS(\Hi_C)}^2 = \sum_n \int_{\mathbb R^{d-1}} \dd k \norm{T f_{n,k}}^2_{\Hi_C}=\sum_n \int_{\mathbb R^{d-1}} \dd k \norm{f_{n,k}}_{a,\rho}^2 \ .
\]
Using the computation of the norm $\norm{f_{n,k}}_{a,\rho}$, this is
\[
\norm{T f_{n,k}}_{HS(\Hi_C)}^2  = \kappa_{d-1}(\rho)\sum_n \int_{\mathbb R^{d-1}}\dd k E_{n,k}^{-1-2a} \ .
\]
Now, $\kappa_{d-1}(\rho)$  is finite, since $\rho>\frac{d-1}{2}$ by assumption. Moreover, if $2(1+2a)>d-1$, or equivalently $a > \frac{d-3}{4}$, the integral in $k$ is convergent  and can be computed explicitly, giving
\[
\norm{T}_{HS(\Hi_C)}^2  = \kappa_{d-1}(\rho) \pi^{\frac{d-1}{2}} \frac{ \Gamma\left(1+2a-\frac{d-1}{2}\right)} {\Gamma(1+2a)} \sum_{n \in \mathbb Z} \left[ \left(\frac{2\pi n}{\beta}\right)^2+m^2 \right]^{\frac{d-1}{2}-(1+2a)} \ .
\]
The series is convergent if $2+4a >d$, which is equivalent to the assumption $a > \frac{d-2}{4}$. Note that this condition dominates over the condition on the convergence of the integral.

Then, $T$ is a Hilbert-Schmidt operator, and by Lemma \ref{lemma:measurable-norm} it follows that $\mnorm{ \ \cdot \ }_{a,\rho}$ is a measurable norm on $\Hi_C$.
\end{proof}

\begin{corollary}[Abstract Wiener space for the Euclidean massive free field]
The triple $(\Hi_C, \mathcal L_{a,\rho},p)$ is an abstract Wiener space.
\end{corollary}
\begin{proof}
Since $\mathcal L_{a, \rho}$ is the completion of $\Hi_C$ with respect to the measurable norm $\mnorm{ \ \cdot \ }_{a,\rho}$, by Theorem \ref{thm:gross} it carries a unique Gaussian measure $\p$, and the triple $(\Hi_C, \mathcal L_{a,\rho},\p)$ is an abstract Wiener space. 
\end{proof}

We can actually prove that $\mathcal L_{a,\rho}$ is a Hilbert space. The additional Hilbert space structure will be useful in deriving the Polchinski equation in the next section.
\begin{lemma}\label{lemma:Hilbert-space}
Let $J_{a,\rho}: \mathcal L_{a,\rho} \to L^2(\Mb)$ be defined by
\[
J_{a,\rho} u \equiv w_\rho C^{a} u \ .
\]
Then, $J_{a,\rho}$ is a linear isometric isomorphism, with inverse
\[
J_{a,\rho}^{-1} g = C^{-a}\left( w_\rho^{-1} g \right ) \ , \quad g \in L^2 \left (\Mb \right) \ .
\]
Therefore, in particular $\mathcal L_{a,\rho}$ is a separable Hilbert space with inner product
\[
(u,v)_{\mathcal L_{a,\rho}} \equiv (u,w_\rho C^{a} v)_{L^2(\Mb)} \ .
\]
\end{lemma}
\begin{proof}
First, since $C^a$ and all its derivatives have Fourier symbols with at most polynomial growth, it defines a continuous linear isomorphism on $\mathscr S \subset L^2(\Mb)$, with inverse $C^{-a}$. By duality, it extends to a continuous linear isomorphism $C^a : \mathscr S'(\Mb) \to \mathscr S'(\Mb)$, again with inverse $C^{-a}$.

By definition, elements in $\mathcal L_{a,\rho}$ satisfy $J_{a,\rho} u \in L^2(\Mb)$, with norm $\norm{u}_{a,\rho} = \norm{J_{a,\rho} u}_{L^2(\Mb)}$, and so $J_{a,\rho}$ is an isometry.

Now we consider its inverse. First, since $w_\rho^{-1}$ has only polynomial growth, it can be regarded as a tempered distribution. Hence, for every $g \in L^2(\Mb)$, $w_\rho^{-1} g \in \mathscr S'(\Mb)$, and we can consider the distribution $u = C^{-a}\left(w_\rho^{-1} g \right ) \in \mathscr S'(\Mb)$. Defining $J_{a,\rho}^{-1} g = C^{-a}\left(w_\rho^{-1} g \right )$, it is obvious that $J_{a,\rho} u = g \in L^2(\Mb)$, and so $J_{a,\rho}$ is surjective. Therefore, we have a continuous linear isomorphism $\mathcal L_{a,\rho} \simeq L^2(\Mb)$, and since $ L^2(\Mb)$ is a separable Hilbert space, also $\mathcal L_{a,\rho}$ is.
\end{proof}

Now, we conclude the proof that $\mathcal L_{a,\rho}$ supports the Euclidean massive free field with probability one.
\begin{theorem}
For any $\rho >\frac{d-1}{2}$ and $a >\frac{d-2}{4}$, $\mathcal L_{a,\rho}(\Mb)$ is a set of measure one for the Gaussian process over $\mathscr D(\Mb)$ with covariance $\langle f, C g\rangle$.
\end{theorem}

\begin{proof}
By the Minlos theorem, the Gaussian process can be realised as a probability space on $\mathscr D'(\Mb)$ with Gaussian measure $\mu_C$. 

We apply the proof strategy by Reed and Rosen~\cite{Reed1974} to the present case. 
First, define
\[
\norm{f}_{-1}^2 \equiv ( f, C f)_{L^2(\Mb)} \ , \quad \mathscr H_{-1} \equiv \overline{\mathscr D(\Mb)}^{\norm{ \ \cdot \ }_{-1}} \ .
\]
The covariance map $C$, first defined on the dense subspace $\mathscr D(\Mb)$ and extended by continuity, is an isometry from  $\mathscr H_{-1}$ to the Cameron-Martin space. Indeed, $\norm{f}_{\mathscr H_{-1}}^2 = \norm{C^{1/2} f}_{L^2(\Mb)} = \norm{Cf}_{\Hi_C}^2$, so $C:\mathscr H_{-1} \to \Hi_C$ is a unitary operator.

Secondly, define on $\mathscr D$ the inner product and norm
\[
(f,g)_{a,\rho,*} \equiv (w_\rho^{-1}C^{-a}f,w_\rho^{-1}C^{-a}g)_{L^2(\Mb)} \ , \quad \norm{f}_{a,\rho,*} \equiv \norm{w_\rho^{-1} C^{-a} f}_{L^2(\Mb)} \ ,
\]
and consider the Hilbert space completion
\[
\mathscr H_{a,\rho,*} \equiv \overline{\mathscr D}^{\norm{\ \cdot \ }_{a,\rho,*}} \ .
\]
On $\mathscr H_{-1}$, consider the operator $\tilde T_{a,\rho} \equiv C^a w_\rho C^{1/2}$. Its inverse, initially defined on $\mathscr D(\Mb)$, is
\[
\tilde T_{a,\rho}^{-1} = C^{-1/2}w_\rho^{-1} C^{-a} \ ,
\] 
so that
\[
\norm{\tilde T_{a,\rho}^{-1}f}_{-1}^2 = \norm{w_\rho^{-1} C^a f}_{L^2(\Mb)} = \norm{f}_{a,\rho,*} \ .
\]
We want to show that $\tilde T_{a,\rho}$ is Hilbert-Schmidt on $\mathscr H_{-1}$. Under the unitary map $C:\mathscr H_{-1} \to \Hi_C$, the operator is represented in $\Hi_C$ by
\[
C \tilde T_{a,\rho} C^{-1} =  C^{a+1} w_\rho C^{-1/2} = C( C^a w_\rho C^{1/2}) C^{-1} = T_{a,\rho}^* \ ,
\]
where the adjoint is in the Cameron-Martin space. Since $T_{a,\rho}$ is Hilbert-Schmidt, by Prop.~\ref{prop:measurable-norm}, then also $\tilde T_{a,\rho}$ is Hilbert-Schmidt.

By the Minlos support theorem~\cite{Reed1974}, this implies that $\mathscr H_{a,\rho,*}^*$ is a subset of $\mathscr D'(\Mb)$ with measure $1$ for the Gaussian process $\varphi(f) \equiv \langle \ \cdot \ , f \rangle$ over $\mathscr D(\Mb)$, with covariance $\E(\varphi(f) \varphi(g))=\langle f,Cg\rangle$.

We now identify $H_{a,\rho,*}'$ with $\mathcal L_{a,\rho}$. If $u\in\mathcal L_{a,\rho}$, then for every $f\in\mathscr D(\Mb)$,
\[
\langle u,f\rangle = \left(w_\rho C^a u,w_\rho^{-1}C^{-a}f\right)_{L^2(\Mb)} \ ,
\]
and, therefore,
\[
\abs{\langle u,f\rangle} \leq \norm{u}_{a,\rho} \norm{f}_{a,\rho,*} \ .
\]
Thus, $u$ defines a continuous linear functional on $H_{a,\rho,*}$.

Conversely, let $u\in\mathscr D'(\Mb)$ define a continuous functional on $H_{a,\rho,*}$. By the Riesz representation theorem there exists $g\in L^2(\Mb)$ such that
\[
\langle u, f \rangle = \left( g , w_\rho^{-1} C^{-a} f \right)_{L^2(\Mb)} \ , \quad f \in \mathscr D(\Mb) \ .
\]
Thus, $u$ is precisely the distribution represented by
\[
u=C^{-a}\left ( w_\rho^{-1} g \right) ,
\]
and, therefore,
\[
w_\rho C^a u = g \in L^2(\Mb) \ .
\]
Hence $u\in\mathcal L_{a,\rho}$, and so
\[
H_{a,\rho,*}'=\mathcal L_{a,\rho} \ ,
\]
as subspaces of $\mathscr D'(\Mb)$. It follows from Minlos' support theorem that
\[
\mu_C(\mathcal L_{a,\rho})=1 \ .
\]
\end{proof}

\begin{remark}
In the abstract Wiener space setting, the covariance $Q_\p$ is a bilinear map $Q_\p: \Ba^* \times \Ba^* \to \mathbb R$ defined by
\[
Q_\p(\ell,\ell') = \int_\Ba \ell(\varphi) \ell'(\varphi) \dd \p(\varphi)  = (\ell, \ell')_{\Hi_C} \ .
\]
Now, every $f \in \mathscr S(\Mb)$ determines a unique linear functional $\ell_f \in \Ba^*$; restricting to linear functionals of this form, we have
\[
Q_\p(\ell_f,\ell_g) = (\ell_f,\ell_g)_{\Hi_C} = (Cf, C g)_{\Hi_C} = \langle f, C g \rangle \ .
\]
Therefore, the covariance $Q_\p$ associated to the measure $\p$ is the natural generalisation of the covariance $C$, defined on the space of test functions, to the dual $\Ba^*$. By Fernique's theorem, we can associate in a canonical way the operator $Q_\p : \Ba^* \to \Ba$ as $Q_\p(\ell) = \int_\Ba \varphi \ell(\varphi) \dd p(\varphi)$ (see, e.g., Ref. \cite[Corollary 3.15]{Hairer2026}). This implies, in particular, that for any $\ell' \in \Ba^*$ we have
\[
\ell'(Q_\p(\ell)) = \ell'\left[\E\left(\varphi \ell(\varphi)\right ) \right ] = \E\left[\ell(\varphi) \ell'(\varphi)\right ] = Q_\p(\ell,\ell') \ .
\]
\end{remark}

\section{Polchinski equation and the multiscale Bakry-Émery criterion}

In this section we introduce the Polchinski equation as a way to decompose  in continuous scales the relative entropy in Eq.~\eqref{eq:relative-entropy-probabilistic} for perturbations of the massive Euclidean free field.
This is done by a continuum decomposition of the Gaussian measure. This RG flow decomposition takes the form of the \emph{Polchinski equation} \cite{Polchinski1983}.

Using the Polchinski equation, Bauerschmidt, Bodineau, and Dagallier derived a  multiscale Bakry-Émery criterion to prove log-Sobolev inequalities for classical field theories and spin systems \cite{BauerschmidtBodineau2020,Bauerschmidt2024,Bauerschmidt2023b,Bauerschmidt2023}. In the following, we show that the same technique implies the relative entropy bound in Thm. \ref{thm:entropy-inequality}.

The main difference from Bauerschmidt, Bodineau, and Dagallier's approach is the infinite-dimensional setting of QFT, arising by the fact that the massive Euclidean free field describes (perturbations of) the relativistic boson field. For this reason, we give a new derivation of the Polchinski equation in abstract Wiener spaces, which may be of independent interest.

\subsection{General assumptions and covariance decomposition}\label{sec:assumptions}

We consider the massive relativistic bosonic field on Minkowski spacetime, described by the von Neumann algebra $\vNb$ with abelian subalgebra $\vNb^+$, in a $(\alpha,\beta)-$KMS Gaussian state $\omega$. This gives a stochastically positive KMS system $(\vNb, \vNb^+,\alpha,\omega)$.

A interacting Hamiltonian is described by a self-adjoint operator $K$ affiliated to $\vNb^+$, and the interacting $(\alpha_K,\beta)-$KMS state is $\omega_K$.

From now on, we mainly focus on the probabilistic side, using the Klein-Landau correspondence. Therefore, we consider the massive Euclidean free field, described by the abstract Wiener space $(\Hi_C,\Ba,\p)$, with $\Hi_C = \overline{C\mathscr D}^{\norm{ \ \cdot \ }_{\Hi_C}}$, with $\Ba = \mathcal L_{a,\rho}$ for some fixed $a > \frac{d-2}{4}$ and $\rho > \frac{d-1}{2}$, and with $p$ the Gaussian measure with covariance $C$. The norm on $\Ba$ space is denoted $\norm{ \ \cdot \ } \equiv \norm{ \ \cdot \ }_{a,\rho}$.
The random field will be denoted $\varphi \in \Ba$, by a slight abuse of notation, since the distribution $\varphi \in \mathscr D'(\Mb)$ is $\mu_C-$almost surely in $\Ba$.

Under the identification of $K$ with a $\Sigma_0-$measurable function using the Klein-Landau correspondence, we assume that there are two real numbers $p,\ q \geq 2$, with $\frac1p+\frac1q=\frac12$, such that $K \in L^p(\Q,\Sigma_0,\mu_C)$ and $e^{-\beta K} \in L^q(\Q,\Sigma_0,\mu_C)$, and define $V_0 \equiv \int_{-\beta/2}^{\beta/2} U(\tau) K \dd \tau$. Thanks to the assumptions on $K$, this is a measurable function, and moreover $e^{-V_0} \in L^1(\Ba,\Sigma,\mu_C)$

We now describe the covariance decomposition underlying the RG flow. 

Let $[0,\infty) \ni t \to \dot C_t$ be a one-parameter family of continuous, positive, symmetric bi-distributions $\dot C_t : \mathscr D(\Mb) \times \mathscr D(\Mb) \to \mathbb R$, admitting positive bounded realisations in $L^2(\Mb)$, and differentiable for $t>0$, so that $\ddot C_t$ is bounded. The map $t \mapsto \dot C_t$ is continuous at $t=0$ in the weak operator topology on $L^2(\Mb)$. Furthermore, we assume that $C_t = \int_0^t \dot C_s \dd s$ are positive and symmetric bi-distributions with endpoint $C_\infty \equiv C : (f,g) \in \mathscr D(\Mb) \times \mathscr (\Mb) \mapsto \langle f, (-\partial_\tau^2-\Delta_x + m^2)^{-1} g \rangle$. The integral in the RG time $t$ is understood weakly on test functions.

For $t \geq s$, set $C_{s,t} \equiv C_t-C_s$. This is again a linear, continuous and positive map $C_{s,t}:\mathscr D(\Mb) \to \mathscr D'(\Mb)$. 

Since
\[
\norm{C^a}_{L^2(\Mb) \to L^2(\Mb)} \leq m^{-2a} \ ,
\]
and $w_\rho(x) \leq 1$, for every $u \in L^2(\Mb)$
\[
\norm{u}_{a,\rho} = \norm{ w_\rho C^a u}_{L^2(\Mb)} \leq m^{-2a} \norm{u}_{L^2(\Mb)} \ .
\]
Then, we have a continuous injection $j:L^2(\Mb) \to \Ba$.

Set
\[
Q_t \equiv j C_t j^* \ , \qquad
Q_{s,t} \equiv Q_t-Q_s = j C_{s,t}j^* \ , \qquad
\dot Q_t \equiv j \dot C_t j^* \ .
\]
Since $0 \leq C_{s,t} \leq C$, we have $0 \leq Q_{s,t} \leq j C j^*$. Moreover, $j C j^*$ is the covariance operator of the Gaussian measure $\p$ on $\Ba$ constructed above. Hence, for any orthonormal basis $\{e_n\}_n$ of $\Ba$, Parseval's identity gives
\[
\sum_n (Q_{s,t}e_n,e_n)_\Ba
\leq \sum_n (j C j^*e_n,e_n)_\Ba
= \int_\Ba \norm{\varphi}_\Ba^2 \, \dd\p(\varphi)
< \infty \ ,
\]
where the last inequality follows from Fernique's theorem. Thus, $Q_{s,t} \in \mathfrak S_1^+(\Ba)$, and therefore is the covariance operator for a centred Gaussian measure $\p_{s,t}$ on $\Ba$. Furthermore, we denote $\E_{s,t}$ the expectation value with respect to the measure $\p_{s,t}$.

For $t>0$, we assume that the integral kernel of $\dot C_t$ belongs to $\mathscr E(\Mb \times \Mb)$, , where $\mathscr E(\Mb) \equiv C^\infty(\Mb)$, and depends continuously on $t$ in this topology. Then, for every $t>0$,
\[
h^{-1} C_{t,t+h}
\xrightarrow[h \downarrow 0]{\mathscr E(\Mb \times \Mb)}
\dot C_t \ .
\]
This is the regularity assumption used to identify the generator of the Gaussian semigroup. For the Bochner formula below, we also assume that $t \mapsto \dot C_t$ is differentiable in the same topology.

Finally, we assume that the integral kernel of $\dot C_t$ is invariant under Euclidean time translations, as is natural for a decomposition of the covariance of the Euclidean massive free field and is verified in the concrete example of the Sine-Gordon model. This assumption simplifies the bound in Thm.~\ref{thm:entropy-inequality}.

Since the Polchinski equation is a nonlinear PDE on $\Ba$, we give a suitable notion of differentiability for functions on $\Ba$.

To this end, consider a real-valued function $F: V \subset \Ba \to \mathbb R$, with $V$ open. If we consider $h \in \Hi_C$, then the function $G_\varphi(h) \equiv F(\varphi+h)$ is defined in a neighbourhood of the origin in $\Hi_C$. The Fréchet derivative of $G$ at $0$ is an element $G'(0) \in \Hi_C^*$. Then, $DF(\varphi) \equiv G'_\varphi(0)$ is the \emph{derivative of $F$ at $\varphi$ in $\Hi_C$ directions}; of course $F$ is \emph{$\Hi_C-$differentiable} if $G'_\varphi(0)$ exists. The second $\Hi_C-$derivative of $F$ at $\varphi$ is denoted $D^2 F(\varphi)$, and when it exists it is a symmetric bounded operator from $\Hi_C$ into $\Hi_C$. We say that $F$ is three times continuously differentiable in $\Hi_C$ directions if, for every $\varphi \in V$, $G_\varphi$ is of class $C^3$ in a neighbourhood of the origin in $\Hi_C$. We denote its third derivative at $0$ by $D^3F(\varphi)$.

Since the $\Hi_C$ norm is stronger than the $\Ba$ norm, the derivative in $\Hi_C$ directions exists if the derivative in $\Ba$ directions exists.
The derivatives in $\Ba-$directions are denoted $F', \ F''$, etc. When $F''(\varphi)$ exists, its restriction to $\Hi_C$ coincides, as a bilinear form, with $D^2 F(\varphi)$. In this case, by Corollary 5 in Ref. \cite{Gross1967}, $D^2 F(\varphi)$ is a trace class operator on $\Hi_C$.

Using the continuous inclusion $\mathscr D(\Mb) \to \Hi_C$, we can restrict the derivatives in $\Hi_C$ directions to test functions. By the Schwartz kernel theorem, for $k=1,2,3$, there is a unique distribution $F^{(k)}(\varphi) \in \mathscr D'(\Mb^k)$ such that
\[
\langle F^{(k)}(\varphi),h_1\otimes\cdots\otimes h_k\rangle
= D^kF(\varphi)[h_1,\ldots,h_k] \ ,
\qquad h_1,\ldots,h_k\in\mathscr D(\Mb) \ .
\]
We denote the first two distribution kernels by $\nabla F(\varphi)$ and $\Hess F(\varphi)$, respectively.

To implement locality, we restrict to functions $F:\Ba \to \mathbb C$ with compact spacetime support.
\begin{definition}[Compactly supported function]\label{def:support}
Let $F: \Ba \to \mathbb C$. The \emph{support} of $F$ is defined as the set $\supp F \subseteq \Mb$ satisfying
\begin{gather*}
\supp F := \{ \x \in \Mb \ | \ \forall \ \text{neighbourhoods} \ U \ \text{of} \ \x \in \Mb \ \exists \ \varphi, \psi \in \Ba, \ \\
\supp \psi \subseteq U : F(\varphi + \psi) \neq F(\varphi)  \}     \ .
\end{gather*}
A function $F:\Ba \to \mathbb C$ is \emph{compactly supported} if there is a compact set $\mathbb A \subset \Mb$ such that $\supp F \subseteq \mathbb A$.
\end{definition}

By definition, when they exist, the derivatives of a compactly supported function are also compactly supported. Indeed, consider $(\tau,x) \notin \supp F$. Then, there exists a neighbourhood $U \ni (\tau,x)$ such that, for every $\psi$ with $\supp \psi \subset U$ and for every $\varphi \in \Ba$, $F(\varphi+\psi) = F(\varphi)$. Taking $\psi=sh$ with $h\in\mathscr D(U)$ and differentiating at $s=0$ gives $DF(\varphi)[h]=0$. Since this holds for every $\varphi$, differentiation in the remaining directions gives $\supp F^{(k)}(\varphi)\subseteq(\supp F)^k$. In particular, $\nabla F(\varphi) \in \mathscr E'(\Mb)$ and $\Hess F(\varphi) \in \mathscr E'(\Mb\times\Mb)$.

The distribution $\nabla F(\varphi)$ admits an $L^2(\Mb)$ realisation if and only if $h \mapsto DF(\varphi)[h]$, restricted to $\mathscr D(\Mb)$, is continuous in the $L^2(\Mb)$ norm. In this case, we denote the representing function by the same symbol, so that
\[
\langle \nabla F(\varphi),h\rangle
= (\nabla F(\varphi),h)_{L^2(\Mb)}
= DF(\varphi)[h] \ , \qquad h \in \mathscr D(\Mb) \ .
\]
In particular, when $F'(\varphi)$ exists, this function is $j^*F'(\varphi)$, where $j:L^2(\Mb) \to \Ba$ is the canonical continuous inclusion. Similarly, if $D^2F(\varphi)$ extends to a bounded bilinear form on $L^2(\Mb)\times L^2(\Mb)$, we denote by $\Hess F(\varphi)$ also the bounded operator with this distribution kernel.

Since $\dot C_t \in \mathscr E(\Mb\times\Mb)$ for $t>0$ by assumption, the pairings
\[
\dot C_t(\nabla F,\nabla F)(\varphi)
\equiv \langle F^{(1)}(\varphi)\otimes F^{(1)}(\varphi),\dot C_t\rangle
\ ,\qquad
\Delta_{\dot C_t}F(\varphi)
\equiv \langle \Hess F(\varphi),\dot C_t\rangle
\]
are well-defined.

For the heat kernel decomposition used in the application to Sine-Gordon in Sec.~\ref{sec:sine-gordon}, $\dot C_0=\frac{1}{2}I$. At $t=0$, whenever $\nabla V_0(\varphi)$ admits an $L^2(\Mb)$ realisation, we set
\[
\dot C_0(\nabla V_0,\nabla V_0)(\varphi)
\equiv \frac{1}{2}\norm{\nabla V_0(\varphi)}_{L^2(\Mb)}^2 \ .
\]

The class of functions we consider for the Polchinski equation is the set $\mathcal F$ of bounded, measurable, compactly supported functions $F : \Ba \to \mathbb C$, such that $e^{-F} \in L^1(\Ba,\Sigma,p)$, which are three times continuously differentiable in $\Hi_C$ directions, with uniformly bounded derivatives,
\[
\sup_{\varphi\in\Ba}\norm{D^kF(\varphi)}<\infty \ ,
\qquad k=1,2,3 \ .
\]
Fix $\eta \in C^\infty_c(\mathbb R^{d-1})$ equal to one on the unit ball, and set $\eta_R(\tau,x) = \eta(x/R)$. We further require that, for every $\varphi \in \Ba$ and every $F \in \mathcal F$,
\begin{equation}\label{eq:limit-support}
\lim_{R\to \infty}F(\eta_R \varphi) = F(\varphi) \ .
\end{equation}
It follows that \(F(\varphi)=F(\widetilde\varphi)\) whenever \(\varphi\) and \(\widetilde\varphi\) agree on a neighbourhood of \(\operatorname{supp}F\). Indeed, \(\eta_R(\varphi-\widetilde\varphi)\) has compact support disjoint from \(\operatorname{supp}F\). By a finite partition of unity, it is a sum of perturbations supported in neighbourhoods on which \(F\) is insensitive to such perturbations. Applying Def.~\ref{def:support} successively gives \(F(\eta_R\varphi)=F(\eta_R\widetilde\varphi)\). Letting \(R\to\infty\) gives the claim.

We then assume that $V_0 \in \mathcal F$. This is sufficient to derive the Polchinski equation. For the entropy bound, we further assume that $\nabla V_0(\varphi)$ admits an $L^2(\Mb)$ realisation for $\p-$almost every $\varphi$, and that
\[
\E\norm{\nabla V_0}_{L^2(\Mb)}^2 < \infty \ .
\]

\subsection{Polchinski equation}
Thanks to the probabilistic formula for Araki entropy, the relevant object to study entropy bounds is the partition function
\[
Z_0 = \E\left (e^{-V_0} \right ) \ .
\]

\begin{remark}
This is closely connected to the moment \emph{generating function} of the interacting measure $\mu_V$, given by the map
\[
\mathscr D \ni j \mapsto Z_0(j) \equiv \frac{1}{Z_0}\E\left(e^{\varphi(j)} e^{-V_0(\varphi)} \right ) \ .
\]
Indeed, by the Cameron-Martin theorem, this is equivalent to
\[
Z_0(j) = \frac{e^{\frac{1}{2}\langle j,Cj\rangle}}{Z_0}\E\left( e^{-V_0(\varphi+\phi_j)} \right ) \ ,
\]
where $\phi_j = Cj \in \Hi_C$. It is natural to isolate the last term from the normalisation factors and define an effective interaction
\[
\Ba \ni \phi \mapsto V_\infty(\phi) \equiv -\log \E\left(e^{-V_0(\varphi+\phi)} \right ) \ .
\]
In particular, $V_\infty(0)$ is the free energy.
\end{remark}

Following the RG philosophy, we can introduce a scale decomposition of the partition function by decomposing the Gaussian covariance.

\begin{definition}
The \emph{Gaussian semigroup} $P_{s,t} : \mathcal F \to \mathcal \mathcal F$, for $t \geq s$ is defined by
\[
P_{s,t} : F \mapsto P_{s,t}F(\phi) \equiv \E_{s,t}(F(\varphi+\phi)) \ .
\]
\end{definition}

\begin{definition}\label{def:effective interaction}
The \emph{effective interaction} is a one-parameter family of functionals 
\[
 \mathbb R^+ \times \Ba \ni t \times \phi \mapsto V_t(\phi) \equiv -\log \left[P_{0,t}e^{-V_0}(\phi) \right ] = - \log \E_{0,t} \left[e^{- V_0(\varphi+\phi)} \right ] \ .
\]
\end{definition}
Since $V_0 \in \mathcal F$ by the assumptions in Sec. \ref{sec:assumptions}, then $e^{-V_0} \in \mathcal F$ and is bounded away from zero.

\begin{lemma}\label{lemma:markov-semigroup}
The operators $(P_{s,t})_{s\leq t}$ define a time-dependent Markov semigroup on $\mathcal F$, satisfying
\[
P_{t,t} F = F(\phi) \ , \quad P_{s,t}1 = 1 \ , \quad P_{r,t}P_{s,r} = P_{s,t} \ , \quad s \leq r \leq t
\ .
\]
\end{lemma}

\begin{proof} 
First, from the definition it follows that $\norm{P_{s,t} F}_\infty \leq \norm{F}_\infty$, where as usual $\norm{ F }_\infty \equiv \sup_{\varphi \in \Ba}\abs{F(\varphi)}$. This implies that $P_{s,t}$ maps bounded functions into bounded functions.

To check that it maps compactly supported functions into compactly supported functions, assume that $\x \notin \supp F$; then, by definition there is a neighbourhood $U$ of $\x$ such that $F(\eta+\zeta) = F(\eta)$ for every $\eta, \ \zeta \in \Ba$ with $\supp \zeta \subseteq U$. Applying this with $\eta = \varphi+\phi$ we get
\[
P_{s,t}F(\varphi+\phi+\zeta) = \int_\Ba F(\varphi+\phi+\zeta) \dd \p_{s,t}(\varphi) = \int_\Ba F(\varphi+\phi) \dd \p_{s,t}(\varphi) = P_{s,t}F(\phi) \ ,
\]
Since this holds for arbitrary $\phi \in \Ba$, then $\Mb \setminus \supp F \subseteq \Mb \setminus \supp(P_{s,t}F)$. 

For \(R\) sufficiently large, \(\eta_R\phi+\varphi\) and \(\phi+\varphi\) agree near \(\supp F\). By the assumptions in Sec.~\ref{sec:assumptions}, their values under \(F\) coincide. Integrating gives \(P_{s,t}F(\eta_R\phi)=P_{s,t}F(\phi)\), so \(P_{s,t}F\) also satisfies the condition~\eqref{eq:limit-support}.

Since $Q_{t,t}= 0$ we have $\p_{t,t} = \delta_0$, so that $P_{t,t}F(\phi) = F(\phi)$. From the definition, $P_{s,t}1=1$ is obvious. 

It remains to prove the semigroup property. For each $\p_{s,t}$, its characteristic function is
\[
\hat \p_{s,t}(\ell) = \int_\Ba e^{i\ell(\varphi)} \dd \p_{s,t}(\varphi) = e^{-\frac{1}{2}\ell(Q_{s,t} \ell)} \ , \ \forall \ell \in \Ba^* \ .
\]
From
\[
\widehat{ \p_{r,t} \ast \p_{s,r}}(\ell) = \hat \p_{r,t}(\ell)\hat \p_{s,r}(\ell) \ ,
\]
we get
\[
\widehat{ \p_{r,t}\ast \p_{s,r}}(\ell) = \exp{-\frac{1}{2}\left(\ell(Q_{r,t} \ell) + \ell(Q_{s,r}(\ell)) \right)} = \exp{-\frac{1}{2}\left(\ell(Q_{s,t} \ell) \right)} = \hat \p_{s,t}(\ell) \ ,
\]
since $Q_{r,t} + Q_{s,r} = Q_{s,t}$. By uniqueness of the characteristic function for Radon probability measures on a separable Banach space (see, e.g., Ref. \cite[Proposition 3.12]{Hairer2026}) we have
\[
\p_{r,t}\ast \p_{s,r} = \p_{s,t} \ ,
\]
from which, using Fubini's theorem since $F$ is bounded, we get
\[
P_{r,t} P_{s,r} F(\phi) = P_{s,t}F(\phi) \ .
\]
\end{proof}
Thanks to this lemma, since $V_0, \ e^{-V_0} \in \mathcal F$, and $e^{-V_0}$ is bounded away from zero, the composition of the Gaussian convolution with the logarithm guarantees that $V_t, \ e^{-V_t} \in \mathcal F$.

To each covariance $Q_{s,t}$ we can associate its Cameron-Martin space $\Hi_{s,t} \equiv \Hi_{Q_t-Q_s}$. Since $Q_{t_1} \leq Q_{t_2}$ if $t_1 \leq t_2$, we can use Proposition \ref{prop:Cameron-Martin-bound} to show that, for every fixed $s$ and a collection $s \leq t_0 \leq t_1 \leq \ldots t_n$ we have
\[
\Hi_{s,t_0} \subseteq \Hi_{s,t_1} \subseteq \ldots \subseteq \Hi_{s,t_n} \ ,
\]
and similarly, for fixed $t$ and a family $s_1 \leq s_2 \leq \ldots \leq s_n \leq t$ we have
\[
\Hi_{s_n,t} \subseteq \Hi_{s_{n-1},t} \subseteq \ldots \subseteq \Hi_{s_1,t} \ . 
\]
The full Cameron-Martin space is $\Hi_C = \Hi_{0,\infty}$; therefore, for every $t \geq s \geq 0$ we have the inclusion $\Hi_{s,t} \subseteq \Hi_C$. Furthermore, we can consider the Cameron-Martin space $\Hi_{\dot Q_t}$ associated with $\dot Q_t$.

Let \(\mathcal F^k\) denote the bounded measurable local functionals satisfying the support and cutoff conditions of Sec.~\ref{sec:assumptions}, with \(k\) bounded, continuous derivatives along \(\Hi_C\). Then, $\mathcal F = \mathcal F^3$.

\begin{proposition}\label{prop:generators}
Assume $F \in \mathcal F^2$. Then, the Gaussian convolution has pointwise generator 
\[
\lim_{h\to 0} \frac{P_{t,t+h}F(\phi)-F(\phi)}{h}= \frac12 \Delta_{\dot C_t} F \equiv \langle \Hess F(\varphi),\dot C_t \rangle \ ,
\]
where $\Delta_{\dot C_t}$ is the \emph{functional Laplacian}. Moreover, $\Delta_{\dot C_t}$ is bounded, and $\Delta_{\dot C_t} F \in \mathcal F^1$ if $F \in \mathcal F$.

The gaussian convolution then satisfies
\begin{equation}\label{eq:derivative-Polchinski-semigroup}
\frac{\partial}{\partial t}P_{s,t}F(\phi) = \frac{1}{2}\Delta_{\dot C_t} P_{s,t}F(\phi) \ , \quad \frac{\partial}{\partial s}P_{s,t} F(\phi) = -\frac{1}{2}P_{s,t}\Delta_{\dot C_s} F(\phi) \ ,
\end{equation}
where $\Delta_{\dot C_t} F \equiv \langle \Hess F(\varphi),\dot C_t \rangle$ is the \emph{functional Laplacian}. 
\end{proposition}
\begin{proof}
The proof follows the same steps as Gross' proof of Proposition 8 in Ref. \cite{Gross1967potential}: his construction corresponds to the special case $C_t = t C$. 

Let $F \in \mathcal F$, and let $\chi \in \mathscr D(\Mb)$ be a compactly supported function, equal to $1$ on an open neighboorhood $U$ such that $\supp F \subset U$. For every $\varphi, \ \psi \in \Ba$, the configurations $\varphi+\psi$ and $\varphi+\chi\psi$ agree on $U$. Hence, by the observation following the definition of $\mathcal F$ in Sec.~\ref{sec:assumptions},
\[
F(\varphi+\psi) = F(\varphi + \chi \psi) \ .
\] 

Fix $t>0$, let $\psi$ have law $p_{t,t+s}$, and set $Y_s \equiv \chi \psi$. The covariance kernel of $Y_s$ is
\[
K_s = (\chi \otimes \chi) C_{t,t+s} \ , \quad s^{-1} K_s \to (\chi \otimes \chi) \dot C_t \ \text{in} \ \mathscr D(\Mb^2) \ .
\]
Since $\dot C_t$ is smooth, $C_{t,t+s}$ is smooth for every $t >0$, and therefore $K_s$ is also smooth. Then, $Y_s$ admits a smooth realisation (by an application of Kolomogorov's continuity theorem, see e..g \cite[Appendix A.9]{Nazarov2015}), and so $Y_s \in C^\infty_c(\Mb) \subset \Hi_C$ almost surely.

By assumption on $\mathcal F$, the function $x \mapsto F(\varphi + x Y_s)$ is three times differentiable on $[0,1]$, with
\[
\frac{d}{dx}F(\varphi+xY_s) =DF(\varphi+xY_s)[Y_s] \ , \quad \frac{d^2}{d x^2}F(\varphi+xY_s) =D^2 F(\varphi+xY_s) [Y_s,Y_s] \ .
\]
Since the second derivative is bounded, the first derivative is absolutely continuous, and two integrations by parts yield
\[
F(\varphi+Y_s) - F(\varphi) = DF(\varphi)[Y_s] + \int_0^1 d x (1-x) D^2 F(\varphi+xY_s)[ Y_s,Y_s] \ ,
\]
Integrating this identity against $\dd p_{t,t+s}(\psi)$ then
\[
\begin{aligned}
\frac1s \left[P_{t,t+s}F(\varphi)-F(\varphi)\right ]
&= \frac1s \int_\Ba \left\{ F(\varphi+\chi\psi)-F(\varphi)\right \}\dd \p_{t,t+s}( \psi) 
 \ .
\end{aligned}
\]
Since the measure is even and $DF(\varphi)[Y_s]$ is an odd function in $Y_s$, the first term vanishes identically. 

Then, it remains to evaluate the second derivative. By adding and subtracting the integral of $D^2 F(\varphi) [Y_s,Y_s]$, we get
\[
\frac{1}{s}\left[ P_{t,t+s}F(\varphi)-F(\varphi)\right] = \frac{1}{2s}\int_\Ba D^2 F(\varphi)[ Y_s,Y_s] \dd p_{t,t+s}(\psi) + R_s(\varphi) \ ,
\]
where
\[
R_s(\varphi) \equiv \frac{1}{s}\int_0^1 (1-x) \int_\Ba  \left [ D^2F(\varphi + x Y_s)[Y_s,Y_s] - D^2 F(\varphi)[Y_s,Y_s] \right ] \rangle \dd p_{t,t+s}(\psi) \dd x \ .
\]
We now bound the remainder. Let $M_3 \equiv \sup_{\varphi \in \Ba}\norm{D^3 F(\varphi)}$. The fundamental theorem of calculus gives the bound
\[
\abs{R_s(\varphi)} \leq \frac{1}{s}\int_0^1 (1-x)x \dd x M_3 \E_{t,t+s}\norm{Y_s}_{\Hi_C}^3 = \frac{M_r}{6s} \E_{t,t+s}\norm{Y_s}_{\Hi_C}^3 \ .
\]
Recalling that $\Hi_C = H^1_m(\Mb)$, we have the estimate
\[
\E_{t,t+s}\norm{Y_s}_{\Hi_C}^2 = \int_{\Mb} \left[ (m^2 + \nabla_x \nabla_y) K_s(x,y) \right ]_{y=x} \dd x = O(s) \ ,
\]
where the last estimate comes from the definition of $K_s$ and the observation that $s^{-1} K_s \to (\chi \otimes \chi) \dot C_t$ for $s\to 0$ in $\mathscr D(\Mb^2)$. Then, Wick theorem gives
\[
\E_{t,t+s}\norm {Y_s}_{\Hi_C}^4 = (\E_{t,t+s} \norm{Y_s}_{\Hi_C}^2)^2 + 2 \\E_{t,t+s}\norm{Y_s}_{\Hi_C}^4 \leq 3 (\E_{t,t+s} \norm{Y_s}_{\Hi_C}^2)^2 \ ,
\]
while from the Cauchy-Schwarz inequality we get
\[
\E_{t,t+s} \norm{Y_s}_{\Hi_C}^3 \leq (\E_{t,t+s}\norm{Y_s}_{\Hi_C}^2)^{1/2}(\E_{t,t+s}\norm{Y_s}_{\Hi_C}^4)^{1/2} \leq \sqrt{3}(\E_{t,t+s}\norm{Y_s}_{\Hi_C}^2)^{3/2} \ .
\]
From the estimate $\E_{t,t+s}\norm{Y_s}_{\Hi_C}^2 = O(s)$ we can then conclude that there is a $C_{t,\chi}$ such that
\[
\abs{R_s(\varphi)} \leq \frac{M_3}{s}\int_0^1(1-x)x \dd x \E_{t,t+s} \norm{Y_s}_{\Hi_C}^3 \leq \frac{\sqrt{3}}{6} M_3 c_{t,\chi}^{3/2} \sqrt{s} \ .
\]
It remains to evaluate the term containing $D^2F$. Since $\varphi$ is a fixed background, the definition of the distributional pairing gives
\[
\begin{aligned}
\int_\Ba  D^2 F(\varphi) [Y_s, Y_s] \dd p_{t,t+s}(\psi) &=
\int_\Ba  \langle \Hess F(\varphi) ,Y_s \otimes Y_s\rangle \dd p_{t,t+s}(\psi) \\
&= \langle \Hess F (\varphi),K_s \rangle \\
&= \langle \Hess F(\varphi), C_{t,t+s} \rangle \ . 
\end{aligned}
\]
The second equality follows from $\E_{t,t+s}\left[ Y_s \otimes Y_s \right] = K_s$; the expectation value passes through the distribution pairing because $Y_s$ has finite second moments in every $C^k$ norm on its fixed compact support. The last equality follows since $\chi \otimes \chi = 1$ near $\supp\Hess F(\varphi) \subseteq \supp F \subset U$. Since $s^{-1} C_{t,t+s} \to \dot C_t$ in $\mathscr E(\Mb^2)$ we conclude
\[
\lim_{s \to 0}\frac{1}{s}\left[ P_{t,t+s}F(\varphi)-F(\varphi)\right] = \frac{1}{2}\langle \Hess F(\varphi), \dot C_t \rangle = \frac{1}{2}\Delta_{\dot C_t}F(\varphi) \ .
\]
Finally, using the semigroup algebra, for $s \leq t$ and positive $h$ we have
\[
\begin{aligned}
P_{s,t+r} &= P_{t,t+r} P_{s,t} \\
P_{s-r,t} &= P_{s,t} P_{s-r,s} \ .
\end{aligned}
\]
From this, we can decompose
\[
\begin{aligned}
\left ( P_{s,t+r}-P_{s,t}\right )F(\varphi) &= (P_{t,t+r}-1)P_{s,t} F(\varphi) \\
\left ( P_{s-r,t}-P_{s,t}\right )F(\varphi) &= P_{s,t}(P_{s-r,s}-1) F(\varphi) \ .
\end{aligned}
\]
Since Gaussian convolution preserves locality and the bounds on the Cameron-Martin derivatives, we can apply the above result to these identities to conclude.
\end{proof}

\begin{proposition}[Polchinski equation]\label{prop:polchinski}
Let $C_t$ be as in Section \ref{sec:assumptions}, and let $V_0 \in \mathcal F$. Then, the effective interaction satisfies the \emph{Polchinski equation}
\begin{equation}\label{eq:polchinski}
\frac{\partial}{\partial t} V_t = \frac{1}{2}\Delta_{\dot C_t} V_t - \frac{1}{2}\dot C_t(\nabla V_t, \nabla V_t) \ .
\end{equation}
\begin{proof}
The proof follows by direct computation from the definition. Indeed, consider $Z_t(\phi) \equiv P_{0,t} e^{-V_0(\varphi)}$;  by Eq. \eqref{eq:derivative-Polchinski-semigroup} it satisfies
\[
\frac{\partial}{\partial t} Z_t(\varphi) = \frac{1}{2}\Delta_{\dot C_t} Z_t(\varphi) \ .
\]
Then, 
\[
\frac{\partial}{\partial t} V_t 
= -Z_t^{-1} \frac{\partial}{\partial t}Z_t 
= -e^{V_t} \frac{1}{2}\Delta_{\dot C_t} e^{-V_t} 
=  \frac{1}{2}\Delta_{\dot C_t} V_t - \frac{1}{2}\dot C_t(\nabla V_t, \nabla V_t) \ .
\]
The second term, in particular, arises from the restriction of $V_t'$ to $L^2(\Mb)$ directions.
\end{proof}
\end{proposition}


\subsection{RG decomposition of the relative entropy and proof of Theorem \ref{thm:entropy-inequality}}

We now have all the necessary ingredients to prove Theorem \ref{thm:entropy-inequality}. The proof has two steps. First, we decompose the relative entropy along the RG flow; secondly, we estimate the relative entropy adapting the Bauerschmidt, Bodineau, and Dagallier's multiscale Bakry-Émery criterion to the infinite-dimensional setting. 

The next Proposition proves the decomposition of the relative entropy along the RG flow.
\begin{proposition}\label{prop:entropy-decomposition}
Assume the conditions stated in Section \ref{sec:assumptions}. Then, the relative entropy between a $(\beta, \alpha)-$KMS state $\omega$ and a $(\beta, \alpha^K)-$KMS state $\omega_K$ can be expressed as
\begin{equation}\label{eq:entropy-decomposition}
S \left (\omega \mid \omega_K \right ) = \frac{1}{2} \int_0^\infty \dd t P_{t,\infty}[\dot C_t(\nabla V_t, \nabla V_t) ](0) \ .
\end{equation}

\end{proposition}

\begin{proof}
Let $t \mapsto S_t$ be the function
\[
S_t(\phi) \equiv \E(V_0(\varphi+\phi)) - P_{t,\infty} V_t(\phi) \ .
\]

By definition, recalling that $P_{0,\infty} F(\phi) = \E(F(\varphi+\phi))$, we immediately have $S_0(\phi)=0$.

The second endpoint follows recalling that $P_{0,0}=1$ and the definition of the effective interaction, so that
\[
S_\infty(\phi) = \E(V_0(\varphi+\phi)) - V_\infty(\phi) = \E(V_0(\varphi+\phi)) + \log \E\left [ e^{-V_0(\varphi+\phi)} \right ] \ ,
\]
and $S_\infty(0)$ coincides with Eq. \eqref{eq:relative-entropy-probabilistic} for the relative entropy.

The Polchinski equation, together with Eq. \eqref{eq:derivative-Polchinski-semigroup} gives
\[
\frac{\partial}{\partial t}S_t(\phi) = - P_{t,\infty}\frac{\partial}{\partial t} V_t(\phi) + \frac{1}{2}\Delta_{\dot C_t} P_{t,\infty} V_t(\phi) = \frac{1}{2}P_{t,\infty} \left[ \dot C_t(\nabla V_t, \nabla V_t ) \right ](\phi) \ ,
\]
and so,
\[
S \left (\omega \mid \omega^{V_0} \right ) = \left[\int_0^\infty \frac{\partial}{\partial t}S_t(\phi) \right ]_{\phi=0} = \frac{1}{2} \int_0^\infty P_{t,\infty} \left[ \dot C_t(\nabla V_t, \nabla V_t ) \right ](0) \ .
\]
\end{proof}

As second step, we prove the analogue of a \emph{Bochner's formula} for the Polchinski flow. 

\begin{proposition}[Bochner's formula]\label{prop:Bochner}
Define $L_t \equiv\frac{1}{2}\Delta_{\dot C_t}$, and denote $\dot C_t ( \nabla V_t, \nabla V_t) = (\nabla V_t)^2_{\dot C_t}$. Then,
\begin{equation}\label{eq:Bochner}
(\partial_t - L_t)(\nabla V_t)^2_{\dot C_t} = (\nabla V_t)^2_{\ddot C_t} - \norm{{\dot C_t}^{\frac{1}{2}} \Hess V_t {\dot C_t}^{\frac{1}{2}} }^2_2 - 2 \dot C_t\left ( \nabla V_t,  \Hess V_t \dot C_t \nabla V_t\right) \ ,
\end{equation}
where $\norm{ \ \cdot \ }_2$ denotes the $L^2(\Mb)$ Hilbert-Schmidt norm.
\end{proposition}

\begin{proof}
By direct computation,
\[
\partial_t (\nabla V_t)^2_{\dot C_t} = 2 (\nabla V_t, \nabla \partial_t V_t)_{\dot C_t} + (\nabla V_t)^2_{\ddot C_t}  \ .
\]
Moreover, denoting $\nabla_i = \frac{\delta}{\delta \varphi(f_i)}$,
\[
\nabla_i (\nabla V_t)^2_{\dot C_t} = 2 ( \nabla V_t, \nabla_i \nabla V_t)_{\dot C_t} \ , 
\]
and
\[
\nabla_j \nabla_i (\nabla V_t)^2_{\dot C_t} =
2 \left \{  ( \nabla_i \nabla V_t, \nabla_j \nabla V_t)_{\dot C_t} +( \nabla V_t, \nabla_i \nabla_j \nabla V_t)_{\dot C_t} \right \} \ .
\]

Since $L_t = \frac{1}{2}\Delta_{\dot C_t}$, differentiating $(\nabla V_t,\nabla V_t)_{\dot C_t}$ twice gives
\[
\begin{aligned}
L_t(\nabla V_t,\nabla V_t)_{\dot C_t}
&= \left\langle
\Hess V_t(x,y) \otimes \Hess V_t(z,w),
\dot C_t(x,z) \dot C_t(y,w)
\right\rangle \\
&\quad + \left\langle
\nabla V_t(x) \otimes V_t^{(3)}(y,z,w),
\dot C_t(x,y) \dot C_t(z,w)
\right\rangle \ .
\end{aligned}
\]
Both pairings are well-defined, since the derivatives of $V_t$ are compactly supported and $\dot C_t \in \mathscr E(\Mb \times \Mb)$. The second pairing is equal to $2(\nabla V_t,\nabla L_t V_t)_{\dot C_t}$, since
\[
L_t V_t = \frac{1}{2}\left\langle \Hess V_t,\dot C_t \right\rangle \ .
\]

The first pairing is non-negative, and, moreover, we have
\[
\begin{aligned}
\left\langle
\Hess V_t(x,y) \otimes \Hess V_t(z,w),
\dot C_t(x,z) \dot C_t(y,w)
\right\rangle
&= \Tr_{L^2(\Mb)}
\left[ \left( \dot C_t^{\frac{1}{2}} \Hess V_t
                    \dot C_t^{\frac{1}{2}} \right)^2 \right] \\
&= \norm{ \dot C_t^{\frac{1}{2}} \Hess V_t
          \dot C_t^{\frac{1}{2}} }_2^2 \ .
\end{aligned}
\]
Here, we use $\chi \dot C_t \chi$ in the operator expressions, where $\chi \in \mathscr D(\Mb)$ is equal to $1$ on a neighbourhood of $\supp V_t$. Then, $\chi \dot C_t \chi$ is a positive bounded operator on $L^2(\Mb)$, and the distributional pairing is independent of the choice of $\chi$. Expanding the smooth localised covariance in both factors of the distributional pairing gives the sum of the squared matrix elements of \(K^{1/2}\operatorname{Hess}V_t K^{1/2}\), where $K \equiv \chi \dot C_t$, hence its squared Hilbert–Schmidt norm.


We conclude that
\[
L_t(\nabla V_t,\nabla V_t)_{\dot C_t}
= \norm{ \dot C_t^{\frac{1}{2}} \Hess V_t
         \dot C_t^{\frac{1}{2}} }_2^2
+ 2(\nabla V_t,\nabla L_t V_t)_{\dot C_t} \ .
\]

Putting the two computations together gives
\[
(\partial_t - L_t)(\nabla V_t)^2_{\dot C_t}
= (\nabla V_t)^2_{\ddot C_t} 
-  \norm{{\dot C_t}^{\frac{1}{2}} \Hess V_t {\dot C_t}^{\frac{1}{2}} }^2_2
+ 2 (\nabla V_t, \nabla \partial_t V_t)_{\dot C_t} - 2(\nabla V_t, \nabla L_t V_t)_{\dot C_t} \ .
\]
The Polchinski equation $\partial_t V_t = L_t V_t - \frac{1}{2}(\nabla V_t)^2_{\dot C_t}$ simplifies the expression into
\[
(\partial_t - L_t)(\nabla V_t)^2_{\dot C_t}
= (\nabla V_t)^2_{\ddot C_t} 
-  \norm{{\dot C_t}^{\frac{1}{2}} \Hess V_t {\dot C_t}^{\frac{1}{2}} }^2_2  
- (\nabla V_t, \nabla (\nabla V_t)^2_{\dot C_t})_{\dot C_t} \ .
\]
It remains to compute the last term. This is
\[
(\nabla V_t, \nabla (\nabla V_t)^2_{\dot C_t})_{\dot C_t} 
= 2 \left ( \nabla V_t,  \Hess V_t \dot C_t \nabla V_t\right)_{\dot C_t} \ ,
\]
so the statement follows.
\end{proof}

Bochner's formula \eqref{eq:Bochner} gives the key tool to find an upper bound on the RG decomposition of the relative entropy, leading to the proof of Theorem \ref{thm:entropy-inequality}.
\begin{corollary}\label{cor:key-estimate}
Assume that 
\[
\dot C_t \Hess V_t \dot C_t- \frac{1}{2}\ddot C_t \geq \dot \lambda_t \dot C_t \ ,
\]
and define $\lambda_t \equiv \int_0^t \dot \lambda_s d s$.
Then,
\[
(\nabla V_t)^2_{\dot C_t} \leq e^{-2\lambda_t}P_{0,t} ((\nabla V_0)^2_{\dot C_{0}}) \ .
\]
\end{corollary}
\begin{proof}
 It is a direct consequence of Bochner's formula and the assumption that
\[
(L_t-\partial_t)(\nabla V_t)^2_{\dot C_t} \geq 2 \dot \lambda_t (\nabla V_t)^2_{\dot C_t} \ .
\]
Then, consider the function
\[
s \in [0,\infty) \mapsto \psi(s) = e^{-2\lambda_t + 2\lambda_s} P_{s,t}\left[ (\nabla V_s)^2_{\dot C_{s}}\right] \ .
\]
This satisfies
\[
\psi'(s)=2 \dot \lambda_s \psi(s) + e^{-2\lambda_t + 2\lambda_s}P_{s,t}(\partial_s - L_s)(\nabla V_s)^2_{\dot C_{s}} \ , 
\]
since 
\[
\partial_s P_{s,t} = - P_{s,t} L_s \ .
\]
Then,
\[
\psi'(s) \leq 2 \dot \lambda_s \psi(s) - 2\dot \lambda_s \psi(s) =0\ .
\]
This implies $\psi(t) \leq \psi(0)$, which is the statement. In particular, $\psi(0)$ is well-defined because, by assumption, $\nabla V_0$ admits an $L^2(\Mb)$ realisation, and $\E\norm{\nabla V_0}_{L^2(\Mb)}^2 < \infty$.

\end{proof}

Corollary \ref{cor:key-estimate} is now the key ingredient to prove the multiscale Bakry-Émery criterion for the quantum relative entropy.

\begin{proposition}[Probabilistic bound on the relative entropy]
Assume that 
\[
\dot C_t \Hess V_t \dot C_t- \frac{1}{2}\ddot C_t \geq \dot \lambda_t \dot C_t \ ,
\]
and define $\lambda_t \equiv \int_0^t \dot \lambda_s d s$.
Then, for every $V_0 \in \mathcal F$ satisfying the general assumptions in Sec.~\ref{sec:assumptions},
\begin{equation}\label{eq:probabilistic-bound}
S \left (\omega \mid \omega_K \right ) \leq \frac{1}{2\gamma} \E(\nabla V_0)^2_{\dot C_{0}} \ .
\end{equation}
\end{proposition}
\begin{proof}
From Prop. \ref{prop:entropy-decomposition}, the decomposition of the relative entropy along the RG flow gives
\[
S \left (\omega \mid \omega_K \right ) = \frac{1}{2} \int_0^\infty \dd t P_{t,\infty}[\dot C_t(\nabla V_t, \nabla V_t) ](0) \ .
\]
Corollary \ref{cor:key-estimate} implies
\[
S \left (\omega \mid \omega_K \right ) \leq \frac{1}{2} \int_0^\infty \dd t e^{-2\lambda_t}P_{t,\infty}P_{0,t} (\nabla V_0)^2_{\dot C_{0}}(0) \ ,
\]
and recalling $P_{t,\infty}P_{0,t} = P_{0,\infty}$, and $P_{0,\infty} F(0) = \E(F)$, we arrive at
\[
S \left (\omega \mid \omega_K \right ) \leq \frac{1}{2} \int_0^\infty e^{-2\lambda_t} \dd t \E(\nabla V_0)^2_{\dot C_{0}} = \frac{1}{2\gamma} \E(\nabla V_0)^2_{\dot C_{0}} \ .
\]
\end{proof}

\begin{remark}\label{rmk:generalisation}
The last Proposition can be easily generalised to a form useful to prove a bound for the Sine-Gordon model. For every $s \geq 0$, set
\[
\frac{1}{\gamma_s}\equiv \int_s^\infty e^{-2(\lambda_t - \lambda_s)} \dd t \ .
\]
Then, the relative entropy satisfies
\begin{equation}\label{eq:entropy-inequality-general}
S\left ( \omega \mid \omega_K \right) \leq \frac12 \int_0^s  P_{t,\infty}\left [ \left(\nabla V_t\right)_{\dot C_t}^2\right](0) \dd t + \frac{1}{2\gamma_s}\E_s\left [\left(\nabla V_s\right)_{\dot C_s}^2 \right] \ ,
\end{equation}
where $\E_s(F) = \E_{s,\infty}(F)=P_{s,\infty}F(0)$. The bound follows splitting the relative entropy decomposition of Eq~\eqref{eq:entropy-decomposition} at $s$ and applying the monotonicity of Corollary~\ref{cor:key-estimate} to get $\psi(t) \leq \psi(s)$ for $s \leq t$.
\end{remark}

\begin{proof}[Proof of Theorem \ref{thm:entropy-inequality}]
The last step in the proof consists of the identification of the RHS in Eq.~\eqref{eq:probabilistic-bound} with its quantum counterpart. Let $K$ be the sharp-time perturbation, which, under the Klein-Landau isomorphism, can be identified with an element affiliated to $\vNb^+$. For real $p, \ q \in C^\infty_c(\mathbb R^{d-1})$, let $f_p = \K(0,p) = i \h^{-1/2}p$ and $g_q = \K(q,0)\h^{1/2}q$ be the two corresponding elements in $\mathfrak h=L^2(\mathbb R^{d-1})$. Then, the symplectic form reduces to
\[
\sigma(f_p,g_q) = -(p,q)_{L^2(\mathbb R^{d-1})} \ .
\]
The generators of the Abelian algebra $\vNb^+$ are the Weyl operators smeared with functions in $\mathfrak h$ with vanishing initial momentum, so that they are in the form $W_\omega(g_q)$. The automorphism $\beta^s_f$ of Section~\ref{sec:automorphisms-translations} therefore acts by
\[
\beta_{f_p}^s(W_\omega(g_q)) = e^{is(p,q)_{L^2(\mathbb R^{d-1})} }W_\omega(g_q) \ .
\]
Under the Klein-Landau correspondence, the Weyl operator $W_\omega(g_q)$ is identified with $e^{i\varphi(0,h q)}$, and $\beta_{f_p}^s$ corresponds to the classical automorphism $T_{f_p}^s$ acting by
\[
T_{f_p}^s e^{i\varphi(0,h q)} \equiv e^{i\varphi(0,h q)+is(p,q)_{L^2(\mathbb R^{d-1})}} \ .
\]
Since $(p,q)_{L^2(\Mb)}=(\h^{-1}p,\h q)_{L^2(\Mb)}$, this is translation of the sharp-time field by $s\h^{-1}p$.
Therefore, denoting with $\nabla K$ the $L^2-$derivative of the sharp-time function, we get
\begin{equation}\label{eq:quantum-classical-gradient}
D_{f_p} K = (\nabla K, \h^{-1} p)_{L^2(\mathbb R^{d-1})} \ ,
\end{equation}
under the identification of the abelian subalgebra $\vNb^+$ with the sub-$\sigma-$algebra $\Sigma_0$ of the Klein-Landau isomorphism.

Now, recalling that $V_0=\int_{-\beta/2}^{\beta/2} U(\tau)K \dd\tau$, let $r\in C_c^\infty(\mathbb R^{d-1};\mathbb R)$ and let $a$ be a smooth real function on the thermal circle. For a smooth sharp-time cylinder function $K$, differentiation under the integral sign gives
\[
\partial_{ar}V_0
=\int_{-\beta/2}^{\beta/2} a(\tau)U(\tau)\partial_r K \dd\tau \ ,
\]
where $(ar)(\tau,x)=a(\tau)r(x)$. This identity extends weakly to $K\in L^2(\mathcal Q,\Sigma_0,\mu_C)$ by approximation with smooth cylinder functions, since Gaussian integration by parts transfers each directional derivative to its adjoint on the test function. By assumption, the left-hand side is represented by
$\int_{-\beta/2}^{\beta/2} a(\tau)(\nabla V_0(\tau,\cdot),r)_{L^2(\mathbb R^{d-1})}\dd\tau$.
Since $a$ is arbitrary, translating back by $U(-\tau)$ gives
\[
\partial_r K
=\left(U(-\tau)\nabla V_0(\tau,\cdot),r\right)_{L^2(\mathbb R^{d-1})}
\quad\text{for almost every }\tau .
\]
Taking $r$ in a countable dense family of spatial test functions thus identifies the weak sharp-time gradient and gives $\nabla V_0(\tau,\cdot)=U(\tau)\nabla K$ for almost every $\tau$.

Therefore, recalling that $\dot C_0$ is translation invariant in Euclidean time,
\[
\E\left [ \dot C_0(\nabla V_0,\nabla V_0) \right ]
= \int_{-\beta/2}^{\beta/2}\dd\tau
  \int_{-\beta/2}^{\beta/2}\dd\rho
  \E\left [ \left( U(\tau)\nabla K,
  \dot C_0(\tau-\rho)U(\rho)\nabla K
  \right)_{L^2(\mathbb R^{d-1})} \right ] \ .
\]
By invariance of the measure under $U(\tau)$, the integrand depends only on $\tau-\rho$, and by periodicity on the thermal circle we get
\[
\E\left [ \dot C_0(\nabla V_0,\nabla V_0) \right ]
= \beta \int_{-\beta/2}^{\beta/2}\dd\tau
  \E\left [ \left( U(\tau)\nabla K,\dot C_0(\tau)\nabla K
  \right)_{L^2(\mathbb R^{d-1})} \right ] \ .
\]
Since $\E\norm{\nabla V_0}_{L^2(\Mb)}^2<\infty$ by assumption, Fubini's theorem and invariance of the Gaussian measure under $U(\tau)$ give
\[
\beta \E\norm{\nabla K}_{L^2(\mathbb R^{d-1})}^2
= \E\norm{\nabla V_0}_{L^2(\Mb)}^2 < \infty \ ,
\]
so that the distributional pairing can be regarded as the inner product in $L^2(\Mb)$.

To expand this distributional pairing, let $\{p_j\}_{j\geq 1}$ be an orthonormal basis of $L^2(\mathbb R^{d-1})$, and set $f_j \equiv \K(0,hp_j)= i \h^{1/2}p_j$. Eq.~\eqref{eq:quantum-classical-gradient} with $p = \h p_j$ gives $D_{f_j}K = (\nabla K,p_j)_{L^2(\mathbb R^{d-1})}$. The correspondence between probability expectations and Green's functions in Eq.~\eqref{eq:n-point-functions-probabilistic-Weyl}, extended to $L^2$ sharp-time observables by Cauchy-Schwarz, gives
\[
\E\left [ \dot C_0(\nabla V_0,\nabla V_0) \right ] = \beta \sum_{i,j}\int_{-\beta/2}^{\beta/2}\dd \tau (p_i ,\dot C_0(\tau) p_j)_{L^2(\mathbb R^{d-1})} G^E_2 \left (D_{f_i}K,D_{f_j}K,\tau,0 \right ) \ .
\]
\end{proof}

\section{Application: Sine-Gordon model up to $\alpha^2 < 4\pi$}

We can now focus our attention on the massive Sine-Gordon model in Lorentzian signature. We show that the model satisfies the multiscale Bakry-Émery criterion, and the relative entropy bound follows. We limit our computations to the UV finite regime $\alpha^2 < 4\pi$, since thermal states for this model in the Lorentzian case have been constructed only in this case. The extension to $\alpha^2 < 6\pi$ can be performed using estimates similar to the ones by Bauerschmidt and Bodineau, and it is left for future work. The proof follows the strategy of Bauerschmidt and Bodineau \cite{BauerschmidtBodineau2020}, with minor modifications to adapt it to our continuum formulation.

We summarise a few known results on the thermal propagator for the free, massive scalar field on $1+1$ dimensional Minkowski spacetime. 

The integral kernel of the thermal covariance in Eq. \eqref{eq:thermal-covariance} can be written in Fourier transform as
\begin{equation}
C_\beta(\tau,x) \equiv \frac{1}{2\pi}\int_0^\infty \dd p \frac{1}{\omega_p}\frac{\cosh\left[ \left ( \frac{\beta}{2}-\tau\right ) \omega_p\right ]}{\sinh \left ( \frac{\beta}{2}\omega_p \right )} \cos(px) \ .
\end{equation}
Here, thanks to translation invariance the thermal propagator $C_\beta(\tau_x-\tau_y,x-y)$ is written as a function of the imaginary time $\tau$ and the position $x$ only. Moreover, $\omega_p \equiv \sqrt{p^2 + m^2}$. The distibutional kernel $C_\beta(\tau,x)$ is a symmetric function of $x$ and periodic on $\mathbb S_\beta$, in the sense that $C_\beta(\tau,x) = C_\beta(\beta-\tau, x)$. Moreover, it is positive on $(0,\beta) \times \mathbb R$ \cite[Lemma 3.1]{Bahns2021}.

The thermal propagator admits a representation in terms of the heat kernel on the cylinder $\mathbb S_\beta \times \mathbb R$, which allows to find a natural scale decomposition. Consider the heat kernel on the real line
\[
K_{\mathbb R}(x,t) \equiv \frac{1}{\sqrt{2\pi t}}e^{-\frac{x^2}{2t}} \ ,
\]
and the heat kernel on the circle $\mathbb S_\beta = (0,\beta]$,
\[
K_{\mathbb S_\beta}(\tau,t) \equiv \frac{1}{\beta}\sum_{n \in \mathbb Z} e^{i \frac{2\pi}{\beta}n \tau}e^{-\left ( \frac{2\pi}{\beta}n\right )^2 \frac{t}{2}} \ .
\]
Then, the thermal propagator is
\begin{equation}
C_\beta(\tau,x) = \frac{1}{2} \int_0^\infty \dd t e^{-\frac{m^2}{2}t} K_{\mathbb S_\beta}(\tau,t)K_{\mathbb R}(x,t)   \ .
\end{equation}

The natural scale decomposition of the covariance, then, is
\[
C_t(\tau,x) \equiv \frac{1}{2} \int_0^t \dd s e^{-\frac{m^2}{2}s} K_{\mathbb S_\beta}(\tau,s)K_{\mathbb R}(x,s)   \ , \quad \dot C_t(\tau,x) \equiv \frac{1}{2} e^{-\frac{m^2}{2}t} K_{\mathbb S_\beta}(\tau,t)K_{\mathbb R}(x,t) \ .
\]

Consider the Cauchy surface $\Cauchy = \{ (\tau,x) \in \mathbb S_\beta \times \mathbb R \ | \ \tau = 0 \}$, and the family of functions $f_{\epsilon,x} \in C^\infty_c(\Cauchy)$ approximating the Dirac delta on $\Cauchy$, such that $\lim_{\epsilon \to 0} f_{\epsilon,x} = \delta_x$.
We denote $\x = (\tau,x)$, and
\begin{equation}\label{eq:C-dot-epsilon-t}
\Dcet(\x_i,\x_j) \equiv \langle \delta(\cdot - \tau_i) \otimes f_{\epsilon,x_i}, \dot C_t \delta(\cdot - \tau_j) \otimes f_{\epsilon,x_j}  \rangle \ ,
\end{equation}
and
\begin{equation}\label{eq:C-epsilon-t}
\Cet(\x_i,\x_j) \equiv \int_0^t \dd s \dot C_{\epsilon,s}(\x_i,\x_j) \ .
\end{equation}

From Section \ref{sec:Euclidean-free-field}, we recall that this sharp-time covariance is well-defined.
The proof of the multiscale Bakry-Émery criterion for the Sine-Gordon model are based on fundamental estimates on sharp-time $C_t$ and $\dot C_t$, which are easily proven from properties of the heat kernel. We collect them in the next Lemma.

\begin{lemma}\label{lemma:key-estimates}
Consider $\Cet$ and $\dot C_{\epsilon,t}$ as in Eqs. \eqref{eq:C-dot-epsilon-t}, \eqref{eq:C-epsilon-t}. Then,
\begin{equation}\label{eq:key-estimates}
\int_s^t \dot C_{\epsilon,r}(0,0) \dd r \leq \frac{1}{4\pi}\log \frac{t}{s}+ c_{\beta,m} \ , \quad \int \dd x \dd u \dot C_{\epsilon,t}(u,x) = \frac{1}{2} e^{-\frac{m^2}{2}t}  \ ,
\end{equation}
with $c_{\beta,m} \equiv \frac{1}{\pi}\sum_{k \geq 1} K_0(m\beta k)$, with $K_0$ the modified Bessel function.
\end{lemma}
\begin{proof}

Consider a positive, smooth, and compactly supported function $\rho \in C^\infty_c(\Cauchy)$ with $\supp \rho \subset [-1,1]$ and $\int_\Cauchy \rho =1$, and consider the standard family of Dirac approximants
\[
f_{\epsilon,x}(y)\equiv f_{\epsilon}(y-x) \equiv \frac{1}{\epsilon}\rho\left (\frac{y-x}{\epsilon} \right ) \ .
\]
Naturally, $\supp f_{\epsilon} \subset [-\epsilon,\epsilon]$ and $\int_\Sigma f_{\epsilon}(x) = 1$.
Then, by definition
\[
\Dcet(0,0)=\frac{1}{2} e^{-\frac{m^2}{2}t}K_{\mathbb S_{\beta}}(0,t) \int_{\mathbb R^{\otimes 2}}  f_{\epsilon}(x)f_{\epsilon}(x')K_{\mathbb R}(x-x',t) \ .
\]
Since $f_\epsilon \geq 0$ and has total mass one, we have the elementary bound
\[
\Dcet(0,0) \leq \frac{1}{2} e^{-\frac{m^2}{2}t}K_{\mathbb S_{\beta}}(0,t) \sup_{x \in \mathbb R} K_{\mathbb R}(x,t) = \frac{e^{-\frac{m^2}{2}t}}{\sqrt{2\pi t}}K_{\mathbb S_{\beta}}(0,t) \ .
\]

To prove the inequality for the scale integral of $\Dcet$, consider the heat kernel on the thermal circle; 
by Poisson summation formula this is
\[
K_{\mathbb S_\beta}(\tau,t) = \frac{1}{\sqrt{2\pi t}}\sum_{k \in \mathbb Z} e^{-\frac{(\tau-\beta k)^2}{2t}} \ .
\]
Then, the regularised covariance in the coincidence limit becomes bounded by
\[
\int_s^t \dot C_{\epsilon,r}(0,0) \dd r \leq \frac{1}{4\pi}\int_s^t \dd r  \frac{e^{-m^2 \frac{r}{2}}}{r} \left [ 1 + \sum_{k \in \mathbb Z \setminus \{ 0 \}} e^{-\frac{\beta^2 k^2}{2r}} \right ] \ .
\]
Consider the first term: this gives
\[
I_{s,t} \equiv   \frac{1}{4\pi}\int_s^t \dd r  \frac{e^{-m^2 \frac{r}{2}}}{ r}
 \leq \frac{1}{4\pi}\log \frac{t}{s} \ , 
\]
using $e^{-x}\leq 1$ for positive $x$. The remainder terms can be estimated using 
\[
\frac{1}{4\pi}\int_0^\infty \frac{\dd r}{r} e^{-\frac{m^2}{2}r-\frac{\beta^2k^2}{2r}} = \frac{1}{2\pi}K_0(m \beta k) \ ,
\]
where $K_0$ is the modified Bessel function. Since $K_0(x) \sim \sqrt{\frac{\pi}{2x}}e^{-x}$ for $x \to \infty$, the remainder term is bounded by the convergent series
\[
R_{s,t} \equiv \frac{1}{4\pi}\sum_{k \neq 0} \int_s^t \frac{\dd r}{r} e^{-m^2 \frac{r}{2}} e^{-\frac{\beta^2 k^2}{2r}} 
\leq \frac{1}{\pi} \sum_{k \geq 1} K_0(m \beta k)  \ ,
\]
From which the first inequality follows.

The $L^1-$estimate on $\Dcet$ follows directly from the definitions. Indeed, since the heat kernels define probability densities, by translation invariance
\[
\int \dd \tau \int \dd x \Dcet(u,x) = \frac{1}{2} e^{-\frac{m^2}{2}t}\left (\int_{\mathbb R} f_{\epsilon}(x)\right)^2 \int_{\mathbb R} K_{\mathbb R}(x,t) \int_{\mathbb S_{\beta}}K_{\mathbb S_\beta}(u,t) = \frac{1}{2} e^{-\frac{m^2}{2}t} \ ,
\]
and so we get the second equality.

\end{proof}

\subsection{Sine-Gordon Hamiltonian}\label{sec:sine-gordon}

The Sine-Gordon Hamiltonian is formally written in terms of $\cos(\varphi(t,x))$, with the pointwise evaluation of the random field, formally smeared with the Dirac delta. Of course, this is not a measurable function in $\Ba$ due to the well-known short-distance divergences. Therefore, we introduce a regularised Sine-Gordon Hamiltonian depending on a regularisation parameter $\epsilon$ and a normal-ordering prescription; we will then prove the Bakry-\Emery criterion uniformly in $\epsilon$. Similarly, the Sine-Gordon Hamiltonian is initially defined with support on a compact region in $\Cauchy = \{ (\tau,x) \in \Mb \ | \ \tau = 0 \}$, and we will prove our estimates uniformly in the volume of the support.

Consider again the family of Dirac approximants $f_{\epsilon, x} \in C^\infty_c(\Cauchy)$; we denote $\varphi_{\epsilon}(0,x) \equiv \varphi(0,f_{\epsilon,x}) = \int \dd y \varphi(0,y) f_{\epsilon,x}(y)$. 

\begin{definition}[Regularised vertex operators]
Fix $T>0$. The \emph{regularised vertex operators} are a family of functionals $\varphi \mapsto V_{\epsilon,\alpha}(\varphi)$, depending on a parameter $\alpha \in \mathbb R$ and on a regularised parameter $\epsilon >0$, defined in terms of a \emph{normal-ordering} prescription by
\[
V_{\epsilon,\alpha}(\varphi) \equiv : e^{i\alpha \varphi_\epsilon(0,x)}: \equiv e^{\frac{\alpha^2}{2}C_{\epsilon,T}(0,0)} e^{i\alpha\varphi_\epsilon(0,x)} \ .
\]

\end{definition}

In turn, the regularised vertex operators define the Sine-Gordon Hamiltonian.
\begin{definition}[Sine-Gordon Hamiltonian]\label{def:sine-gordon-hamiltonian}
Let $\alpha \in \mathbb R$ and $z \in \mathbb R$ be two real parameters defining the theory and let $\epsilon >0$ be a UV regularisation parameter. Let $g \in C^\infty_c(\Cauchy)$, $0\leq g(x) \leq 1$ for every $x\in \mathbb R$. The \emph{Sine-Gordon Hamiltonian} is
\[
K_\epsilon(\varphi)  \equiv z \int_{\mathbb R} \dd x g(x) :\cos(\alpha \varphi_\epsilon(0,x)) : =  z e^{\frac{\alpha^2}{2}C_{\epsilon,T}(0,0)}\int_{\mathbb R} \dd x g(x) \cos(\alpha \varphi_\epsilon(0,x))   \ ,
\]
and the \emph{Sine-Gordon interaction} is
\[
V_{0}(\varphi) \equiv \int_0^\beta \dd \tau U(\tau)K_\epsilon = z e^{\frac{\alpha^2}{2}C_{\epsilon,T}(0,0)}  \int_0^\beta \dd \tau \int_{\mathbb R} \dd x g(x) \cos\left ( \alpha \varphi_\epsilon(\tau,x) \right  ) \ .
\]
\end{definition}
$K_\epsilon$ by definition is an interaction in $\vNb^+$, since it is bounded and constructed from field operators restricted to a Cauchy surface, which are commutative due to causality.

The following elementary lemma shows that the Sine-Gordon Hamiltonian satisfies the conditions in Prop. \ref{prop:FKN-kernels}, and thus defines a Feynman-Kac-Nelson kernel, and it satisfies the conditions in Sec.~\ref{sec:assumptions}, so that we can use the Polchinski equation.
\begin{lemma}
Consider the Sine-Gordon Hamiltonian $K_\epsilon$ in Def. \ref{def:sine-gordon-hamiltonian}. Then, $K_\epsilon \in L^\infty(\Ba, \mu_C)$ and $e^{-\frac{\beta}{2}K_\epsilon} \in L^\infty(\Ba,\mu_C)$. Moreover, the first three derivatives  of $V_0$ in $\Hi_C$ directions are bounded, and  the first derivative is also continuous with respect to the $L^2(\Mb)$ norm.
\end{lemma}
\begin{proof}
Since the cosine is pointwise bounded and thanks to the $\epsilon-$dependent smearing, we immediately have that $C_{\epsilon,T}(0,0)$ is finite, and so
\[
\abs{K_\epsilon} \leq \abs{z e^{\frac{\alpha^2}{2}C_{\epsilon,T}(0,0)}} \norm{g}_{L^1(\mathbb R)} \equiv M_\epsilon \ \forall \varphi \in \Ba \ ,
\]
and so $K_\epsilon \in L^\infty(\Ba, \mu_C)$. Similarly,
\[
\norm{e^{-\frac{\beta}{2}K_\epsilon}}_{L^\infty(\Ba)} \leq e^{\frac{\beta}{2}M_\epsilon} \ ,
\]
and so $e^{-\frac{\beta}{2}K_\epsilon} \in L^\infty(\Ba,\mu_C)$. Similarly,
\[
\abs{V_0(\varphi)}
\leq
\beta\abs{z}e^{\frac{\alpha^2}{2}C_{\epsilon,T}(0,0)}
\norm{g}_{L^1(\mathbb R)} \ ,
\]
so that $V_0$ and $e^{-V_0}$ are bounded.

It remains to check the regularity assumptions in Sec.~\ref{sec:assumptions}. For $h\in\Hi_C$, differentiation under the integral sign gives
\[
DV_0(\varphi)[h]
=
-\alpha z e^{\frac{\alpha^2}{2}C_{\epsilon,T}(0,0)}
\int_0^\beta\dd\tau\int_{\mathbb R}\dd x g(x)
\sin\left(\alpha\varphi_\epsilon(\tau,x)\right)
\int_{\mathbb R}\dd y h(\tau,y)f_{\epsilon,x}(y) \ .
\]
More generally, the second and third derivatives are obtained by differentiating the sine and cosine and are therefore finite sums of terms bounded in absolute value by
\[
\abs{z}e^{\frac{\alpha^2}{2}C_{\epsilon,T}(0,0)}
\abs{\alpha}^k
\int_0^\beta\dd\tau\int_{\mathbb R}\dd x\,g(x)
\prod_{j=1}^k
\abs{\int_{\mathbb R}\dd y h_j(\tau,y)f_{\epsilon,x}(y)}
\ ,\quad k=1,2,3 \ .
\]
By the Cauchy-Schwarz inequality,
\[
\abs{\int_{\mathbb R}\dd y h_j(\tau,y)f_{\epsilon,x}(y)}
\leq
\norm{h_j(\tau,\cdot)}_{L^2(\mathbb R)}
\norm{f_{\epsilon,x}}_{L^2(\mathbb R)} \ .
\]
Since $\Hi_C=H^1_m(\Mb)$, the one-dimensional Sobolev embedding in the Euclidean time variable gives a constant $A$ such that
\[
\sup_{\tau\in\mathbb S_\beta}
\norm{h_j(\tau,\cdot)}_{L^2(\mathbb R)}
\leq
A\norm{h_j}_{\Hi_C} \ ,
\]
and hence, for $k=1,2,3$,
\[
\abs{D^kV_0(\varphi)[h_1,\ldots,h_k]}
\leq
A
\abs{z}e^{\frac{\alpha^2}{2}C_{\epsilon,T}(0,0)}
\abs{\alpha}^k
\left(
\int_{\mathbb R}\dd x g(x)
\norm{f_{\epsilon,x}}_{L^2(\mathbb R)}^k
\right)
\prod_{j=1}^k\norm{h_j}_{\Hi_C} \ .
\]
The right hand side is finite for every fixed $\epsilon>0$ and is independent of $\varphi$. Thus $V_0$ is three times continuously differentiable in $\Hi_C$ directions, with uniformly bounded derivatives. Moreover, since $g$ and $f_{\epsilon,x}$ are compactly supported, $V_0$ is compactly supported in the sense of Sec.~\ref{sec:assumptions}. $V_0$ depends only on the field on $A_\epsilon = S_\beta \times (\supp g + \supp f_\epsilon)$. This set is compact, and for $R$ sufficiently large $\eta_R=1$ on a neighbourhood of $A_\epsilon$, so that $V_0(\eta_R \varphi) = V_0(\varphi)$. Thus, $V_0$ satisfies Eq.~\eqref{eq:limit-support}, and so $V_0\in\mathcal F$.

Finally, we check the additional $L^2(\Mb)$ regularity required for the entropy bound. From the expression for the first derivative we see directly that
\[
\nabla V_0(\varphi)(\tau,y)
=
-\alpha z e^{\frac{\alpha^2}{2}C_{\epsilon,T}(0,0)}
\int_{\mathbb R}\dd x\,g(x)
\sin\left(\alpha\varphi_\epsilon(\tau,x)\right)
f_{\epsilon,x}(y) \ .
\]
Therefore, by the Cauchy-Schwarz inequality and $\abs{\sin}\leq1$,
\[
\begin{aligned}
\norm{\nabla V_0(\varphi)}_{L^2(\Mb)}
&\leq
\abs{\alpha z}e^{\frac{\alpha^2}{2}C_{\epsilon,T}(0,0)}
\left[
\int_0^\beta\dd\tau
\left(
\int_{\mathbb R}\dd x g(x)
\norm{f_{\epsilon,x}}_{L^2(\mathbb R)}
\right)^2
\right]^{1/2}
\\
&=
\sqrt{\beta} 
\abs{\alpha z}e^{\frac{\alpha^2}{2}C_{\epsilon,T}(0,0)}
\int_{\mathbb R}\dd x g(x)
\norm{f_{\epsilon,x}}_{L^2(\mathbb R)}
<\infty \ .
\end{aligned}
\]
In particular, $\nabla V_0(\varphi)$ admits an $L^2(\Mb)$ realisation for every $\varphi$, with a bound independent of $\varphi$, and consequently
\[
\E\norm{\nabla V_0}_{L^2(\Mb)}^2<\infty \ .
\]

The same computation at sharp time gives
\[
\nabla K_\epsilon(\varphi)(y)
=
-\alpha z e^{\frac{\alpha^2}{2}C_{\epsilon,T}(0,0)}
\int_{\mathbb R}\dd x\,g(x)
\sin\left(\alpha\varphi_\epsilon(0,x)\right)
f_{\epsilon,x}(y) \ ,
\]
and therefore
\[
\norm{\nabla K_\epsilon(\varphi)}_{L^2(\mathbb R)}
\leq
\abs{\alpha z}e^{\frac{\alpha^2}{2}C_{\epsilon,T}(0,0)}
\int_{\mathbb R}\dd x\,g(x)
\norm{f_{\epsilon,x}}_{L^2(\mathbb R)}
<\infty \ .
\]
Thus all the assumptions of Sec.~\ref{sec:assumptions} are satisfied.
\end{proof}

\subsection{Polchinski equation for the Sine-Gordon model}
We are now ready to formulate the Polchinski equation for the Sine-Gordon model. This gives the main ingredient in the proof of the Theorem \ref{thm:entropy-inequality-sine-gordon}.

We follow Brydges and Kennedy's approach \cite{Brydges1987}, as developed by Bauerschmidt and Bodineau \cite{BauerschmidtBodineau2020}. This consists of rewriting the Polchinski equation as an integral equation, and expanding its solution in series to derive useful bounds. Proving that the series converges uniformly implies that solutions of the integral equation are also solutions of the Polchinski equation.

Introduce the charges $\sigma_i \in \{-1,+1\}$, denote $\x \equiv (\tau,x) \in \mathbb S_\beta \times \mathbb R$ and $\xi \equiv (\tau,x,\sigma) \in \mathbb M_{\beta,c} \equiv \mathbb S_\beta \times \mathbb R \times \{-1,1\}$, and further denote 
\[
\dd \mu_{\beta,c}(\xi) \equiv g(x) \dd \tau \dd x \ \times \ \text{the counting measure on $\{-1,1\}$} \ .
\]
Now, consider a functional $V : \varphi \in \Ba \mapsto V(\varphi)$; its Fourier representation is
\begin{gather}
\label{eq:Fourier-expansion}
V(\varphi) \equiv \sum_{n=0}^\infty \Vn(\varphi) \ , \\
\Vn(\varphi) \equiv \frac{1}{n!}\int_{\mathbb M_{\beta,c}^{\otimes n}} \dd \mu_{\beta,c}(\xi_1) \ldots \dd \mu_{\beta,c}(\xi_n) \FVn(\xi_1,\ldots,\xi_n) e^{i\alpha \sum_{k=1}^n\sigma_k \varphi_\epsilon(\x_k)} \ ,
\end{gather}
with $V^{(0)} \in \mathbb R$ a constant term,
and $\FVn : \mathbb M_{\beta,c}^n \to \mathbb R$. Then, the Sine-Gordon interaction corresponds to
\begin{equation}\label{eq:SG-initial-conditions}
V_0^{(0)}=0 \ , \quad \hat V_0^{(1)}(\xi_1) = \frac{z}{2}e^{\frac{\alpha^2}{2}C_{\epsilon,T}(0,0)} \ , \quad \hat V_0^{(n\geq 2)} = 0 \ .
\end{equation}

Furthermore, we define, for all $f^{(n)} : \mathbb M_{\beta,c}^n \to \mathbb R$ and $[n] = \{ 1, \ldots, n \}$,
\begin{gather}
\dot u_{\epsilon,t}(\xi_i,\xi_j) \equiv \alpha^2 \sigma_i \sigma_j \Dcet(\x_i,\x_j) \ ,  \\
\dot w_{\epsilon,t}(\xi_1, \ldots, \xi_n) \equiv \frac{1}{2}\sum_{k, l \in [n]} \dot u_{\epsilon,t}(\xi_k,\xi_l) \ , \quad w_{\epsilon,t} \equiv \int_0^t \dot w_{\epsilon,s} \ , \\
\norm{f^{(1)}} \equiv \sup_{\xi_1}\abs{f(\xi_1)} \ , \\
\norm{f^{(n)}} \equiv \sup_{\xi_1} \int \dd \mu_{\beta,c}(\xi_2)\ldots \dd \mu_{\beta,c}(\xi_n) \abs{f^{(n)}(\xi_1,\ldots,\xi_n)} \ .
\end{gather}

Applying the Fourier expansion in Eqs.~\eqref{eq:Fourier-expansion} to a solution of the Polchinski equation $V_t(\varphi)$, Brydges and Kennedy proved that each Fourier coefficient satisfies a recursive integral equation, corresponding to the Duhamel formula for the Polchinski equation~\cite{Brydges1987}. Their derivation is independent on the lattice regularisation of the underlying spacetime, and so it carries over to the present context. Therefore, we have the following lemma.

\begin{lemma}\label{lemma:integral-polchinski}
\emph{\cite[Lemma 3.4]{BauerschmidtBodineau2020}}.
Given two non-empty subsets $I_1, \ I_2 \subset [n]$, we denote $I_1 \dot \cup I_2 = [n]$ the disjoint union of $I_1$ and $I_2$, with $I_1 \cup I_2 = [n]$. Moreover, if $I = \{ i_1,\ldots,i_k\} \subset [n]$, then we denote $\xi_I = (\xi_{i_1},\ldots, \xi_{i_k})$.

Then, the integral equation
\begin{multline}\label{eq:polchinski-integral}
\FVn_t(\xi_1,\ldots, \xi_n) = e^{-w_{\epsilon,t}(\xi_1,\ldots,\xi_n)} \FVn_0(\xi_1,\ldots,\xi_n) \\
+ \frac{1}{2}\int_0^t \dd s \sum_{I_1 \dot \cup I_2 = [n] } \sum_{i\in I_1,j \in I_2}  \dot u_{\epsilon,s}(\x_i,\x_j) \hat V_s^{(\abs{I_1})}(\xi_{I_1})\hat V_s^{(\abs{I_2})}(\xi_{I_2}) e^{-(w_{\epsilon,t} - w_{\epsilon,s})(\xi_1,\ldots,\xi_n)} 
\end{multline}
has a unique solution $\FVn_t$ for all $n$ and $t$. Moreover, if $V_t$, defined in terms of $\hat V_t$ in Eq. \eqref{eq:Fourier-expansion}, converges absolutely, locally uniformly in $t>0$, then $V_t$ is given by the convolution solution of the Polchinski equation in Def.~\ref{def:effective interaction}.
\end{lemma}
\begin{proof}
The details of the proof can be found in Ref.~ \cite{BauerschmidtBodineau2020} and Ref.~\cite{Brydges1987},  where the different underlying manifold plays no role. Here we summarise the proof. First observe that, for $n\leq 1$, the unique solution to Eq. \eqref{eq:polchinski-integral} is
\begin{equation}\label{eq:fourier-n-1}
\hat V_t^{(0)} = \hat V_0^{(0)} \ , \quad \hat V_t^{(1)}(\xi_1) = e^{-\frac{\alpha^2}{2} \Cet(0,0) }\hat V_0^{(1)}(\xi_1) \ .
\end{equation}
For $n>1$, $V_t^{(n)}$ is determined by the integral equation in terms of $\hat V_t^{(k)}$ with $k<n$, so that, by induction, Eq. \eqref{eq:polchinski-integral} has a unique solution for all $n$ and $t$. If $\hat V_t$ are such that the series for $V_t$ in Eq. \eqref{eq:Fourier-expansion} and its derivatives converge absolutely, then $V_t$ is a smooth function. Moreover, the integral equation~\eqref{eq:polchinski-integral} is the Duhamel formula for the Polchinski equation~\eqref{eq:polchinski}, so that a solution $V_t(\varphi)$ of the Polchinski equation has Fourier coefficients satisfying Eq.~\eqref{eq:polchinski-integral}. Then, the Fourier expansion of $V_t$ coincides with the convolution solution in Def. \ref{def:effective interaction}.
\end{proof}

The last result needed to prove that the Sine-Gordon model satisfies the Bakry-\Emery criterion is an estimate on the norm of $\hat V_n$, originally due to Brydges and Kennedy \cite{Brydges1987} and adapted to the context of LSI by Bauerschmidt and Bodineau \cite{BauerschmidtBodineau2020}.

\begin{proposition}\label{prop:fourier-estimates}
\emph{\cite[Proposition 3.5]{BauerschmidtBodineau2020}}.
Let
\[
z_{\epsilon,t} \equiv e^{-\frac{\alpha^2}{2}\left(C_{\epsilon,t}-C_{\epsilon,T}\right)(0,0)}z \ .
\]
For $n\geq 1$, the solution to Eq. \eqref{eq:polchinski-integral} with the Sine-Gordon initial conditions \eqref{eq:SG-initial-conditions} satisfies
\begin{equation}\label{eq:BK-bound}
\norm{\FVn_t} \leq n^{n-2} \abs{z_{\epsilon,t}}^n M_{\epsilon,t}^{n-1} \ , \quad M_{\epsilon,t} \equiv \int_0^t \dd s \norm{\dot u_{\epsilon,s}}e^{\alpha^2(C_{\epsilon,t}-C_{\epsilon,s})(0,0)} \ .
\end{equation}
\end{proposition}
\begin{proof}
From Eq. \eqref{eq:fourier-n-1}, the bound holds for $n=1$. For $n>1$, from the initial conditions \eqref{eq:SG-initial-conditions} and by the integral equation Eq. \eqref{eq:polchinski-integral}, we see that $\FVn_0(\xi_1,\ldots,\xi_n) = 0$; then, the first term of Eq. \eqref{eq:polchinski-integral} does not contribute. Secondly, since $t\geq s$ then $w_t - w_s \geq 0$, and the exponential in the second term of Eq. \eqref{eq:polchinski-integral} can be bounded by $1$. Therefore, we obtain
\[
\abs{\FVn_t(\xi_1,\ldots,\xi_n)} \leq \frac{1}{2}\int_0^t \dd s \sum_{I_1 \dot \cup I_2 = [n]} \sum_{i \in I_1, j \in I_2} \abs{\dot u_{\epsilon,s}(\x_i,\x_j) \hat V_s^{\abs{I_1}}(\xi_{I_1})\hat V_s^{\abs{I_2}}(\xi_{I_2})} \ .
\]
If $\abs{I_1} = n-k$ and $\abs{I_2} = k$ this implies
\[
\sup_{\xi_1} \int_{\mathbb M_{\beta,c}^n} \dd \mu_{\beta,c}(\xi_2)\ldots \dd \mu_{\beta,c}(\xi_n) \abs{\dot u_{\epsilon,s}(\x_i,\x_j) \hat V_s^{\abs{I_1}}(\xi_{I_1})\hat V_s^{\abs{I_2}}(\xi_{I_2})} 
\leq \norm{\dot u_{\epsilon,s}} \norm{\hat V^{(n-k)}_s}\norm{\hat V^{(k)}_s} \ .
\]
Then, assume by induction that the bound \eqref{eq:BK-bound} holds for all $k<n$, so that
\begin{align*}
\norm{\FVn} 
&\leq \frac{1}{2}\int_0^t \dd s\norm{\dot u_{\epsilon,s}} \sum_{k=1}^{n-1}\binom{n}{k} k(n-k) \norm{\hat V^{(n-k)}_s}\norm{\hat V^{(k)}_s} \\
&\leq \frac{1}{2}\int_0^t \dd s\norm{\dot u_{\epsilon,s}} \sum_{k=1}^{n-1}\binom{n}{k}  (n-k)^{n-k-1}k^{k-1} \abs{z_{\epsilon,t}}^n M_{\epsilon,s}^{n-2}  \ .
\end{align*}
Using the fact that \cite[Lemma 4.2]{Brydges1987}
\[
\sum_{k=1}^{n-1}\binom{n}{k}  (n-k)^{n-k-1}k^{k-1}=2(n-1)n^{n-2}  \ ,
\]
and $n\leq 2(n-1)$ for $n \geq 2$, we get (substituting $z_{\epsilon,s} = z_{\epsilon,t} e^{-\frac{\alpha^2}{2}(C_{\epsilon,s}-C_{\epsilon,t})(0,0)}$)
\begin{align*}
\norm{\FVn} &\leq (n-1)n^{n-2} \abs{z_{\epsilon,t}}^n \int_0^t \dd s \norm{\dot u_{\epsilon,s}} e^{-\frac{n}{2}\alpha^2\left ( C_{\epsilon,s}-C_{\epsilon,t}\right)(0,0)}M_{\epsilon,s}^{n-2} \\
&\leq (n-1)n^{n-2} \abs{z_{\epsilon,t}}^n \int_0^t \dd s \norm{\dot u_{\epsilon,s}} e^{-(n-1)\alpha^2\left ( C_{\epsilon,s}-C_{\epsilon,t}\right)(0,0)}M_{\epsilon,s}^{n-2} \\
&= n^{n-2} \abs{z_{\epsilon,t}}^n M_{\epsilon,t}^{n-1} \ .
\end{align*}
For $n>2$, by a change of variables we have
\[
(n-1) \int_0^t \dd s g(s) \left (\int_0^s \dd s' g(s')\right)^{n-2} = \left( \int_0^t \dd s g(s) \right)^{n-1} \ .
\]
Choosing $g(s) = \norm{\dot u_{\epsilon,s}} e^{-\alpha^2 C_{\epsilon,s}(0,0)}$ we get
\begin{multline*}
 (n-1) \int_0^t \dd s \norm{\dot u_{\epsilon,s}} e^{-(n-1)\alpha^2\left ( C_{\epsilon,s}-C_{\epsilon,t}\right)(0,0)}M_{\epsilon,s}^{n-2} \\
 = (n-1) e^{(n-1)\alpha^2 C_{\epsilon,t}(0,0)}\int_0^t \dd s \norm{\dot u_{\epsilon,s} } e^{-\alpha^2 C_{\epsilon,s}(0,0)} \left( \int_0^s \dd s' \norm{\dot u_{\epsilon,s'}} e^{-\alpha^2 C_{\epsilon,s'}(0,0)}\right)^{n-2} = M_{\epsilon,t}^{n-1} \ .
\end{multline*}
This concludes the inductive step.
\end{proof}

\begin{proposition}\label{prop:fourier-convergence}
Fix $\beta,m,T>0$ and $\alpha^2<4\pi$.
There exists $z_*=z_*(\alpha,\beta,m,T)>0$ such that, for $\abs{z}<z_*$, every $\epsilon>0$, and every $g\in C_c^\infty(\mathbb R)$ with $0\leq g\leq1$, the Fourier series~\eqref{eq:Fourier-expansion} with Sine-Gordon initial conditions converges absolutely, locally uniformly in $t>0$, to the convolution solution of the Polchinski equation.
\end{proposition}
\begin{proof}
First, consider the case in which $t \leq T$. From Eq. \eqref{eq:key-estimates},
\[
\norm{\dot u_{\epsilon,s}} \leq 2\alpha^2 \sup_{\tau_x,x}\int \dd \tau_y \dd y \dd  \abs{\dot C_{\epsilon,s}(\tau_x-\tau_y,x-y)} \leq \alpha^2 e^{-\frac{m^2}{2}s} \ .
\]
Moreover, again from Lemma \ref{lemma:key-estimates} we have
\[
\exp{\alpha^2\left( C_{\epsilon,t}-C_{\epsilon,s}\right)(0,0)} \leq e^{\alpha^2 c_{\beta,m}} \left(\frac{t}{s}\right )^{\frac{\alpha^2}{4\pi}} \ ,
\]
and
\begin{equation}\label{eq:z-bound-low}
\abs{z_{\epsilon,t}} = \abs{z} e^{\frac{\alpha^2}{2}\left( C_{\epsilon,T}-C_{\epsilon,t}\right)(0,0)}\leq e^{\frac{\alpha^2}{2} c_{\beta,m}}  \abs{z} \left(\frac{T}{t}\right )^{\frac{\alpha^2}{8\pi}} \ .
\end{equation}
Then, if $\alpha^2 < 4\pi$,
\[
M_{\epsilon,t} =  \int_0^t \dd s \norm{\dot u_{\epsilon,s}}e^{\alpha^2(C_{\epsilon,t}-C_{\epsilon,s})(0,0)} 
\leq e^{\alpha^2 c_{\beta,m}} \alpha^2  t^{\frac{\alpha^2}{4\pi}} \int_0^t \dd s s^{-\frac{\alpha^2}{4\pi}} 
= e^{\alpha^2 c_{\beta,m}} \frac{4\pi\alpha^2}{4\pi - \alpha^2} t \ ,
\]
implying
\begin{equation}
\abs{z_{\epsilon,t}} M_{\epsilon,t} \leq \abs{z} e^{\frac{3}{2}\alpha^2 c_{\beta,m}} \frac{4\pi \alpha^2}{4\pi - \alpha^2} T^{\frac{\alpha^2}{8\pi}}t^{1-\frac{\alpha^2}{8\pi}} \ .
\end{equation}
Since $\alpha^2/8\pi < 1/2$ and $t \leq T$ we arrive at the bound
\[
\abs{z_{\epsilon,t}} M_{\epsilon,t} \leq \abs{z} e^{\frac{3}{2}\alpha^2 c_{\beta,m}} \frac{4\pi \alpha^2}{4\pi - \alpha^2} T \ .
\]

Now, consider the case in which $t \geq T$. In this case, $\abs{z_{\epsilon,t}} \leq \abs{z}$. Moreover,
\[
M_{\epsilon,t} 
\leq \int_0^\infty
\norm{\dot u_{\epsilon,s}}
e^{\alpha^2(C_{\epsilon,\infty}-C_{\epsilon,s})(0,0)} \dd s \ .
\]

The heat kernel bound of Lemma~\ref{lemma:key-estimates} gives
\[
D_T \equiv \sup_{\epsilon >0} \int_T^\infty \dot C_{\epsilon, r}(0,0)\dd r < \infty \ ,
\]
since $K_{\mathbb S_\beta}(0,r)$ is bounded for $r \geq T$ and $m>0$. Thus, splitting the covariance integral at $T$ and using the definition of the regularised covariance~\eqref{eq:C-epsilon-t},
\[
e^{\alpha^2 \left (C_{\epsilon,\infty} - C_{\epsilon,s}\right)(0,0)} \leq
\begin{cases}
e^{\alpha^2\left( D_T + c_{\beta,m}\right)} \left ( \frac{T}{s}\right)^{\alpha^2/(4\pi)} \ , &0 < s \leq T \ , \\
e^{\alpha^2 D_T} \ , &s > T \ .
\end{cases}
\]
Since
\(
\norm{\dot u_{\epsilon,s}}
\leq \alpha^2 e^{-m^2s/2}
\)
and $\alpha^2 < 4 \pi$, then,
\[
\sup_{\epsilon>0,t\geq T}M_{\epsilon,t}
\le C_{\alpha,\beta,m,T}<\infty \ .
\]

Combining the two regimes for $t$, we conclude that there exists a constant $C_{\alpha,\beta,m,T}$ such that
\begin{equation}\label{eq:last-estimate}
\sup_{\epsilon>0,t>0}
\abs{z_{\epsilon,t}}M_{\epsilon,t}
\le
C_{\alpha,\beta,m,T}\abs{z} \ .
\end{equation}
Choose $z_* >0$ such that $e C_{\alpha,\beta,m,T} z_* < 1$. Then, for every $\abs{z} < z_*$, we have $e \abs{z_{\epsilon,t}} M_{\epsilon,t} <1$, uniformly in $\epsilon, \ t$, and $g$.

Finally, by Prop.~\ref{prop:fourier-estimates} and using $n^n/n! \leq e^n$, the series \eqref{eq:Fourier-expansion} converges absolutely, since
\begin{align*}
\abs{V_t(\varphi)} &\leq \sum_{n=1}^\infty \frac{1}{n!}\int_{\mathbb M_{\beta,c}^{\otimes n}} \dd \mu_{\beta,c}(\xi_1) \ldots \dd \mu_{\beta,c}(\xi_n)\abs{ \FVn(\xi_1,\ldots,\xi_n) } \\
&\leq \sum_{n=1}^\infty \frac{n^{n-2}}{n!} \abs{z_{\epsilon,t}}^n M_{\epsilon,t}^{n-1}\beta \norm{g}_{L^1(\mathbb R)} \\
&\leq \beta  \norm{g}_{L^1(\mathbb R)} \frac{e\abs{z_{\epsilon,t}}}{1-e\abs{z_{\epsilon,t}}M_{\epsilon,t}} < \infty \ .
\end{align*}
The factor \(2\) comes from the sum over the charge of the remaining vertex. Since \(|z_{\epsilon,t}|\) is bounded on compact subintervals of \((0,\infty)\) and the geometric ratio is bounded, the convergence is locally uniform in \(t>0\). By Lemma~\ref{lemma:integral-polchinski}, the sum is the convolution solution of the Polchinski equation with Sine-Gordon initial conditions.
\end{proof}

Using Prop. \ref{prop:fourier-estimates} and Prop.~\ref{prop:fourier-convergence} we can now conclude the proof of Thm. \ref{thm:entropy-inequality-sine-gordon}.

\begin{proof}[Proof of Theorem \ref{thm:entropy-inequality-sine-gordon}]
The proof relies on showing that the multiscale Bakry-\Emery criterion holds in this case, for a finite Sobolev constant $\gamma >0$ independent on $\epsilon$ and $g$.

Since $\dot C_{s}$ is given by the heat kernel of the massive Laplacian on $\mathbb S_\beta \times \mathbb R$, we can simplify the multiscale Bakry-\Emery criterion \eqref{eq:BE-criterion} setting $\dot C_{s} = \frac12 e^{-sP/2}\equiv G_s^2$ with $P= -\partial^2_\tau-\Delta_x + m^2$. Since
\[
\ddot C_{t} = - \frac{1}{2} P \dot C_{t} \ ,
\]
the multiscale Bakry-\Emery criterion in Eq.~\eqref{eq:BE-criterion} becomes
\[
G_t \left [ G_t \Hess V_t G_t + \frac{1}{4}P\right ] G_t \geq \dot \lambda_t G_t^2 \ .
\]
Therefore, since $P \geq m^2$, it is sufficient to prove
\[
G_t \Hess V_t G_t \geq \dot \mu_t  I \ ,
\]
and the multiscale Bakry-\Emery criterion \eqref{eq:BE-criterion} follows with $\dot \lambda_t = \dot \mu_t + \frac{1}{4}m^2$.

We are left to bound $G_t \Hess V_t G_t$ as quadratic form in $L^2(\Mb)$.

Let $h\in\mathscr D(\Mb)$ be real and set
\[
h_\epsilon(\tau,x) \equiv \int_{\mathbb R}\dd y f_{\epsilon,x}(y)h(\tau,y)
 \ .
\]
Since $f_{\epsilon,x}(y)=f_\epsilon(y-x)$, with $f_\epsilon\geq 0$ and
$\int_{\mathbb R}\dd y f_\epsilon(y)=1$, Jensen's inequality and
$0\leq g\leq 1$ give
\[
\int \dd\mu_{\beta,c}(\xi) \abs{h_\epsilon(\x)}^2
\leq 2\int_{\Mb}\dd\tau\dd x \int_{\mathbb R}\dd y
f_{\epsilon,x}(y)\abs{h(\tau,y)}^2
= 2\norm{h}_{L^2(\Mb)}^2 \ .
\]
Here the factor $2$ comes from the counting measure on $\{-1,1\}$.

Differentiating the $n$-th term in Eq.~\eqref{eq:Fourier-expansion} twice in the direction $h$ gives the factor $-\alpha^2(\sum_{j=1}^n\sigma_j h_\epsilon(\x_j))^2$. As in Ref.~\cite[Eqs.~(3.51)--(3.53) ]{BauerschmidtBodineau2020}, we use the Cauchy-Schwarz inequality to bound the square by $n\sum_{j=1}^n\abs{h_\epsilon(\x_j)}^2$. The initial conditions~\eqref{eq:SG-initial-conditions} and the integral equation~\eqref{eq:polchinski-integral} are invariant under permutations of the $\xi_j$, so by uniqueness $\FVn_t$ is symmetric. Therefore, all $n$ summands give the same integral, and
\begin{align*}
\abs{D^2\Vn_t(\varphi)[h,h]}
&\leq \frac{\alpha^2 n^2}{n!}
\int \dd\mu_{\beta,c}(\xi_1)\abs{h_\epsilon(\x_1)}^2
\int \dd\mu_{\beta,c}(\xi_2)\ldots\dd\mu_{\beta,c}(\xi_n)
\abs{\FVn_t(\xi_1,\ldots,\xi_n)} \\
&\leq \frac{2\alpha^2 n^2}{n!}\norm{\FVn_t}\norm{h}_{L^2(\Mb)}^2 \ .
\end{align*}
In the last inequality, the supremum in the definition of $\norm{\FVn_t}$
bounds the inner integral for every fixed $\xi_1$.

By Prop.~\ref{prop:fourier-estimates} and $2n^n/n!\leq e^n$, the second derivatives are absolutely summable, uniformly in $\varphi$. Hence, differentiating term by term, we obtain
\[
\abs{D^2V_t(\varphi)[h,h]}
\leq \alpha^2\sum_{n=1}^\infty e^n
\abs{z_{\epsilon,t}}^n M_{\epsilon,t}^{n-1}\norm{h}_{L^2(\Mb)}^2
= \frac{\alpha^2 e\abs{z_{\epsilon,t}}}
{1-e\abs{z_{\epsilon,t}}M_{\epsilon,t}}\norm{h}_{L^2(\Mb)}^2 \ .
\]
As in the proof of Prop.~\ref{prop:fourier-convergence}, choose $z_*$ such that $q\equiv e C_{\alpha,\beta,m,T}z_* <1$. Then, Eq.~\eqref{eq:last-estimate} gives
\[
1- e \abs{z_{\epsilon,t}} M_{\epsilon,t} \geq 1-q >0 \ ,
\]
and so
\[
\abs{D^2V_t(\varphi)[h,h]}
\leq \frac{\alpha^2e}{1-q}\abs{z_{\epsilon,t}} \norm{h}_{L^2(\Mb)}^2 \ .
\]
By polarization and density, $D^2V_t(\varphi)$ extends to a bounded symmetric bilinear form on $L^2(\Mb)\times L^2(\Mb)$. Its representing operator is $\Hess V_t(\varphi)$, and $\norm{\Hess V_t(\varphi)}$ denotes its operator norm on $L^2(\Mb)$. Thus, for every $f\in L^2(\Mb)$,
\[
\abs{(f,\Hess V_t(\varphi)f)}
\leq \norm{f}_{L^2(\Mb)}\norm{\Hess V_t(\varphi)f}
\leq \norm{\Hess V_t(\varphi)}\norm{h}_{L^2(\Mb)}^2 \ .
\]

From $\norm{G_t f}_{L^2(\Mb)} \leq e^{-\frac{m^2}{4}t}\norm{f}_{L^2(\Mb)}$, we get
\[
\begin{aligned}
\abs{(G_t f,\Hess V_t(\varphi )G_t f)} 
&\leq \norm{\Hess V_t(\varphi)} \norm{G_t f}_{L^2(\Mb)}^2 \\
\leq \frac{\alpha^2 e}{2(1-q)}\abs{z_{\epsilon,t}} e^{-m^2t/2}\norm{f}_{L^2(\Mb)}^2  \ .
\end{aligned}
\]

Set \(b=\alpha^2/(8\pi)\) and \(A=e^{\alpha^2 c_{\beta,m}/2}\).
Using the bounds \(|z_{\epsilon,t}|\leq A|z|(T/t)^b\) for \(0<t\leq T\) and \(|z_{\epsilon,t}|\leq |z|\) for \(t>T\), we can choose
\[
\dot{\mu}_t = -B z_* e^{-m^2t/2}
\begin{cases}
A(T/t)^b, & 0<t\leq T,\\
1, & t>T,
\end{cases}
\qquad B=\frac{\alpha^2 e}{2(1-q)} \ .
\]
Then, for every \(t\geq 0\),
\[
\mu_t \equiv \int_0^t \dot{\mu}_s \dd s
\geq -B z_* \left[AT^b\int_0^T s^{-b}e^{-m^2s/2}\dd s 
+ \int_T^\infty e^{-m^2s/2} \dd s\right]
\equiv -\mu_* \ .
\]
Since \(b<1/2\) and \(m>0\), \(\mu_*\) is finite and independent of \(\epsilon\) and \(g\).
Therefore \(\lambda_t=m^2t/4+\mu_t\geq m^2t/4-\mu_*\), and
\[
\frac{1}{\gamma} = \int_0^\infty e^{-2\lambda_t}\,\dd t
\leq e^{2\mu_*}\int_0^\infty e^{-m^2t/2}\,\dd t
= \frac{2}{m^2}e^{2\mu_*}.
\]
Hence \(\gamma\geq (m^2/2)e^{-2\mu_*}>0\), uniformly in \(\epsilon\) and \(g\).

Therefore, we can apply the multiscale Bakry-\Emery criterion, together with Remark~\ref{rmk:generalisation}, to bound the relative entropy. It remains to show that the bound is uniform in $\epsilon$. 

For $0< t \leq T$, differentiating Eq.~\eqref{eq:Fourier-expansion} in a real direction $h \in \mathscr D(\Mb)$ gives
\[
\abs{D V_t(\varphi)[h]} \leq 2 \abs{\alpha} \sum_{n \geq 1} \frac{n}{n!}\norm{\widehat V_t^{(n)}}\int_{\Mb } g(x) \abs{h_\epsilon(\tau,x)} \dd \tau \dd x \ .
\] 
The factor of $2$ comes from the sum over the charge of the remaining vertex. 

As in the bound for the Hessian, Jensen's inequality gives $\norm{h_\epsilon}_{L^2(\Mb)}^2 \leq \abs{h}_{L^2(\Mb)}^2$, and so by Cauchy-Schwarz
\[
\begin{aligned}
\int_{\Mb } g(x) \abs{h_\epsilon(\tau,x)} \dd \tau \dd x &\leq \left( \int_{\Mb} g(x)^2 \dd x \dd \tau \right)^{1/2} \norm{h}_{L^2(\Mb)}\\
&\leq \sqrt{\beta}\norm{g}_{L^2(\mathbb R)} \norm{h}_{L^2(\Mb)} \ .
\end{aligned}
\]
By Prop.~\ref{prop:fourier-convergence} and $e \abs{z_{\epsilon,t}} M_{\epsilon,t} \leq q\equiv e C_{\alpha,\beta,m,T}z_* <1$, we obtain
\[
\abs{D V_t(\varphi)[h]}\leq 
\frac{2e \abs{\alpha}}{1-q}\sqrt{\beta}\norm{g}_{L^2(\mathbb R)} \norm{h}_{L^2(\Mb)} \abs{z_{\epsilon,t}} \ .
\]
Taking the supremum over test functions $h$ with $\norm{h}_{L^2(\Mb)}=1$, by $L^2-$duality we have
\begin{equation}\label{eq:gradient-estimate}
\norm{\nabla V_t(\varphi)}_{L^2(\mathbb M_\beta)}
\leq \frac{2e\abs{\alpha}}{1-q}\sqrt{\beta}
\norm{g}_{L^2(\mathbb R)}\abs{z_{\epsilon,t}} \ .
\end{equation}

Write $F_t\equiv \dot C_t(\nabla V_t,\nabla V_t)$. By Eq.~\eqref{eq:entropy-inequality-general} with $s=T$, we have
\[
S\left (\omega \mid \omega_{K_\epsilon} \right ) \leq \frac12 \int_0^T \sup_\varphi F_t(\varphi) \dd t + \frac{1}{2\gamma_T}\sup_\varphi F_T(\varphi) \ .
\]
For the first integral, we use
\[
0 \leq \dot C_t \leq \frac12 e^{-m^2 t/2} I \leq I \ ,
\] 
so that
\[
\frac12 \sup_\varphi F_t(\varphi) \leq \frac14 e^{-m^2 t/2} \sup_\varphi \norm{\nabla V_t}_{L^2(\Mb)}^2 \ .
\]
Then, the gradient estimate~\eqref{eq:gradient-estimate}, together with the bound on $\abs{z_{\epsilon,t}}$ in Eq.~\eqref{eq:z-bound-low} give
\[
\frac{1}{2}\sup_\varphi F_t(\varphi) \leq \frac{e^2 \alpha^2}{(1-q)^2}\abs{z}^2 \beta  \norm{g}_{L^2(\mathbb R)}^2  e^{\alpha^2 c_{\beta,m}}   \left(\frac{T}{t}\right )^{\frac{\alpha^2}{4\pi}} \ .
\]
At scale $T$, we have $z_{\epsilon,T} = z$ and $\norm{\dot C_T} = \frac12 e^{-m^2 T/2}$. Thus,
\[
S\left (\omega \mid \omega_{K_\epsilon} \right ) \leq
\frac{e^2 \alpha^2}{(1-q)^2} \abs{z}^2 \beta  \norm{g}_{L^2(\mathbb R)}^2
\left ( e^{\alpha^2 c_{\beta,m}}\frac{4\pi T}{4\pi-\alpha^2}+ \frac{1}{\gamma_T}e^{-m^2 T/2}
\right ) \ .
\]
The first term is finite since $\alpha^2 < 4\pi$. For the second term, recall that $\lambda_t = m^2 t/4  + \mu_t$, with $\mu_t$ negative, so that $\lambda_T \leq m^2 T/4$. Therefore,
\[
\frac{1}{\gamma_T}e^{-m^2 T/2} = e^{2\lambda_T-m^2 T/2}\int_T^\infty e^{-2\lambda_t} \dd t \leq \frac{1}{\gamma} \ .
\]
Since $\gamma$ can be chosen uniformly in $\epsilon$ and $g$, we arrive at the statement with the constant
\[
C_{\alpha,\beta,m,T}
=\frac{e^2\alpha^2}{(1-q)^2}
\left(e^{\alpha^2c_{\beta,m}}\frac{4\pi T}{4\pi-\alpha^2}
+\frac1{\gamma}\right) \ .
\]
\end{proof}

\subsection*{Acknowledgments}
I dedicate this work to Elena and Caterina, who gifted me the days and nights needed to complete it. I acknowledge the use of AI tools for correction of proofs, literature search, and review. I am supported by an Alexander von Humboldt Fellowship for Postdocs, and I was partially supported by the ERC Starting Grant FermiMath, grant agreement nr. 101040991. Views and opinions expressed are mine and do not necessarily reflect those of the European Union or the European Research Council Executive Agency. Neither the European Union nor the granting authority can be held responsible for them. This work has been partially carried out under the auspices of Gruppo Nazionale per la Fisica Matematica of INdAM.

\printbibliography

@Article{	  albeverio1977,
  Author	= {Albeverio, Sergio and Høegh-Krohn, Raphael},
  Date		= {1977},
  JournalTitle	= {Communications in Mathematical Physics},
  Volume	= {56},
  Number	= {2},
  Pages		= {173--187},
  Publisher	= {Springer},
  Title		= {Dirichlet Forms and {{Markov}} Semigroups on
		  {$C^*-$}algebras}
}

@Article{	  araki1968,
  Title		= {Multiple {{Time Analyticity}} of a {{Quantum Statistical
		  State Satisfying}} the {{KMS Boundary Condition}}},
  Author	= {Araki, Huzihiro},
  Date		= {1968-08-31},
  JournalTitle	= {Publications of the Research Institute for Mathematical
		  Sciences},
  Volume	= {4},
  Number	= {2},
  Pages		= {361--371},
  ISSN		= {0034-5318},
  DOI		= {10.2977/prims/1195194880},
  URL		= {https://ems.press/journals/prims/articles/2450},
  URLDate	= {2026-09-04},
  langid	= {english}
}

@article{Nazarov2015,
  author        = {Fedor Nazarov and Mikhail Sodin},
  title         = {Asymptotic Laws for the Spatial Distribution and the Number of Connected Components of Zero Sets of Gaussian Random Functions},
  journal       = {Journal of Mathematical Physics, Analysis, Geometry},
  volume        = {12},
  number        = {3},
  pages         = {205--278},
  year          = {2016},
  doi           = {10.15407/mag12.03.205},
  eprint        = {1507.02017},
  archivePrefix = {arXiv},
  primaryClass  = {math.PR},
  url           = {https://arxiv.org/abs/1507.02017}
}

@Article{	  araki1973,
  Title		= {Relative {{Hamiltonian}} for {{Faithful Normal States}} of
		  a von {{Neumann Algebra}}},
  Author	= {Araki, Huzihiro},
  Date		= {1973-04-30},
  JournalTitle	= {Publications of the Research Institute for Mathematical
		  Sciences},
  Volume	= {9},
  Number	= {1},
  Pages		= {165--209},
  ISSN		= {0034-5318},
  DOI		= {10.2977/prims/1195192744},
  URL		= {https://ems.press/journals/prims/articles/2638},
  URLDate	= {2026-09-03},
  langid	= {english}
}

@Article{	  araki1975,
  Title		= {Relative {{Entropy}} of {{States}} of von {{Neumann
		  Algebras}}},
  Author	= {Araki, Huzihiro},
  Date		= {1975-12-31},
  JournalTitle	= {Publications of the Research Institute for Mathematical
		  Sciences},
  Volume	= {11},
  Number	= {3},
  Pages		= {809--833},
  ISSN		= {0034-5318},
  DOI		= {10.2977/prims/1195191148},
  URL		= {https://ems.press/journals/prims/articles/2800},
  URLDate	= {2026-09-06},
  langid	= {english}
}

@Article{	  arakiwoods1963,
  Title		= {Representations of the Canonical Commutation Relations
		  Describing a Nonrelativistic Infinite Free Bose Gas},
  Author	= {Araki, H. and Woods, E. J.},
  Date		= {1963-05},
  JournalTitle	= {Journal of Mathematical Physics},
  Volume	= {4},
  Number	= {5},
  EPrint	= {https://pubs.aip.org/aip/jmp/article-pdf/4/5/637/19220915/637_1_online.pdf},
  Pages		= {637--662},
  ISSN		= {0022-2488},
  DOI		= {10.1063/1.1704002},
  URL		= {https://doi.org/10.1063/1.1704002}
}

@Article{	  bahns2021,
  Title		= {Equilibrium States for the Massive {{Sine-Gordon}} Theory
		  in the {{Lorentzian}} Signature},
  Author	= {Bahns, Dorothea and Pinamonti, Nicola and Rejzner, Kasia},
  Date		= {2023},
  JournalTitle	= {Journal of Mathematical Analysis and Applications},
  ShortJournal	= {J. Math. Anal. Appl.},
  Volume	= {526},
  EPrint	= {2103.09328},
  EPrintType	= {arXiv},
  EPrintClass	= {math-ph},
  Pages		= {127249},
  DOI		= {10.1016/j.jmaa.2023.127249}
}

@Book{		  bakry2013,
  Title		= {Analysis and Geometry of Markov Diffusion Operators},
  Author	= {Bakry, Dominique and Gentil, Ivan and Ledoux, Michel},
  Date		= {2013},
  Publisher	= {Springer Cham},
  URL		= {https://doi.org/10.1007/978-3-319-00227-9}
}

@Article{	  barashkov2020,
  Author	= {Barashkov, N. and Gubinelli, M.},
  Date		= {2020-11},
  JournalTitle	= {Duke Mathematical Journal},
  Volume	= {169},
  Number	= {17},
  Publisher	= {Duke University Press},
  ISSN		= {0012-7094},
  DOI		= {10.1215/00127094-2020-0029},
  URL		= {http://dx.doi.org/10.1215/00127094-2020-0029},
  Title		= {A Variational Method for {$\phi^4_3$}}
}

@Article{	  bauerschmidt2023,
  Author	= {Bauerschmidt, Roland and Dagallier, Benoit},
  Date		= {2023-10},
  JournalTitle	= {Communications on Pure and Applied Mathematics},
  Volume	= {77},
  Number	= {5},
  Pages		= {2579--2612},
  Publisher	= {Wiley},
  ISSN		= {1097-0312},
  DOI		= {10.1002/cpa.22173},
  URL		= {http://dx.doi.org/10.1002/cpa.22173},
  Title		= {Log‐{{Sobolev}} Inequality for the {$\varphi^4_2$} and
		  {$\varphi^4_3$} Measures}
}

@Article{	  bauerschmidt2023b,
  Title		= {Log‐{{Sobolev}} Inequality for near Critical {{Ising}}
		  Models},
  Author	= {Bauerschmidt, Roland and Dagallier, Benoit},
  Date		= {2023-10},
  JournalTitle	= {Communications on Pure and Applied Mathematics},
  Volume	= {77},
  Number	= {4},
  Pages		= {2568--2576},
  Publisher	= {Wiley},
  ISSN		= {1097-0312},
  DOI		= {10.1002/cpa.22172},
  URL		= {http://dx.doi.org/10.1002/cpa.22172}
}

@Article{	  bauerschmidt2024,
  Title		= {Stochastic Dynamics and the {{Polchinski}} Equation:
		  {{An}} Introduction},
  Author	= {Bauerschmidt, Roland and Bodineau, Thierry and Dagallier,
		  Benoit},
  Date		= {2024-01},
  JournalTitle	= {Probability Surveys},
  Volume	= {21},
  Publisher	= {Institute of Mathematical Statistics},
  ISSN		= {1549-5787},
  DOI		= {10.1214/24-ps27},
  URL		= {http://dx.doi.org/10.1214/24-PS27},
  Issue		= {none}
}

@Article{	  bauerschmidtbodineau2020,
  Title		= {Log‐sobolev Inequality for the Continuum Sine‐gordon
		  Model},
  Author	= {Bauerschmidt, Roland and Bodineau, Thierry},
  Date		= {2020-07},
  JournalTitle	= {Communications on Pure and Applied Mathematics},
  Volume	= {74},
  Number	= {10},
  Pages		= {2064--2113},
  Publisher	= {Wiley},
  ISSN		= {1097-0312},
  DOI		= {10.1002/cpa.21926},
  URL		= {http://dx.doi.org/10.1002/cpa.21926}
}

@Article{	  bekenstein1980,
  Title		= {A Universal Upper Bound on the Entropy to Energy Ratio for
		  Bounded Systems},
  Author	= {Bekenstein, Jacob D.},
  Date		= {1981},
  JournalTitle	= {Physical Review D: Particles and Fields},
  ShortJournal	= {Phys. Rev. D},
  Volume	= {23},
  Number	= {NSF-ITP-80-38},
  Pages		= {287},
  DOI		= {10.1103/PhysRevD.23.287}
}

@Article{	  benoist2025,
  Title		= {Entropic Fluctuations in Statistical Mechanics {{II}}.
		  {{Quantum}} Dynamical Systems},
  Author	= {Benoist, T. and Bruneau, L. and Jakšić, V. and Panati,
		  A. and Pillet, C.-A.},
  Date		= {2025-08},
  JournalTitle	= {Communications in Mathematical Physics},
  Volume	= {406},
  Number	= {9},
  Publisher	= {{Springer Science and Business Media LLC}},
  ISSN		= {1432-0916},
  DOI		= {10.1007/s00220-025-05360-z},
  URL		= {http://dx.doi.org/10.1007/s00220-025-05360-z}
}

@Article{	  birke2002,
  Title		= {Kms, Etc.},
  Author	= {Birke, Lothat and Frölich, Jürg},
  Date		= {2002-07},
  JournalTitle	= {Reviews in Mathematical Physics},
  Volume	= {14},
  Pages		= {829--871},
  Publisher	= {World Scientific Pub Co Pte Lt},
  ISSN		= {1793-6659},
  DOI		= {10.1142/s0129055x02001442},
  URL		= {http://dx.doi.org/10.1142/S0129055X02001442},
  Issue		= {07n08}
}

@Book{		  brattelirobinson-1,
  Title		= {Operator Algebras and Quantum Statistical Mechanics. 1.
		  {{C}}* and {{W}}* Algebras, Symmetry Groups, Decomposition
		  of States},
  Author	= {Bratteli, O. and Robinson, D. W.},
  Date		= {1987},
  Series	= {Theoretical and Mathematical Physics},
  Publisher	= {Springer Berlin, Heidelberg},
  DOI		= {10.1007/978-3-662-02520-8}
}

@Book{		  brattelirobinson-2,
  Title		= {Operator Algebras and Quantum Statistical Mechanics.
		  {{Vol}}. 2: {{Equilibrium}} States. {{Models}} in Quantum
		  Statistical Mechanics},
  Author	= {Bratteli, O. and Robinson, D. W.},
  Date		= {1997},
  Series	= {Theoretical and Mathematical Physics},
  Publisher	= {Springer Berlin, Heidelberg},
  DOI		= {10.1007/978-3-662-02520-8}
}

@Article{	  brydges1987,
  Title		= {Mayer Expansions and the {{Hamilton-Jacobi}} Equation},
  Author	= {Brydges, D. C. and Kennedy, T.},
  Year		= {1987/07/01, 1987},
  JournalTitle	= {Journal of Statistical Physics},
  Volume	= {48},
  Number	= {1},
  Pages		= {19--49},
  DOI		= {10.1007/BF01010398},
  URL		= {https://doi.org/10.1007/BF01010398},
  id		= {Brydges1987},
  ISBN		= {1572-9613}
}

@Misc{		  carbone2014,
  Title		= {Logarithmic {{Sobolev}} Inequalities and Exponential
		  Entropy Decay in Non-Commutative Algebras},
  Author	= {Carbone, Raffaella},
  Date		= {2014},
  EPrint	= {1402.6948},
  EPrintType	= {arXiv},
  EPrintClass	= {math.OA},
  URL		= {https://arxiv.org/abs/1402.6948}
}

@Article{	  carbone2015,
  Title		= {Logarithmic {{Sobolev}} Inequalities in Non-Commutative
		  Algebras},
  Author	= {Carbone, Raffaella and Martinelli, Andrea},
  Date		= {2015},
  JournalTitle	= {Infinite Dimensional Analysis, Quantum Probability and
		  Related Topics},
  Volume	= {18},
  Number	= {02},
  EPrint	= {https://doi.org/10.1142/S0219025715500113},
  Pages		= {1550011},
  DOI		= {10.1142/S0219025715500113},
  URL		= {https://doi.org/10.1142/S0219025715500113}
}

@Unpublished{	  chialastri2026,
  Title		= {Bounds on Relative Modular {{Hamiltonians}} in General
		  {{QFT}}},
  Author	= {Chialastri, Adriano and Minz, Christoph and Sanders, Ko},
  Date		= {2026-05},
  EPrint	= {2605.27198},
  EPrintType	= {arXiv},
  EPrintClass	= {math-ph}
}

@Article{	  cipriani1997,
  Title		= {Dirichlet Forms and Markovian Semigroups on Standard Forms
		  of von Neumann Algebras},
  Author	= {Cipriani, Fabio},
  Date		= {1997},
  JournalTitle	= {Journal of Functional Analysis},
  Volume	= {147},
  Number	= {2},
  Pages		= {259--300},
  ISSN		= {0022-1236},
  DOI		= {10.1006/jfan.1996.3063},
  URL		= {https://www.sciencedirect.com/science/article/pii/S0022123696930633}
}

@Article{	  cipriani2003,
  Title		= {Derivations as Square Roots of {{Dirichlet}} Forms},
  Author	= {Cipriani, Fabio and Sauvageot, Jean-Luc},
  Date		= {2003},
  JournalTitle	= {Journal of Functional Analysis},
  Volume	= {201},
  Number	= {1},
  Pages		= {78--120},
  ISSN		= {0022-1236},
  DOI		= {10.1016/S0022-1236(03)00085-5},
  URL		= {https://www.sciencedirect.com/science/article/pii/S0022123603000855}
}

@Article{	  derezinski2026,
  Title		= {Miniatures on Open Quantum Systems},
  Author	= {Dereziński, Jan and Jakšić, Vojkan and Pillet,
		  Claude-Alain},
  Date		= {2026-06},
  JournalTitle	= {Complex Analysis and Operator Theory},
  Volume	= {20},
  Number	= {5},
  Publisher	= {{Springer Science and Business Media LLC}},
  ISSN		= {1661-8262},
  DOI		= {10.1007/s11785-026-01967-9},
  URL		= {http://dx.doi.org/10.1007/s11785-026-01967-9}
}

@Article{	  derezinskijaksicpillet2003,
  Title		= {Perturbation Theory of {{W}}*-Dynamics, Liouvilleans and
		  {{KMS-states}}},
  Author	= {Dereziński, J. and Jakšić, V. and Pillet, C.-A.},
  Date		= {2003},
  JournalTitle	= {Reviews in Mathematical Physics},
  Volume	= {15},
  Number	= {05},
  EPrint	= {https://doi.org/10.1142/S0129055X03001679},
  Pages		= {447--489},
  DOI		= {10.1142/S0129055X03001679},
  URL		= {https://doi.org/10.1142/S0129055X03001679}
}

@Article{	  duch2022,
  Title		= {Renormalization of Singular Elliptic Stochastic {{PDEs}}
		  Using Flow Equation},
  Author	= {Duch, Paweł},
  Date		= {2025-02},
  JournalTitle	= {Probability and Mathematical Physics},
  Volume	= {6},
  Number	= {1},
  Pages		= {111--138},
  Publisher	= {Mathematical Sciences Publishers},
  ISSN		= {2690-0998},
  DOI		= {10.2140/pmp.2025.6.111},
  URL		= {http://dx.doi.org/10.2140/pmp.2025.6.111}
}

@Article{	  duch2025,
  Title		= {Flow Equation Approach to Singular Stochastic {{PDEs}}},
  Author	= {Duch, Paweł},
  Date		= {2025-03},
  JournalTitle	= {Probability and Mathematical Physics},
  Volume	= {6},
  Number	= {2},
  Pages		= {327--437},
  Publisher	= {Mathematical Sciences Publishers},
  ISSN		= {2690-0998},
  DOI		= {10.2140/pmp.2025.6.327},
  URL		= {http://dx.doi.org/10.2140/pmp.2025.6.327}
}

@online{	  frob2025,
  Title		= {Bounding Relative Entropy for Non-Unitary Excitations in
		  Quantum Field Theory},
  Author	= {Fröb, Markus B. and Sangaletti, Leonardo},
  Date		= {2026-04},
  EPrint	= {2604.18383},
  EPrintType	= {arXiv},
  EPrintClass	= {math-ph},
  PubState	= {prepublished}
}

@Article{	  frohlich1980,
  Title		= {Unbounded, Symmetric Semigroups on a Separable {{Hilbert}}
		  Space Are Essentially Selfadjoint},
  Author	= {Fröhlich, J},
  Date		= {1980},
  JournalTitle	= {Advances in Applied Mathematics},
  Volume	= {1},
  Number	= {3},
  Pages		= {237--256},
  ISSN		= {0196-8858},
  DOI		= {10.1016/0196-8858(80)90012-3},
  URL		= {https://www.sciencedirect.com/science/article/pii/0196885880900123}
}

@online{	  galanda2025,
  Title		= {Equilibrium States for Non Relativistic {{Bose}} Gases
		  with Condensation},
  Author	= {Galanda, Stefano and Pinamonti, Nicola},
  Date		= {2025-09-29},
  EPrint	= {2509.25101},
  EPrintType	= {arXiv},
  EPrintClass	= {math-ph},
  DOI		= {10.48550/arXiv.2509.25101},
  URL		= {http://arxiv.org/abs/2509.25101},
  PubState	= {prepublished}
}

@online{	  galanda2026,
  Title		= {Equilibrium States for Non Relativistic {{Bose}} Gases and
		  the {{Gross-Pitaevskii}} Limit},
  Author	= {Galanda, Stefano and Pinamonti, Nicola},
  Date		= {2026-07-07},
  EPrint	= {2606.25775},
  EPrintType	= {arXiv},
  EPrintClass	= {math-ph},
  DOI		= {10.48550/arXiv.2606.25775},
  URL		= {http://arxiv.org/abs/2606.25775},
  PubState	= {prepublished}
}

@Article{	  gerard2005,
  Title		= {Thermal Quantum Fields with Spatially Cutoff Interactions
		  in 1+1 Space–Time Dimensions},
  Author	= {Gérard, Christian and Jäkel, Christian D.},
  Date		= {2005},
  JournalTitle	= {Journal of Functional Analysis},
  Volume	= {220},
  Number	= {1},
  Pages		= {157--213},
  ISSN		= {0022-1236},
  DOI		= {10.1016/j.jfa.2004.08.003},
  URL		= {https://www.sciencedirect.com/science/article/pii/S0022123604003106}
}

@Article{	  gerard2005b,
  Title		= {Thermal Quantum Fields without Cut-Offs in 1+1 Space-Time
		  Dimensions},
  Author	= {Gérard, Christian and Jäkel, Christian D.},
  Date		= {2005},
  JournalTitle	= {Reviews in Mathematical Physics},
  Volume	= {17},
  Number	= {02},
  EPrint	= {https://doi.org/10.1142/S0129055X05002303},
  Pages		= {113--173},
  DOI		= {10.1142/S0129055X05002303},
  URL		= {https://doi.org/10.1142/S0129055X05002303}
}

@Article{	  gielerak1998,
  Title		= {Stochastically Positive Structures on {{Weyl}} Algebras.
		  {{The}} Case of Quasi-Free States},
  Author	= {Gielerak, R. and Jakóbczyk, L. and Olkiewicz, R.},
  Date		= {1998-12},
  JournalTitle	= {Journal of Mathematical Physics},
  Volume	= {39},
  Number	= {12},
  Pages		= {6291--6328},
  Publisher	= {AIP Publishing},
  ISSN		= {1089-7658},
  DOI		= {10.1063/1.532639},
  URL		= {http://dx.doi.org/10.1063/1.532639}
}

@Book{		  glimmjaffe1987,
  Title		= {Quantum Physics: A Functional Integral Point of View},
  Author	= {Glimm, James and Jaffe, Arthur},
  Date		= {1987},
  Edition	= {2},
  Publisher	= {Springer},
  Location	= {New York},
  DOI		= {10.1007/978-1-4612-4728-9},
  ISBN		= {978-1-4612-4728-9}
}

@InCollection{	  gross1967,
  Title		= {Abstract Wiener Spaces},
  BookTitle	= {Proceedings of the Fifth Berkeley Symposium on
		  Mathematical Statistics and Probability},
  Author	= {Gross, Leonard},
  Editor	= {Le Cam, Lucien M. and Neyman, Jerzy},
  Date		= {1967},
  Volume	= {2},
  Pages		= {31--42},
  Publisher	= {University of California Press},
  Location	= {Berkeley, California},
  URL		= {https://projecteuclid.org/euclid.bsmsp/1200513262},
  Part		= {1}
}

@Article{	  gross1967potential,
  Title		= {Potential Theory on {{Hilbert}} Space},
  Author	= {Gross, Leonard},
  Date		= {1967},
  JournalTitle	= {Journal of Functional Analysis},
  Volume	= {1},
  Number	= {2},
  Pages		= {123--181},
  ISSN		= {0022-1236},
  DOI		= {10.1016/0022-1236(67)90030-4},
  URL		= {https://www.sciencedirect.com/science/article/pii/0022123667900304}
}

@Article{	  gross1975,
  Title		= {Logarithmic Sobolev Inequalities},
  Author	= {Gross, Leonard},
  Date		= {1975},
  JournalTitle	= {American Journal of Mathematics},
  Volume	= {97},
  Number	= {4},
  EPrint	= {2373688},
  EPrintType	= {jstor},
  Pages		= {1061--1083},
  Publisher	= {Johns Hopkins University Press},
  ISSN		= {00029327, 10806377},
  URL		= {http://www.jstor.org/stable/2373688},
  URLDate	= {2026-08-29}
}

@Misc{		  hairer2026,
  Title		= {Advanced Stochastic Calculus},
  Author	= {Hairer, Martin},
  Date		= {2026},
  URL		= {https://hairer.org/notes/StochasticAnalysisCourse.pdf},
  Note		= {Lecture notes, EPFL and Imperial College London}
}

@online{	  hollands2025,
  Title		= {Bekenstein Bound for Approximately Local Charged States},
  Author	= {Hollands, Stefan and Longo, Roberto},
  Date		= {2025-01},
  EPrint	= {2501.03849},
  EPrintType	= {arXiv},
  EPrintClass	= {hep-th},
  PubState	= {prepublished}
}

@online{	  hollands2026,
  Title		= {Bekenstein's Bound for Wave Packets},
  Author	= {Hollands, Stefan and Longo, Roberto and Morsella,
		  Gerardo},
  Date		= {2026-02},
  EPrint	= {2602.03606},
  EPrintType	= {arXiv},
  EPrintClass	= {math-ph},
  PubState	= {prepublished}
}

@Article{	  keller1992,
  Author	= {Keller, G. and Kopper, Christoph and Salmhofer, M.},
  Date		= {1992},
  JournalTitle	= {Helvetica Physica Acta},
  ShortJournal	= {Helv. Phys. Acta},
  Volume	= {65},
  Number	= {MPI-PAE-PTH-65-90},
  Pages		= {32--52},
  Title		= {Perturbative Renormalization and Effective {{Lagrangians}}
		  in {$\phi^4$} in Four-Dimensions}
}

@Article{	  keller1993,
  Author	= {Keller, Georg and Kopper, Christoph},
  Date		= {1994},
  JournalTitle	= {Communications in Mathematical Physics},
  ShortJournal	= {Commun. Math. Phys.},
  Volume	= {161},
  Number	= {UNIGOE-THPHY-4-93},
  Pages		= {515--532},
  DOI		= {10.1007/BF02101931},
  Title		= {Perturbative Renormalization of Massless {$\phi^4$} in
		  Four-Dimensions with Flow Equations}
}

@Article{	  klein1978,
  Title		= {The Semigroup Characterization of {{Osterwalder-Schrader}}
		  Path Spaces and the Construction of {{Euclidean}} Fields},
  Author	= {Klein, Abel},
  Date		= {1978},
  JournalTitle	= {Journal of Functional Analysis},
  Volume	= {27},
  Number	= {3},
  Pages		= {277--291},
  ISSN		= {0022-1236},
  DOI		= {10.1016/0022-1236(78)90009-5},
  URL		= {https://www.sciencedirect.com/science/article/pii/0022123678900095}
}

@article{Duch2026,
  author = {Duch, Pawe{\l} and Gubinelli, Massimiliano and Rinaldi, Paolo},
  year = 2026,
  month = oct,
  journal = {Journal of the European Mathematical Society},
  issn = {1435-9855},
  doi = {10.4171/jems/1816},
  url = {https://ems.press/journals/jems/articles/14300101},
  urldate = {2026-11-30},
  langid = {english},
  title = {Parabolic Stochastic Quantisation of the Fractional {$\Phi^{4}_{3}$} Model in the Full Subcritical Regime}
}

@Article{	  klein1981,
  Title		= {Stochastic Processes Associated with {{KMS}} States},
  Author	= {Klein, Abel and Landau, Lawrence J},
  Date		= {1981},
  JournalTitle	= {Journal of Functional Analysis},
  Volume	= {42},
  Number	= {3},
  Pages		= {368--428},
  ISSN		= {0022-1236},
  DOI		= {10.1016/0022-1236(81)90096-3},
  URL		= {https://www.sciencedirect.com/science/article/pii/0022123681900963}
}

@Article{	  klein1981b,
  Title		= {Construction of a Unique Self-Adjoint Generator for a
		  Symmetric Local Semigroup},
  Author	= {Klein, Abel and Landau, Lawrence J},
  Date		= {1981},
  JournalTitle	= {Journal of Functional Analysis},
  Volume	= {44},
  Number	= {2},
  Pages		= {121--137},
  ISSN		= {0022-1236},
  DOI		= {10.1016/0022-1236(81)90007-0},
  URL		= {https://www.sciencedirect.com/science/article/pii/0022123681900070}
}

@InCollection{	  ledoux2004,
  Title		= {Logarithmic {{Sobolev}} Inequalities for Unbounded Spin
		  Systems Revisited},
  BookTitle	= {Séminaire de Probabilités {{XXXV}}},
  Author	= {Ledoux, Michel},
  Date		= {2004},
  Pages		= {167--194},
  Publisher	= {Springer}
}

@Article{	  longo2024,
  Title		= {A Bekenstein-Type Bound in {{QFT}}},
  Author	= {Longo, Roberto},
  Date		= {2025},
  JournalTitle	= {Communications in Mathematical Physics},
  ShortJournal	= {Commun. Math. Phys.},
  Volume	= {406},
  Number	= {5},
  EPrint	= {2409.14408},
  EPrintType	= {arXiv},
  EPrintClass	= {math-ph},
  Pages		= {95},
  DOI		= {10.1007/s00220-025-05261-1}
}

@online{	  longomorinelli2024,
  Title		= {An Entropy Bound Due to Symmetries},
  Author	= {Longo, Roberto and Morinelli, Vincenzo},
  Date		= {2024-01},
  EPrint	= {2401.02345},
  EPrintType	= {arXiv},
  EPrintClass	= {math.OA},
  PubState	= {prepublished}
}

@Article{	  longoxu2018,
  Title		= {Comment on the Bekenstein Bound},
  Author	= {Longo, Roberto and Xu, Feng},
  Date		= {2018},
  JournalTitle	= {Journal of Geometry and Physics},
  ShortJournal	= {J. Geom. Phys.},
  Volume	= {130},
  EPrint	= {1802.07184},
  EPrintType	= {arXiv},
  EPrintClass	= {math-ph},
  Pages		= {113--120},
  DOI		= {10.1016/j.geomphys.2018.03.004}
}

@InProceedings{	  nelson1966,
  Title		= {A Quartic Interaction in Two Dimensions},
  BookTitle	= {Mathematical Theory of Elementary Particles, Proc.
		  {{Conf}}., Dedham, Mass., 1965},
  Author	= {Nelson, Edward},
  Date		= {1966},
  Pages		= {69--73},
  Publisher	= {MIT press}
}

@Article{	  nelson1973,
  Title		= {Construction of Quantum Fields from {{Markoff}} Fields},
  Author	= {Nelson, Edward},
  Date		= {1973},
  JournalTitle	= {Journal of Functional Analysis},
  Volume	= {12},
  Number	= {1},
  Pages		= {97--112},
  ISSN		= {0022-1236},
  DOI		= {10.1016/0022-1236(73)90091-8},
  URL		= {https://www.sciencedirect.com/science/article/pii/0022123673900918}
}

@Article{	  nelson1973a,
  Title		= {The Free {{Markoff}} Field},
  Author	= {Nelson, Edward},
  Date		= {1973},
  JournalTitle	= {Journal of Functional Analysis},
  Volume	= {12},
  Number	= {2},
  Pages		= {211--227},
  ISSN		= {0022-1236},
  DOI		= {10.1016/0022-1236(73)90025-6},
  URL		= {https://www.sciencedirect.com/science/article/pii/0022123673900256}
}

@Article{	  olkiewicz1999,
  Author	= {Olkiewicz, Robert and Zegarlinski, Bogusław},
  Date		= {1999},
  JournalTitle	= {Journal of Functional Analysis},
  Volume	= {161},
  Number	= {1},
  Pages		= {246--285},
  ISSN		= {0022-1236},
  DOI		= {10.1006/jfan.1998.3342},
  URL		= {https://www.sciencedirect.com/science/article/pii/S0022123698933420},
  Title		= {Hypercontractivity in Noncommutative {$L_p$} {{Spaces}}}
}

@Article{	  polchinski1983,
  Title		= {Renormalization and Effective Lagrangians},
  Author	= {Polchinski, Joseph},
  Date		= {1984},
  JournalTitle	= {Nuclear Physics B},
  ShortJournal	= {Nucl. Phys. B},
  Volume	= {231},
  Number	= {HUTP-83-A018},
  Pages		= {269--295},
  DOI		= {10.1016/0550-3213(84)90287-6},
  URL		= {https://doi.org/10.1016/0550-3213(84)90287-6}
}

@Article{	  reed1974,
  Title		= {Support Properties of the Free Measure for {{Boson}}
		  Fields},
  Author	= {Reed, M. and Rosen, L.},
  Year		= {1974/06/01, 1974},
  JournalTitle	= {Communications in Mathematical Physics},
  Volume	= {36},
  Number	= {2},
  Pages		= {123--132},
  DOI		= {10.1007/BF01646326},
  URL		= {https://doi.org/10.1007/BF01646326},
  id		= {Reed1974},
  ISBN		= {1432-0916}
}

@Book{		  salmhofer1999,
  Title		= {Renormalization: {{An}} Introduction},
  Author	= {Salmhofer, M.},
  Date		= {1999},
  Series	= {Theoretical and Mathematical Physics},
  Publisher	= {Springer Berlin},
  Location	= {Heidelberg},
  DOI		= {10.1007/978-3-662-03873-4},
  ISBN		= {978-3-540-64666-2}
}

@Article{	  schwinger1958,
  Title		= {On the Euclidean Structure of Relativistic Field Theory},
  Author	= {Schwinger, Julian},
  Date		= {1958},
  JournalTitle	= {Proceedings of the National Academy of Sciences of the
		  United States of America},
  Volume	= {44},
  Number	= {9},
  EPrint	= {90042},
  EPrintType	= {jstor},
  Pages		= {956--965},
  Publisher	= {National Academy of Sciences},
  ISSN		= {00278424, 10916490},
  URL		= {http://www.jstor.org/stable/90042},
  URLDate	= {2026-08-29}
}

@Article{	  sheffield2007,
  Title		= {Gaussian Free Fields for Mathematicians},
  Author	= {Sheffield, Scott},
  Date		= {2007},
  JournalTitle	= {Probability theory and related fields},
  Volume	= {139},
  Number	= {3},
  Pages		= {521--541},
  Publisher	= {Springer}
}

@Article{	  symanzik1966,
  Title		= {Euclidean Quantum Field Theory. {{I}}. {{Equations}} for a
		  Scalar Model},
  Author	= {Symanzik, K.},
  Date		= {1966},
  JournalTitle	= {Journal of Mathematical Physics},
  ShortJournal	= {J. Math. Phys.},
  Volume	= {7},
  Number	= {3},
  Pages		= {510},
  DOI		= {10.1063/1.1704960}
}

@Book{		  takesaki1970,
  Title		= {Tomita's Theory of Modular Hilbert Algebras and Its
		  Applications},
  Author	= {Takesaki, M.},
  Date		= {1970},
  Series	= {Lecture Notes in Mathematics},
  Publisher	= {Springer-Verlag},
  DOI		= {10.1007/bfb0065832}
}

@Article{	  wilson73,
  Title		= {The {{Renormalization}} Group and the Epsilon Expansion},
  Author	= {Wilson, K. G. and Kogut, John B.},
  Date		= {1974},
  JournalTitle	= {Phys. Rept.},
  Volume	= {12},
  Pages		= {75--199},
  DOI		= {10.1016/0370-1573(74)90023-4},
  URL		= {https://doi.org/10.1016/0370-1573(74)90023-4}
}

@Article{	  ziebell2023,
  Title		= {A {{Rigorous Derivation}} of the {{Functional
		  Renormalisation Group Equation}}},
  Author	= {Ziebell, Jobst},
  Date		= {2023-11},
  JournalTitle	= {Communications in Mathematical Physics},
  ShortJournal	= {Commun. Math. Phys.},
  Volume	= {403},
  Number	= {3},
  EPrint	= {2106.09466},
  EPrintType	= {arXiv},
  EPrintClass	= {math-ph},
  Pages		= {1329--1361},
  ISSN		= {0010-3616, 1432-0916},
  DOI		= {10.1007/s00220-023-04821-7},
  URL		= {http://arxiv.org/abs/2106.09466},
  URLDate	= {2026-09-01}
}

\end{document}